%% file: ms.tex
\documentclass[sigconf, nonacm]{acmart}
\pdfoutput=1
\newcommand\vldbdoi{XX.XX/XXX.XX}

\newcommand\vldbpagestyle{plain}

\usepackage{color}
\usepackage{amsmath, amsthm, amsfonts}
\usepackage[normalem]{ulem}
\usepackage{graphicx}
\usepackage{hyperref}
\usepackage{xcolor}
\usepackage{subfig}
\usepackage[ruled,linesnumbered]{algorithm2e}
\usepackage{setspace}
\usepackage{pifont}
\usepackage{tikz}

\usetikzlibrary{arrows.meta, positioning}
\usetikzlibrary{arrows, automata}
\usetikzlibrary{positioning,shapes,shadows,arrows,decorations.markings}
\usetikzlibrary{decorations.pathmorphing}
\usetikzlibrary{decorations.pathreplacing,calligraphy}
\usetikzlibrary{decorations.text}
\SetKwInput{KwInput}{Input}
\SetKwInput{KwOutput}{Output}
\SetKwComment{Comment}{/* }{ */}

\newcommand\mysection[1]{\vspace{-1.6mm}\section{#1}\vspace{-0mm}}
\newcommand\mysubsection[1]{\vspace{-2.2mm}\subsection{#1}\vspace{-.3mm}}
\newcommand\mysubsubsection[1]{\vspace{-2mm}\subsubsection{#1}\vspace{0mm}}

\newcommand{\extVersion}{false}

\newcommand{\printIfExtVersion}[2]%
{%
        \ifthenelse{\equal{\extVersion}{true}}{#1}{}%
        \ifthenelse{\equal{\extVersion}{false}}{#2}{}%
}
      
\newtheorem{theorem}{Theorem}[section]

\newtheorem{lemma}[theorem]{Lemma}
\theoremstyle{definition}
\newtheorem{definition}{Definition}[section]
\theoremstyle{remark}

\newtheorem{property}{Property}

\usepackage{pifont}
\newcommand{\Grow}{\textsc{Grow}\xspace}
\newcommand{\Merge}{\textsc{Merge}\xspace}
\newcommand{\Init}{\textsc{Init}\xspace}
\newcommand{\Mo}{\textsc{Mo}\xspace}

\newcommand{\titledparagraph}[1]{\noindent\textbf{#1}}

\newcommand{\MM}[1]{\textcolor{blue}{\textbf{MM:} #1 }}

\newcommand{\graph}{\ensuremath\mathbf{G}}
\newcommand{\nodes}{\ensuremath\mathbf{N}}
\newcommand{\edges}{\ensuremath\mathbf{E}}
\newcommand{\labels}{\ensuremath\mathbf{L}}
\newcommand{\props}{\ensuremath\mathcal{P}}
\newcommand{\vars}{\ensuremath\mathcal{V}}
\newcommand{\query}{\mbox{:- }}
\newcommand{\bfs}{\textsc{BFT}}
\newcommand{\bfsm}{\textsc{BFT-M}}
\newcommand{\bfsam}{\textsc{BFT-AM}}
\newcommand{\esp}{\textsc{ESP}}
\newcommand{\lesp}{\textsc{LESP}}
\newcommand{\moesp}{\textsc{MoESP}}
\newcommand{\molesp}{\textsc{MoLESP}}
\newcommand{\TO}{\texttt{TimeOut}}

\begin{document}
\title{Integrating connection search in graph queries}

\author{Angelos Christos Anadiotis}
\authornote{Work done while at Ecole Polytechnique.}
\affiliation{%
  \institution{Oracle, Switzerland}
}
\email{angelos.anadiotis@oracle.com}

\author{Ioana Manolescu}
\affiliation{%
  \institution{Inria and IPP, France}
}
\email{ioana.manolescu@inria.fr}

\author{Madhulika Mohanty}
\affiliation{%
  \institution{Inria and IPP, France}
}
\email{madhulika.mohanty@inria.fr}

\begin{abstract}
Graph data management and querying has many practical applications. 
When graphs are very heterogeneous and/or users are unfamiliar with their structure, they may need to {\em find how two or more groups of nodes are connected in a graph}, even when users are not able to describe the connections. This is only partially supported by existing query languages, which allow searching for {\em paths}, but not for {\em trees connecting three or more node groups}. The latter is related to the NP-hard Group Steiner Tree problem, and has been previously considered for keyword search in databases.

In this work, we formally show how to integrate {\em connecting tree patterns} (CTPs, in short) within a graph query language such as SPARQL or Cypher, leading to an {\em Extended Query Language} (or EQL, in short). 
We then study a set of algorithms for evaluating CTPs; we generalize prior keyword search work, most importantly by ($i$)~considering bidirectional edge traversal and ($ii$)~allowing users to select {\em any} score function for ranking CTP results.
To cope with very large search spaces, we propose an efficient pruning technique and formally establish a large set of cases where our algorithm, \molesp, is complete even with pruning. Our experiments validate the performance of our CTP and EQL evaluation algorithms on a large set of synthetic and real-world workloads. 
\end{abstract}

\maketitle

\pagestyle{\vldbpagestyle}

\input{introduction}

\input{language-construct}

\input{baseline-algo}

\input{moesp}

\input{smesp}

\input{evaluation1}

\input{evaluation2}

\input{related-work}

\balance
\bibliographystyle{ACM-Reference-Format}
\bibliography{main.bib}

\end{document}

%% file: introduction.tex
\vspace{1.5mm} 
\mysection{Introduction}\label{sec:introduction}

Graph databases are increasingly adopted in a wide range of applications spanning from social network analysis to scientific data exploration, the financial industry, and many more.
To query RDF graphs, one can use the W3C's standard SPARQL~\cite{sparql11}  query language; for property graphs,  Cypher~\cite{neo4j-cypher} is among the best known. 
An interesting but challenging query language feature is  {\em  reachability}:  a SPARQL 1.1 query can \underline{check}, e.g., if there are some paths along which Mr. Shady deposits funds into a given bank ABC. Such queries are important in investigative journalism applications~\cite{anadiotis:hal-03337650}, in the fight against money laundering, etc. SPARQL allows checking for the existence of a path, but does not return the matching paths to users. In contrast, a Cypher query may also \underline{return} the paths between two given sets of nodes.

Unfortunately, none of these languages support \underline{finding trees}, connecting three (or more) sets of nodes, while the latter can be very useful. For instance, when investigating ill-acquired wealth, one may want to find ``all connections between Mr. Shady, bank company ABC, and the tax office of the DEF republic'': an answer to this query is a {\em tree}, connecting three nodes corresponding to the person, bank, and tax office, respectively. 

Searching for connections among $m$ sets of nodes is closely related to the Group Steiner Tree Problem (GSTP), which asks for {\em the least-cost}, e.g., fewest-edges, tree; 
the problem is NP-hard. The database literature has studied many variants of this problem under the name of {\em keyword search in databases}, for e.g., \cite{banks-1, banks-1-demo,banks-2, ease, dpbf, DBXplorer, kssurvey, coffman@tkde2014, qgstp@www21, lancet@vldb2021}. 
To cope with the high complexity, existing algorithms ($i$)~consider a fixed cost function and leverage its properties to limit the search, ($ii$)~propose approximate solutions, within a known distance from the optimum, and/or ($iii$)~implement heuristics without guarantees but which have performed well on some problems.

\noindent\textbf{Requirements} Our recent collaborations with investigative journalists~\cite{anadiotis:hal-03337650,gam-inf-sys-2022} lead to identifying the following set of needs.
First, (\textbf{R1}) {\em  graph query languages should allow returning trees that connect $m$ node sets}, for some integer $m\geq 2$; 
(\textbf{R2}) it must be possible to search for connecting trees  {\em orthogonally to (or, in conjunction with any) score functions} used to compare and rank the trees. This is because different graphs and applications are best served by different scores, and when exploring a graph, journalists need to experiment with several before they find interesting patterns. 
For instance, in the example above, if Mr. Shady is a citizen of DEF and ABC has offices there, the smallest solution connects them through the DEF country node; however, this is not interesting to journalists. Instead, a connection through three ABC accounts, sending money from DEF to Mr. Shady in country GHI, is likely much more interesting.  
An orthogonal requirement is (\textbf{R3}) to {\em treat graphs as undirected when searching for trees}. For instance, the graph may contain ``Mr. Shady $\xrightarrow{\text{hasAccount}}$ acct$_1$'', or, just as likely, ``acct$_1$ $\xrightarrow{\text{belongsTo}}$ Mr. Shady''. We cannot afford to miss a connecting tree because we ``expected'' an edge in a direction and it happens to be in the opposite direction.  
Further,  (\textbf{R4}) {\em \underline{all} answers need to be found (within a time and/or space budget)} for several reasons: ($i$)~continuity with the semantics of standard graph query languages, that also return all results (unless users explicitly \textsf{\small LIMIT} the result size); ($ii$)~to remain independent of, and thus orthogonal to, the cost function (recall (R2)); and,  ($iii$)~for practical reasons, given the problem complexity, which is further exacerbated by (R3), and renders complete search on large graphs unfeasible.
Finally,  (\textbf{R5}) {\em the extended queries should be efficiently executed}, even when graphs are {\em highly heterogeneous}, as in investigative journalism scenarios, where text, structured, and/or semistructured sources are integrated together. 

\titledparagraph{Contributions} 
To address the above requirements, we make the following contributions: 

\indent\textbf{(1)} We {\em formally define an Extended Query Language (EQL, in short)}, which combines together Basic Graph Pattern (or conjunctive) queries at the core of both SPARQL and Cypher, and Connecting Tree Patterns (CTPs, in short). A CTP allows searching for trees that connect $m$ groups of nodes, for $m\geq 2$. BGPs and CTPs can be freely joined. 
This addresses requirements (R1), (R2), and also (R3), since our CTP semantics returns trees regardless of the edge directions (Section \ref{sec:language-construct}).

\indent\textbf{(2)} We provide a {\em scalable EQL query evaluation strategy}, which leverages existing algorithms for the well-studied problem of evaluating conjunctive queries, contributing to (R5) (Section~\ref{sec:query-evaluation}).

\indent\textbf{(3)} For {\em CTP evaluation}, we study a set of baseline algorithms, and explain that their performance suffers due to repeated (wasted) work and/or the need to minimize the trees they find; GAM~\cite{gam-inf-sys-2022} algorithm is more efficient, but it does not scale in all cases. We introduce a powerful {\em Edge Set Pruning (ESP) technique}, which significantly speeds up the execution, but can lead to incompleteness. We then bring two orthogonal modifications which, combined, lead to our \molesp\ algorithm, for which we {\em formally establish completeness for $m\in \{2,3\}$}, which are most frequent, as well as for {\em a large class of results for arbitrarily large $m$}. This addresses requirement (R4) and contributes to (R5) (Section~\ref{sec:ctp-evaluation}). 

\indent\textbf{(4)} We experimentally show that: ($i$)~baseline algorithms inspired from breadth-first search are unfeasible even for small graphs; ($ii$)~the optimizations we bring here over the GAM algorithm~\cite{gam-inf-sys-2022} strongly reduce the search time; ($iii$)~integrating our \molesp\ algorithm with a simple conjunctive graph query engine allows to efficiently evaluate queries in our extended language (Section \ref{sec:experiments}).



%% file: language-construct.tex
\mysection{Extended Query Language (EQL)}\label{sec:language-construct}

\definition[Graph] 
A  graph $\graph(\nodes, \edges)$ consists of a set of nodes $\nodes$ and a set of edges $\edges \subseteq \nodes \times \nodes$. Each node  $n\in \nodes$ carries a label $l(n)$ from a label set $\labels$, which includes the empty label $\epsilon$. Similarly, each edge $e\in \edges$ has a label $l(e)\in \labels$. 

The two main graph data models are RDF graphs, and property graphs (PGs). To illustrate, in the following, we will rely on RDF graphs; our work can be transposed with only surface changes to PGs. 
Figure~\ref{fig:eg-graph} introduces a sample graph, assigning an integer ID and label to each node and edge. We will refer to nodes as $n_1,n_2$, etc., e.g., $n_1$ is the node whose ID is 1 and label is OrgB, and similarly to edges as $e_1,e_2$, etc. Labels of literal nodes, e.g.,  $n_{11}$, are enclosed in quotes; the other nodes are URIs. 

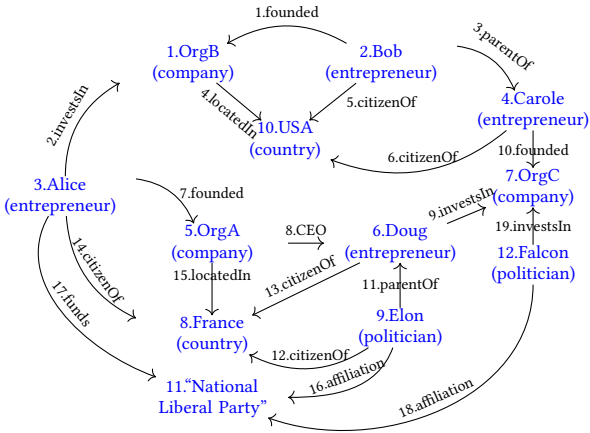
\begin{figure}[t!]
\hspace{-2mm}
\tikzstyle{node} = [ text width=2cm, text centered,inner sep=0.1pt, fill=white,font = {\footnotesize},text=blue]
\tikzstyle{arrow} = [thick,font = {\scriptsize}]

\begin{tikzpicture}[node distance=2.5cm]
\node (n1) [node] {1.OrgB \\ (company)};
\node (n2) [node,right of=n1] {2.Bob \\ (entrepreneur)};
\node (n3) [node, below left of=n1] {3.Alice \\ (entrepreneur)};
\node (n4) [node, below right=0.1cm and 0cm of n2] {4.Carole \\ (entrepreneur)};
\node (n5) [node, below right=0.1cm and 0cm of n3] {5.OrgA \\ (company)};
\node (n6) [node, right of=n5] {6.Doug \\ (entrepreneur)};
\node (n7) [node, below=0.5cm of n4] {7.OrgC \\ (company)};
\node (n8) [node, below=0.7cm of n5] {8.France \\ (country)};
\node (n9) [node, below=0.6cm of n6] {9.Elon \\ (politician)};
\node (n10) [node, below right=0.5cm and -0.8cm of n1] {10.USA \\ (country)};
\node (n11) [node, below=0.3cm of n8] {11.``National Liberal Party''};
\node (n12) [node, below=0.5cm of n7] {12.Falcon \\ (politician)};

\draw[<-] (n1) to [bend left] node [above] {\scriptsize 1.founded} (n2);
\draw[->]  (n3) to [bend left] node [sloped,above] {\scriptsize 2.investsIn} (n1);

\draw[->]  (n2) to [bend left] node [sloped,above] {\scriptsize 3.parentOf} (n4) ;
\draw[->]  (n1) -> node [sloped,below] {\scriptsize 4.locatedIn} (n10);
\draw[->]  (n2) -> (n10) node [midway,right,xshift=2pt] {\scriptsize 5.citizenOf};
\draw[->]  (n4) to [bend left] node [above] {\scriptsize 6.citizenOf} (n10) ;
\draw[->] (n3) to [bend left] node [midway,right] {\scriptsize 7.founded} (n5);

\draw[->]  (n5) -> (n6) node [midway,above] {\scriptsize 8.CEO};
\draw[->]  (n6) ->  node [sloped,above] {\scriptsize 9.investsIn} (n7);
\draw[->]  (n4) -> (n7) node [midway] {\scriptsize 10.founded};
\draw[<-]  (n6) -> (n9) node [midway,centered] {\scriptsize 11.parentOf};

\draw[->] (n9) to [bend left] node [midway,above] {\scriptsize 12.citizenOf} (n8);

\draw[->]  (n6) ->  node [sloped,above] {\scriptsize 13.citizenOf} (n8);

\draw[->] (n3) to [bend right] node [sloped,above] {\scriptsize 14.citizenOf} (n8);

\draw[->]  (n5) -> (n8) node [midway,above] {\scriptsize 15.locatedIn};

\draw[->]  (n9) to [out=250,in=0] node [sloped, above] {\scriptsize 16.affiliation} (n11) ;
\draw[->]  (n3) to [out=240,in=160] node [sloped,above] {\scriptsize 17.funds} (n11) ;
\draw[->]  (n12) to [out=270,in=340] node [sloped,above] {\scriptsize 18.affiliation} (n11) ;
\draw[->]  (n12) -> (n7) node [midway] {\scriptsize 19.investsIn};
\end{tikzpicture}
\vspace{-5.5mm}
\caption{Sample data graph. 
  \label{fig:eg-graph}}
\vspace{-5.5mm}
\end{figure}

\noindent\textbf{Node and edge properties} Graph nodes and edges may have other properties beyond labels; for instance, an RDF node may have $0$ or more {\em types}. In our example, types are shown in parentheses under the nodes. In a PG,  nodes and edges can have multiple properties. We denote by $\props$ the set of all properties that nodes and edges may have; each property $p \in \props$ is a function $p$ that, given a node $n$ (or edge $e$), returns $p(n)$, the value of property $p$ on node $n$ (and similarly for $e$). Without loss of generality, we consider that $l$:$\nodes$$\rightarrow$$\labels$ belongs to $\props$, that is, the label is a node and/or edge property. 


Let $\vars$ be a set of variable names, to be used in queries. 
Let $\Omega=\{=,<,\leq,\sim \}$ be a set of comparison operators, where $\sim$ denotes pattern matching such as SQL's \textsf{\small like} operator.  They are used to express predicates over nodes and/or edges, as follows:

  \definition[Predicate] A {\em condition} over a variable $v\,$$\in$$\,\vars$ is of the form $p(v) \;op \; c$ where $p\in \props$, $op\,$$\in$$\,\Omega$ and $c$ is a constant such that the operator $op$ is well-defined on any value of property $p$ together with $c$. A {\em predicate over $v$} is a conjunction of  conditions over $v$. An empty predicate (no conditions) over $v$ is simply $v$.

A node $n\in \nodes$ (or edge $e\in \edges$) {\em satisfies the predicate} if and only if, in every condition of the predicate, replacing $v$ with $n$ (respectively, $e$) and evaluating $op$ yields true. 
For instance, $l(v)$$\sim$$\text{"*lice"} \wedge \tau(v)$=$\tau_{\mathrm{entrepreneur}}$ is a predicate consisting of two conditions, one on the label (which must end in the string ``lice'') and one on the type, which must be entrepreneur. 
This predicate is true on the node $n_3$ in our example, and false on the other nodes and edges.
Any node or edge satisfies the empty predicate. 
{\em For readability, when a predicate consists of exactly an equality between a node or edge label and a constant, we simply use the constant to denote the predicate}, thus, $l(v)=\text{"Alice"}$ can be simply written "Alice", when this is unambiguous. However, {\em each predicate always involves exactly one variable} ($v$ in our example), {\em even when the short syntax hides it}. We will revert to the longer syntax when we need to make the variable explicit, e.g., use it several times in the query.

\definition[Edge Pattern] An edge pattern is a triple $(p_1, p_2,p_3)$ of three predicates: 
 $p_1$ holds over the source node of an edge, $p_2$ over the edge itself, and $p_3$ over the target node.

 For instance, $(l(s)$$=$$\text{"Alice"},l(e)$$=$$\text{"citizenOf"},d)$
 states that the source node $s$ is labeled "Alice" and the edge $e$ is labeled "citizenOf". The third predicate is a variable. With the above simplification, we can also write this pattern as $(\text{"Alice"},\text{"citizenOf"},d)$. 

 A core construct of graph query languages is:
 
\definition[Basic Graph Pattern]
A Basic Graph Pattern (BGP) $b$ is a set of edge patterns that are {\em connected} in the following sense. If the BGP contains at least $2$ edge patterns, each pattern must have a common variable with another edge pattern.  

A sample BGP $b_1$ is: \{$(x, \text{"citizenOf"}, \text{"USA"}),(x,\text{"founded"},\text{"OrgB"})$\}. 


\definition[CT Pattern]
A connecting tree pattern (CTP, in short) is a tuple of the form: $g = (g_1, g_2, \dots , g_m,\underline{v_{m+1}})$ where each $g_i$, $1$$\leq$$i$$\leq$$m$ is a predicate and $\underline{v_{m+1}}$ is a variable.
All variables occurring in $g_1,\ldots,g_m, \underline{v_{m+1}}$ are pairwise distinct. 

CTPs are used to find connections among nodes, as follows. When replacing each $g_i$ with a graph node,  $v_{m+1}$ is bound to a {\em subtree} of $\graph$, having these nodes as leaves (we formalize this below). To visually distinguish BGPs from CTPs, we always \underline{underline} the last variable of a CTP. 

\definition[Core query]
A core query $Q$ has a {\em head} and a {\em body}. The body is a set of $k$ BGPs, $k\geq 0$,  and $l$ CTPs, $l\geq 0$, such that  $k+l>0$, and each underlined (last) variable from a CTP appears exactly once in $Q$. The head is a subset of the body variables.

An example core query, $Q_1$, consists of $3$ BGPs and a CTP: 
\begin{tabular}{@{}l@{}l@{}l@{}}
&$(x,y,z,\underline{w}) \query$&$( \tau(x)=\tau_{\mathrm{entrepreneur}},\text{"citizenOf"},\text{"USA"})$\\
\fbox{$Q_1$}\;&&$( \tau(y)=\tau_{\mathrm{entrepreneur}},\text{"citizenOf"}, \text{"France"}),$\\
&                  &$( \tau(z)=\tau_{\mathrm{politician}},\text{"citizenOf"}, \text{"France"}),(x,y,z,\underline{w})$\\
\end{tabular}

$Q_1$ asks: ``What are the connections $\underline{w}$ between some American entrepreneur $x$,  some French entrepreneur $y$, and some French politician $z$?''
We denote the CTP of this query by $g^1$.  
%
%
%
To define core query semantics, our first notion is: 

\vspace{-1mm}
\definition[BGP embedding] \label{def:bgp-embedding} Given a BGP $b=\{t_1, \ldots, t_k\}$,  an embedding of $b$ into $\graph$ is a function $\phi$, associating to each variable $v$ in $b$, a node $n\in \nodes$ or an edge $e \in \edges$, such that ($i$)~$\phi(v)$ satisfies all the predicates on $v$ in $b$;  and ($ii$)~for every edge pattern $(s,e,d)$ in $b$, the edge $\phi(e) \in \edges$ goes from $\phi(s)$ to $\phi(d)$. 

A sample embedding $\phi$ for the first BGP of $Q_1$ maps $x$ to $n_4$, 
"USA" to $n_{10}$, 
"citizenOf" to $e_6$, etc. 

  Next, we define:

\vspace{-1mm}
\definition[Set-based CTP result]  \label{def:ctp-result} Let $g=(g_1,\ldots,g_m,\underline{v_{m+1}})$ be a CTP pattern and $S_1,\ldots, S_m$ be sets of $\graph$ nodes, called \textbf{seed sets}, such that every node in $S_i$ satisfies $g_i$, for $1$$\leq$$i$$\leq$$m$. 
The {\em result of $g$ based on $S_1,\ldots,S_m$}, denoted $g(S_1,\ldots,S_m)$, is the set of all $(s_1,\ldots,s_m,t)$ tuples such that $s_1$$\in$$\,S_1$, $\ldots$, $s_m$$\in$$\,S_m$ and $t$ is a {\em minimal} subtree of $\graph$ containing the nodes $s_1,\ldots,s_m$.  By minimal, we mean that ($i$)~removing any edge from $t$ disconnects it and/or removes some $s_i$ from $t$, and ($ii$)~$t$ contains only one node from each $S_i$. 

In our sample graph, let $S_1=\{n_2,n_4\}$ (US entrepreneurs), $S_2=\{n_3,n_6\}$ (French entrepreneurs), and $S_3=\{n_9\}$ (French politicians).  
Then, $g^1(S_1,S_2,S_3)$ includes 
$(n_4, n_6, n_9, t_\alpha)$ where the tree $t_\alpha$ consists of the edges $n_4\xrightarrow{e_{10}} n_7 \xleftarrow{e_9} n_{6} \xleftarrow{e_{11}} n_9$,  also denoted by $\{e_{10},e_{9},e_{11}\}$ for brevity. 
Another result of this CTP is $(n_2,n_3,n_9, t_\beta)$, with $t_\beta=\{e_1,e_2,e_{17},e_{16}\}$. This result is only possible because Def.~\ref{def:ctp-result} allows trees to span over $\graph$ edges {\em regardless of the edge direction}. Had it required directed trees, $t_\beta$ would not qualify, since none of its nodes can reach the others through unidirectional paths. 

The above definition allows arbitrary seed sets, in particular, an $S_i$ can be $\nodes$, the set of all graph nodes. We adjust Def.~\ref{def:ctp-result} to allow a connecting tree to have any number of nodes {\em from those seed sets equal to $\nodes$} (otherwise, only 1-node trees would appear in results). 

\vspace{1mm}
\noindent\textbf{Difference wrt path-based semantics} Consider a simple CTP $g'=(v_1,v_2,\underline{v_3})$ and two seed sets $S_1,S_2$. $g'(S_1,S_2)$  {\em may differ from} the set of all paths between an $S_1$ node and an $S_2$ node: for instance, a path going from $s_1$$\in$$\,S_1$ {\em through $s_1'$$\in$$\,S_1$} to $s_2$$\in$$\,S_2$ cannot appear in $g'(S_1,S_2)$, because of our minimality condition ($ii$), requiring  {\em direct} connections between seeds from different sets. Further, consider a CTP $g''=(v_1,v_2,v_3,\underline{v_4})$ and some seed sets $S_1,S_2,S_3$. One may try to compute $g''(S_1,S_2,S_3)$ by a three-way join of the paths from a common root node $r$, to a node from $S_1$, one from $S_2$ and one from $S_3$; we call this approach \textbf{path stitching}. The results may differ even more: ($i$)~for each tree of $n$ nodes that appears in $g''(S_1,S_2,S_3)$, the three-way join produces $n$ results, that need deduplication; ($ii$)~if a path from $r$ to $s_1$ has common nodes or even common edges with a path from $r$ to $s_2$ and/or the one from $r$ to $s_3$, the join of these paths is {\em not a tree}, thus it cannot appear in a CTP result.
This is why in this work, we compute CTP results directly (not via stitching).

Note that a CTP can have a very large number of results, as illustrated by the graph in Figure~\ref{fig:sample-chain}.  A CTP $(1,N+1,\underline{v_3})$, asking for all the connections between the end nodes, has $2^N$ solutions, or $2^{|E|/2}$, which grows exponentially in $|E|$,  the number of graph edges. This is why \textbf{complete CTP result computation may be unfeasible} in some cases, and we will include in our language \textbf{CTP filters} for limiting the CTP result computation effort. 

\begin{figure}[t!]
  \centering
\begin{tikzpicture}
  \node (0) at (0, 0) {$1$};
  \node (1) at (1, 0) {$2$};
  \draw[->] (0) to [bend left]  node[above] {$a$} (1);
  \draw[->] (0) to [bend right] node[below] {$a$} (1);
  \node (2) at (2, 0) {$3$};
  \draw[->] (1) to [bend left]  node[above] {$a$} (2);
  \draw[->] (1) to [bend right] node[below] {$b$}(2);
  \node (3) at (3, 0) {$4$};
  \draw[->] (2) to [bend left]  node[above] {$a$} (3);
  \draw[->] (2) to [bend right] node[below] {$b$}(3);
  \node (dots) at (4,0) {$\ldots$};  
   \node (5) at (5, 0) {$N$};
  \node (6) at (6.5, 0) {$N+1$};
  \draw[->] (5) to [bend left]  node[above] {$a$} (6);
  \draw[->] (5) to [bend right] node[below] {$b$} (6);
\end{tikzpicture}
\vspace{-5mm}
\caption{Sample ``chain'' graph.\label{fig:sample-chain}}
\vspace{-5mm}
\end{figure}
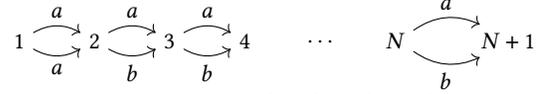




  We call {\em simple variable} in a query a variable that does not occur in the last position in a CTP. For a core query $Q$, we define: 

\definition[Simple embedding] A simple embedding of $Q$  in $\graph$ is a function $\phi$ mapping each simple variable into a $\graph$ node or edge,  such that: 

\begin{enumerate}
 \item The restriction of $\phi$ to each BGP pattern $b$  of $Q$ is an embedding of $b$ into $\graph$ (Def.~\ref{def:bgp-embedding}); 
 \item For each query CTP of the form $g=(g_1,\ldots,g_m,\underline{v_{m+1}})$, such that the simple variable in the predicate $g_i$, for $1$$\leq$$i$$\leq$$m$, is $v_i$, $\phi$ maps each $v_i$ into a $\graph$ node satisfying $g_i$.
\end{enumerate}


\definition[Core query result] \label{def:core-query-result} Let $Q$ be a core query having the head variables $u_1,\ldots,u_n$, and the simple variables  $v_1,\ldots,v_p$. 
Let $\Phi$ be the set of all $(\phi(v_1),\ldots,\phi(v_p))$ tuples for any simple embedding $\phi$ of $Q$ in $\graph$.
For each CTP $g^j$ in $Q$ of the form $(g_1,\ldots,g_m,\underline{v_{m+1}})$, let $v^j_i$ be the simple variable in $g_i$, $1$$\leq$$i$$\leq$$m$.
We define the $i$-th seed set of $g^j$, denoted $S^j_i$, as $\pi_{v^j_i}(\Phi)$, that is: all the nodes to which $v^j_i$ is bound in $\Phi$.  
The result of $Q$ is:

\vspace{.5mm}
\noindent$\;\,Q(\graph)=\pi_{u_1,\ldots,u_n}(\Phi \bowtie g^1(S^1_1,\ldots, S^1_{m_1}) \bowtie \ldots \bowtie g^l(S^l_1,\ldots, S^l_{m_l}))$

\vspace{.5mm}
\noindent where $g^1,\ldots,g^l$ are the CTPs of $Q$, having respectively $m_j$ simple variables, $1$$\leq$$j$$\leq$$l$, $g^j(S^j_1,\ldots, S^j_{m_j})$ is the set-based CTP result of $g^j$ (Def.~\ref{def:ctp-result}) on its seed sets derived from $\Phi$, and $\bowtie$ denotes the natural join on all the simple variables. 

\vspace{1.5mm}
\noindent\textbf{CTP filters} A set of orthogonal language extensions, which allow to filter (restrict) set based CTP results, are also provided. 

The keyword \textsf{\small UNI} after a CTP indicates that only {\em unidirectional} trees are sought, that is: a tree $t$, as in Def.~\ref{def:ctp-result},  must have a {\em root} node, from which a {\em directed path}
goes to each seed node in $t$.

Adding \textsf{\small LABEL} and a set of labels $\{l_1,l_2,\ldots,l_k\}$ after a CTP indicates that the edges in any result of that CTP must have labels from the given set. 

Adding \textsf{\small MAX} $n$ after a CTP indicates that only trees of at most $n$ edges are sought.

A \textbf{score function} $\sigma$ can be used to assign to each tree in a CTP result a real number $\sigma(t)$ (the higher, the better). Specifying (for a given CTP or for the whole query) \textsf{\small SCORE $\sigma$ [TOP $k$]} means that the results of each CTPs must be scored using $\sigma$, and the scores included in the query result. The optional \textsf{\small TOP $k$} allows to restrict the CTP result to those having the $k$-highest $\sigma$ scores. 

Finally, a practical way to limit the evaluation of a CTP (recall the example on Figure~\ref{fig:sample-chain}) is to specify a \textsf{\small timeout} $T$ (maximum allowed evaluation time); for simplicity, we consider the same $T$ is allotted to each CTP in a query. 

\definition[Query] A query consists of a core query, together with $0$ or more filters for each CTP. 

The semantics of a query is easily derived from that of a core query (Def.~\ref{def:core-query-result}), by filtering set-based CTP results accordingly.

%% file: baseline-algo.tex
\mysection{Query Evaluation Strategy}\label{sec:query-evaluation}

An EQL query consists  of a set of BGPs and a set of CTPs. Our evaluation strategy consists of the following steps: 

\noindent(A) Evaluate each BGP $b_i$, that is, compute all embeddings of its variables, and materialize them in a table $B_i$.

\noindent(B) For each CTP $g^j$ of the query, of the form $(g^j_1,\ldots, g^j_{m_j},\underline{v^j_{m_j+1}})$:
  \begin{enumerate}
  \item For $1\leq i \leq m_j$, where $v^j_i$ is the variable in $g^j_i$, compute the seed set $S_i^j$ as follows. 
    \begin{itemize}
      \item If $v^j_i$ appears also in one of the $B_i$, take  $S^j_i$  to be $\pi_{v^j_i}(B_i)$ (all the nodes to which $v^j_i$ has been bound). Further, if $g^j_i$ is not an empty predicate, restrict $S^j_i$ to only those nodes that also satisfy $g^j_i$.
      \item Otherwise, we obtain $S^j_i$ by restricting $\nodes$ (the graph's nodes set) to those that match $g^j_i$.
     \end{itemize}   
  \item \label{item:CTP}  Compute $F_j(g^j(S_1^j,\ldots, S_{m_j}^j))$, where $F_j(\cdot)$ applies all the CTP filters that may be attached to $g^j$. In practice, we actually {\em push the filters in the CTP evaluation}. Thus, we use the notation $g^j(S_1^j,\ldots,S_{m_j}^j,F_j)$ to denote the {\em set-based result of $g^j$ given its seed sets and filters}, and store it in a table $CTP_j$.
  \end{enumerate}

\noindent(C) Compute the query result as a projection on the head variables, over the natural join of the $B_i$ and  $CTP_j$ tables. 


All the above steps but (B) can be implemented by leveraging an existing conjunctive graph query engine. Thus, in the sequel, we focus on efficiently computing set-based CTP results. 

\mysection{Computing set-based CTP results}\label{sec:ctp-evaluation}

To compute $g(S_1,\ldots,S_m,F)$, we must find all the minimal subtrees of $\graph=(\nodes,\edges)$ containing exactly one node (or \textbf{seed}) from each $S_i$, also taking into account the filters $F$. Since $F$ is optional, we first discuss how to compute CTP results without any filter (Section~\ref{sec:algo-bfs} to \ref{sec:algo-molesp}), before discussing pushing filters (Section~\ref{sec:filters}). 

\noindent\textbf{Observation 1.}
Let us call \textbf{leaf} any node in a tree that is adjacent to exactly one edge. 
It is easy to see that \textbf{in each CTP result,  every leaf node is a seed}. (Otherwise, the leaf could be removed while still preserving an answer, which contradicts the minimality of the result.)
Clearly, the converse does not hold: in a result, some seeds may be internal nodes. 
We denote by  \textbf{sat($t$)} the node sets from which $t$ has a seed. 

\noindent\textbf{Observation 2.} 
As stated in Section~\ref{sec:language-construct}, we may be only computing {\em partial} CTP results. In such cases, it is reasonable to {\em return at least the smallest-size results}, given that tree size (smaller is better) is an ingredient of many score functions (see Section~\ref{sec:related-work}), and small results are easy to understand. However, we do not assume ``smaller is always better'': that is for the score function $\sigma$ to decide. Nor do we require users to specify a maximum result size, which may be hard for them to guess.  Rather,  we consider algorithms that {\em find as many results as possible, as fast as possible}, also taking into account the {\em CTP filters}, which may limit the search. 

\noindent\textbf{Seed set size} Most of our discussion assumes that no seed set is $\nodes$, and that they all fit easily in memory. We briefly discuss how the contrary situations could be handled, in Section~\ref{sec:survival}.

\mysubsection{Simple Breadth-First algorithm (\bfs)}
\label{sec:algo-bfs}
The first algorithm we consider finds the {\em t}ree results in {\em b}readth-{\em f}irst fashion, thus we call it \bfs . It starts by creating a first generation of trees $T_0$, containing a one-node tree, denoted \Init($n$), for each seed node  $n\in S_1 \cup \ldots \cup S_m$. Then, from each generation $T_i$, it builds the trees $T_{i+1}$, by  ``growing'' each tree $t$ in $T_i$, successively, with every edge $(n,n')$ adjacent to one of its nodes $n\in t$, such that: 

\vspace{-.5mm}
\begin{itemize}
\item (\textsc{Grow1}):~$n'$ is not already in $t$, and
  \vspace{-.3mm}
  \item (\textsc{Grow2}):~$n'$ is not a seed from a set $S_j\in \text{sat}(t)$.
\end{itemize}
\vspace{-.5mm}

Condition (\textsc{Grow1}) ensures we only build trees. (\textsc{Grow2}) enforces the CTP result minimality condition ($ii$) (Def.~\ref{def:ctp-result}). 
As trees grow from their original seed, they can include more seeds. When a tree has a seed from each set, it must be minimized, by removing all edges that do not lead to a seed, before reporting it in the result. For instance, with the seed sets $\{n_2\}$ and $\{n_4\}$ on the graph in Figure~\ref{fig:eg-graph}, starting from $n_2$, \bfs\ may build $\{e_5,e_4\}$, then $\{e_5,e_4,e_6\}$ before realizing that $e_4$ is useless, and removing it through minimization. Minimization slows \bfs\ down, as we experimentally show  in Section~\ref{sec:exp-gam-bfs}.
\bfs\ can build a tree in multiple ways;  to avoid duplicate work, any tree built during the search must be stored, and each new tree is checked against this memory of the search. 

It is easy to see that \textbf{\bfs\ is complete}, 
i.e., given enough time and memory, it finds all CTP results.

\mysubsection{GAM algorithm}
\label{sec:algo-gam}
The GAM (Grow and Aggressive Merge) algorithm has been introduced recently~\cite{gam-inf-sys-2022}, reusing some ideas from~\cite{dpbf}.
Unlike \bfs\ that views a tree as a set of edges, GAM {\em distinguishes one root node in each tree} it builds. 
The algorithm uses a {\em priority queue} where \textsc{Grow} opportunities are inserted, as (tree, edge) pairs such that the tree could grow from its root with that edge. 

GAM also starts from the set of \Init\ trees built from the seed sets. Next, it inserts in the priority queue all $(t,e)$ pairs for some \Init\ tree $t$ and edge $e$ adjacent to the root (only node) of $t$, satisfying the conditions (\textsc{Grow1}) and (\textsc{Grow2}) introduced in Section~\ref{sec:algo-bfs}. 
GAM then repeats the following, until no new trees can be built, or a time-out is reached: 

  \begin{enumerate}
  \item (\textsc{Grow}): Pop a highest-priority $(t,e)$ pair from the priority queue, where $e=(t.root, n')$, and build the tree $t^i$ having all edges of $t$ as well as $e$, and rooted in $n'$.
  \item (\textsc{Merge}): For any tree $t^{ii}$ already built, such that: 
    \begin{itemize}
    \item (\textsc{Merge1}): $t^{ii}$ has the same root as $t^i$, and no other node in common with $t^i$; and 
  \item (\textsc{Merge2}): sat($t^i$)$\,\cap\,$sat($t^{ii}$)$=\emptyset$, 
  \end{itemize}

  take the following steps:
  \begin{enumerate}
  \item \label{step:merge} Create $t^{iii}$, a tree having the edges of $t^i$ and those of $t^{ii}$, and the same root as $t^i$ and $t^{ii}$;
  \item \label{step:repeat} Immediately \textsc{Merge} $t^{iii}$ with all qualifying trees (see conditions \textsc{Merge1}, \textsc{Merge2}), and again merge the resulting trees etc., until no more \textsc{Merge} are possible;
      \end{enumerate}
  \item \label{step:push} For each tree $t^{iv}$ created via \textsc{Grow} or \textsc{Merge} as above: ($i$)~if $t^{iv}$ has a seed from each set, report it as a result; ($ii$)~otherwise, push in the priority queue all $(t^{iv},e^{iv})$ pairs such that $e^{iv}$ is adjacent to the (only) root node of $t^{iv}$, satisfying the conditions (\textsc{Grow1}) and (\textsc{Grow2}).  
  \end{enumerate}

\property[GAM completeness] \label{prop:gam-completeness} The GAM algorithm is complete.
  
\property[GAM result minimality] \label{prop:gam-minimal} By construction, each result tree built by GAM is minimal (in the sense of Def.~\ref{def:ctp-result}). 
\printIfExtVersion{\begin{proof}
First, note that the \Grow\ and \Merge\ conditions ensure that only trees are built, and they have at most one seed from each seed set.
Next, we show that {\em in any tree ever built by GAM, all the leaves (with the possible exception of the root) are seed nodes}. We show that by induction over the tree structure: 
\begin{itemize}
\item Every initial tree in $T_0$ consists of one seed. 
  \item Now assume the induction hypothesis is true for a tree $t$. A \Grow($t,e$) step adds a new root that is a leaf, and may or may not be a seed; the other leaves of \Grow($t,e$) are also leaves in $t$, thus seeds. 
  \item Similarly, assume this holds for two trees $t_1,t_2$. When $t$ is obtained as \Merge($t_1,t_2$), the root of $t$ is by definition not a leaf (it has at least two adjacent edges), while its leaves (those of $t_1$ and $t_2$) are seeds, by the induction hypothesis.
  \end{itemize}
  Thus, in a GAM tree, all the leaves are seeds; when the root is a leaf, it may or may not be a seed.

  We can now finalize proving that GAM builds minimal results, as follows. 
  Step~(\ref{step:push})($i$) above tests whether each new tree is a result. When this is true of a \Grow\ tree, its root is also a seed, thus Observation 1 holds for GAM results found by \Grow.
When a \Merge tree is a result, we have shown above that all its leaves are seeds, while the root is by definition ot a leaf. Thus, Observation 1 also holds for GAM results found by \Merge.
\end{proof}
}

Thus, \textbf{GAM does not need to minimize} the results it finds. 

~\\
\noindent\textbf{Search space exploration order} Unlike \bfs, GAM does not build trees in the strictly increasing order of their size; \Merge\  may build quite large trees before some other, smaller trees.
The order in which GAM enumerates trees is determined, first, by the priority of the queue which holds $(t,e)$ entries, and second, by the available \Merge\ opportunities.
In this work, \textbf{to remain compatible with \underline{any} score function, we study search algorithms regardless of (orthogonally to) the search order.}

Like \bfs, GAM may also build a tree in multiple ways. Formally:

\begin{definition}[Tree with provenance]
A tree with provenance (or provenance, in short) is a formula of one of the forms shown below, together with one node called the {\em provenance root}: 
\begin{enumerate}
    \item \Init($n$) where $n$ is a seed; the root of such a provenance is $n$ itself; 
    \item \Grow($t,e$) where $t$ is a provenance, its root is $n_0$, $e$ is an edge going from $n_0$ to $n_1$ and $n_1$ does not appear in $t$; in this case, $n_1$ is the root of the \Grow provenance;
    \item \Merge($t_1,t_2$), where $t_1$ and $t_2$ are provenances, rooted in $n_1$$=$$n_2$; in this case, $n_1$ is the root of the \Merge provenance.
\end{enumerate}
\end{definition}

We call \textbf{rooted tree} a set of edges that, together, form a tree, together with one distinguished root node.  GAM may build several provenances for the same rooted tree, e.g., \Merge(\Merge($t_1,t_2$),$t_3$) and \Merge($t_2$,\Merge($t_1,t_3$)), for some trees $t_1,t_2,t_3$. The interest of a tree as part of a possible result does not depend on its provenance. Therefore, \textbf{GAM discards all but the first provenance built for a given rooted tree}. 

\mysubsection{\bfs\ variants with \Merge}
\label{sec:bfs-variants}
The \Merge\ operation can also be injected in the \bfs\ algorithm to allow it to build some larger trees before all the smaller trees have been enumerated. We study two variants: \bfsm\ {\em m}erges each new tree resulting from \Grow\ with all its compatible partners (Step (\ref{step:merge}) in  Section~\ref{sec:algo-gam}), but does not apply \Merge\ on top of these \Merge\ results; in contrast, \bfsam\ applies both Step~(\ref{step:merge}) and Step~(\ref{step:repeat}) to {\em a}ggressively {\em m}erge. \bfsm\ and \bfsam\ are obviously complete. Like \bfs, they still need to minimize a potential result before reporting it. This is because \bfs\ algorithms {\em grow trees from any of their nodes}, thus may add edges on one side of one seed node, which later turn to be useless. GAM avoids this by growing only from the root.

%% file: moesp.tex
\mysubsection{Edge set pruning  and ESP algorithm}
\label{sec:algo-esp}
GAM may build several rooted trees for the same set of edges. For example, on the graph in Figure~\ref{fig:gam-incompleteness-eg} with the seeds $\{B\}$, $\{C\}$, \textbf{denoting a rooted tree by its edges and underlining the root}, successive \Grow\ from B lead to B-3-\underline{C},  successive \Grow\ from C lead to \underline{B}-3-C, and \Merge\ of  two \Grow provenances yields B-\underline{3}-C.
However, the root is meaningless in a CTP result, which is simply a set of edges. We introduce: 

\begin{definition}[Edge set]\label{def:edge-set}
An edge set is a set of edges that, together, form a tree such that at most $1$ leaf is not a seed. 
\end{definition}

A result is a particular case of edge set, where all leaves are seeds (recall Observation 1). 

As GAM builds several rooted trees for an edge set, it {\em repeats some effort}: we only need to find each result once. This leads to the following pruning idea: 

\begin{definition}[Edge-set pruning (ESP)]
  \label{def:esp}
  The ESP pruning technique during GAM consists of discarding any provenance $t_1$ whose edge set is non-empty, such that another provenance $t_0$, corresponding to the same edge set, had been created previously. 
\end{definition}

We will call ESP, in short, the GAM algorithm (Section~\ref{sec:algo-gam}) enhanced with ESP. As we will show, \textbf{ESP significantly speeds up GAM execution}. However,  \textbf{ESP compromises completeness} for some graphs, seed sets, {\em and execution orders}. That is: {\em depending on the order in which various trees are built}, the first (and only, due to ESP) provenance for a given edge set may prevent the algorithm from finding some results. 

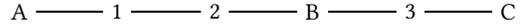
\begin{figure}
\centering
\tikzstyle{node} = [text centered, fill=white]
\tikzstyle{arrow} = [thick,->,>=stealth]

\begin{tikzpicture}[node distance=1.3cm]
\node (A) [node] {A};
\node (n1) [node, right of=A] {1};
\node (n2) [node, right of=n1] {2};
\node (B) [node, right of=n2] {B};
\node (n3) [node, right of=B] {3};
\node (C) [node, right of=n3] {C};
\draw [line width=.8pt] (A) -> (n1); 
\draw [line width=.8pt] (n1) -> (n2);
\draw [line width=.8pt] (n2) -> (B); 
\draw [line width=.8pt] (B) -> (n3);
\draw [line width=.8pt] (n3) -> (C); 
\end{tikzpicture}
\vspace{-4mm}
\caption{ESP incompleteness example. \label{fig:gam-incompleteness-eg}}
\vspace{-4mm}
\end{figure}

For instance, consider the graph in Figure~\ref{fig:gam-incompleteness-eg}, and the seed sets $S_1=\{A\}, S_2=\{B\}, S_3=\{C\}$. A possible execution of GAM is: 

\begin{enumerate}
\item Initial trees: \underline{A}, \underline{B}, \underline{C}. 
\item A set of \Grow\ lead to these trees: A–\underline{1},  B–\underline{2},  B–\underline{3},  C–\underline{3}.
\item B-\underline{3} and C-\underline{3} merge into B-\underline{3}-C.
\item \label{step:a12b} \Grow on A-\underline{1} leads to A-1-\underline{2}, which immediately merges with B-\underline{2}, forming A-1-\underline{2}-B.
\item After this point:
  \begin{itemize}
  \item If the tree A-1-2-\underline{B} is built, for instance by
    \Grow\ on  A-1-\underline{2}, ESP discards it since
    A-1-\underline{2}-B was found in step~(\ref{step:a12b}). Lacking
    A-1-2-\underline{B}, we cannot \Grow over it to build the result provenance A-1-2-B-3-\underline{C}. Nor can we build the result provenance \Merge(A-1-2-\underline{B}, \underline{B}-3-C). 
  \item By a similar reasoning, when \underline{B}-3-C is built, it is discarded by ESP, preventing the construction of of \underline{A}-1-2-B-3-C.
  \end{itemize}
   Thus, no result is found. 
\end{enumerate}

Note that {\em with a favorable execution order}, the CTP  result would be found. For instance, from \underline{A}, \underline{B}, \underline{C}, ESP could build:

\begin{enumerate}
\item Through successive \Grow: A–\underline{1}, A-1-\underline{2}, A-1-2-\underline{B}, C-\underline{3}, C-3-\underline{B}
\item Then, \Merge(A-1-2-\underline{B}, C-3-\underline{B}) is a provenance for the result. 
\end{enumerate}

This raises the question: can we pick a GAM execution order that would
ensure completeness, even when using ESP? Intuitively, the order
should ensure that {\em for each result $r$, there exists a provenance
  $p_r$ for $r$ which is certainly built}, which requires that {\em at
  every sub-expression $e$ of $p_r$, over an edge set $es$, the first
  provenance $p_{es}$ we find for $es$ happens to be rooted in a node that allows to build on $e$ until $p_r$}. Thus, the decisions made up to building $p_{es}$ would need to have a ``look-ahead'' knowledge of the {\em future} of the search, which is clearly not possible.
In the above example 
the ``bad'' order builds A-1-\underline{2}-B  first, whereas it would be more favorable to build A-1-2-\underline{B}. However, when exploring these three edges, the future of the exploration is not known; thus, we cannot ``pre-determine'' the best provenance for $es$. Recall also from Section~\ref{sec:algo-gam} that different orders may be suited for partial exploration with different score functions. In a conservative way, we consider an algorithm incomplete when {\em for some ``bad'' execution order} it may miss results. 

We show that ESP finds {\em some} answers {\em for any execution order}: 

\begin{property}[$2$-seed sets ESP completeness]\label{prop:two-completeness}
Let $t$ be a result of a CTP with $2$ seed sets. Then, $t$ is guaranteed to be found by ESP. 
\end{property}

Here and throughout this paper, {\em guaranteed to be found}, for a rooted tree or an edge set, means that at least one provenance for it is built;  ESP cannot prune the one built first.

For $1$ seed set, Property~\ref{prop:two-completeness} is trivially shown,  thus we focus on $m=2$ (two seed sets). In this case, any result is path of $0$ or more edges. We introduce: 

\begin{definition}[$(n,s)$-rooted path]\label{def:ns-rooted-path}
  Given a CTP and its seed sets $S_1,S_2,\ldots,S_m$, an {\em $(n,s)$-rooted path} is a rooted path from a seed $s$  to a root node $n$, such that the only seed in the path is $s$. 
\end{definition}

\begin{lemma}
  \label{lemma:n-rooted-path}
  Any $(n,s)$-rooted path  is guaranteed to be found by GAM with ESP. 
\end{lemma}

\begin{proof}
  We prove this by exhibiting a provenance for it. First, for each seed $s\in S_1\cup \ldots \cup S_m$, \Init($s$) is guaranteed to be built. ESP pruning does not apply.
  Then, any provenance applying only \Grow steps on an \Init\ provenance, is guaranteed to be built by GAM. Such a provenance is not pruned by ESP, because it is the \emph{only} provenance that could lead to  its edge set. 
Thus, successive \Grow on top of any seed $s$ is guaranteed to build up to $n$, leading  to the $(n,s)$-rooted path.
\end{proof}

Based on the above lemma, we prove Property \ref{prop:two-completeness}: 
\begin{proof}
  If  the result $t$ is a node ($s_1=s_2$), the property is trivial.
  If the result is a path of $1$ edge,  there are two provenances of the form \Grow(\Init); the first is already a result. 
  Now, assume  $t$ has at least two edges. For any internal node $n$ in $t$, the $(n,s_i)$-rooted paths from both the (seed) leaves $s_1,s_2$ of $t$ are guaranteed to be found, by Lemma~\ref{lemma:n-rooted-path}.  Then, one of two cases may occur: 
  \textbf{(1)} For some internal node $n_0$, both rooted paths $(n_0,s_1)$ and $(n_0,s_2)$  are created {\em before} a sequence of \Grow\ gets from \Init($s_1$) to $s_2$, and {\em before} the opposite sequence of \Grow is built from \Init($s_2$), to $s_1$. Without loss of generality, let $n_0$ denote the {\em first} internal node for which these two rooted paths are created. Immediately,  \Merge\ on these  creates a provenance of $t$. By the way we chose $n_0$, this is the first provenance for this edge set, thus not pruned. 
  \textbf{(2)} On the contrary, assume that successive \Grow\ get from one end of the path to another, {\em before} two rooted paths meet in any internal node. Assume without loss of generality that \Grow(\Grow(\ldots \Init($s_1$)\ldots)) is the first one to reach $s_2$. Again, by design, this is the first provenance for $t$, thus not pruned.
\end{proof}

CTP with two seed sets (path queries) are frequent in practice; on these, GAM~\cite{gam-inf-sys-2022} and ESP are comparable, and we experimentally show the latter is much more efficient. Next, we add more algorithmic refinements to significantly extend our completeness guarantees.

\mysubsection{MoESP algorithm}
\label{sec:algo-moesp}
We now introduce an algorithmic variant called {\em Merge-oriented ESP}, or MoESP, which finds many (but not all) CTP results for arbitrary numbers of seed sets.

MoESP works like ESP, but it creates more trees. Specifically, whenever \Grow\ or \Merge\ produces a provenance $t$ {\em having strictly more seeds than any of its (one or two) children}, the algorithm builds from  $t$ all the so-called \textbf{MoESP trees} $t'$ such that: 

\begin{itemize}
\item  $t'$ has the same edges (and nodes) as $t$, but
\item $t'$ is rooted in a seed node, distinct from the root of $t$. 
\end{itemize}

The provenance of any such $t'$ is denoted \Mo($t$, $r$) where \Mo\ is special symbol and $r$ is the root of $t'$.
Within MoESP, 
 \textbf{\Merge\ is allowed on MoESP trees, but not \Grow}. More generally, \Grow\ is disabled on any tree whose provenance includes \Mo.

Clearly, MoESP builds a strict superset of the rooted trees created by ESP (thus, it finds all results of  ESP). It also finds the result in Figure~\ref{fig:gam-incompleteness-eg}. Namely, after creating \underline{A}, \underline{B}, \underline{C}: 

\begin{enumerate}
\item  \Grow\ leads to the trees: A–\underline{1},  B–\underline{2},  B–\underline{3},  C–\underline{3}.
\item B-\underline{3} and C-\underline{3} merge into B-\underline{3}-C. MoESP trees are added at this point: B-3-\underline{C} and \underline{B}-3-C.
\item \Grow on A-\underline{1} leads to A-1-\underline{2}, which  merges with B-\underline{2}, forming A-1-\underline{2}-B. Similarly, A-1-2-\underline{B} and \underline{A}-1-2-B are added. 
\item A-1-2-\underline{B}  merges with \underline{B}-3-C, leading to the result. 
\end{enumerate}

We now generalize the example by establishing completeness guarantees for MoESP.

\begin{figure}
\centering
\tikzstyle{node} = [text centered, fill=white]
\tikzstyle{arrow} = [thick, ->,>=stealth]

\begin{tikzpicture}[node distance=8mm]
\node (A) [node] {A};
\node (n1) [node, right of=A] {1};
\node (n2) [node, right of=n1] {2};
\node (B) [node, right of=n2] {B};
\node (n3) [node, right of=B] {3};
\node (C) [node, right of=n3] {C};
\draw [color=blue, line width=.8pt](A) -> (n1); 
\draw [color=blue, line width=.8pt] (n1) -> (n2);
\draw [color=blue, line width=.8pt] (n2) -> (B); 
\draw  [color=violet, line width=.8pt](B) -> (n3);
\draw  [color=violet, line width=.8pt](n3) -> (C);
\node (n4) [node, below of=A] {4};
\node (n5) [node, below of=n1,xshift=-1.7mm] {5};
\node (n6) [node, below of=n2,xshift=1.3mm] {6};
\node (n7) [node, below of=B] {7};
\node (n8) [node, below of=n3] {8};
\node (n9) [node, below of=C] {9};
\draw [color=red, line width=.8pt] (A) -> (n4); 
\draw  [line width=.8pt](A) -> (n5);
\draw  [line width=.8pt] (B) -> (n6); 
\draw  [color=violet, line width=.8pt] (B) -> (n7);
\draw  [color=violet, line width=.8pt] (B) -> (n8);
\draw  [line width=.8pt] (C) -> (n9);
\node (D) [node, below of=n4] {D};
\node (n10) [node, below of=n5, xshift=5.7mm] {10};
\node (E) [node, below of=n7] {E};
\node (F) [node, below of=n9] {F};
\draw [color=red, line width=.8pt](n4) -> (D); 
\draw  [line width=.8pt](n5) -> (n10);
\draw  [line width=.8pt] (n6) -> (n10); 
\draw  [color=violet, line width=.8pt] (n7) -> (E);
\draw  [color=violet, line width=.8pt] (n8) -> (F);
\draw  [line width=.8pt](n9) -> (F);
\end{tikzpicture}
\vspace{-4mm}
\caption{Sample graph for MoESP discussion. \label{fig:sample-MoESP}}
\vspace{-5mm}
\end{figure}
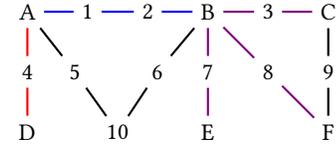

\begin{definition}[Simple and $p$-simple edge set]\label{def:simple-edge-set}
  A simple edge set is an edge set (Def.~\ref{def:edge-set}) where each leaf is a seed 
  and no internal (non-leaf) node is a seed. A simple edge set is $p$-simple, for some integer $p$, if its number of leaves is at most $p$.
\end{definition}

For instance, consider the sample graph in Figure~\ref{fig:sample-MoESP}, and the 6 seed sets $\{A\}$, $\{B\}$, $\{C\}$, $\{D\}$, $\{E\}$, $\{F\}$. The edge set A-4-D, shown in red,  is 2-simple, and so are: A-1-2-B, shown in blue; B-8-F, etc. 

\begin{definition}[Simple tree decomposition of a
  solution]\label{def:simple-decomposition} Let $t$ be a CTP result.  A simple tree decomposition of $t$, denoted $\theta(t)$, is a set of simple edge sets which ($i$)~are a partition of the edges of $t$ and ($ii$)~may share (leaf) nodes with each other.
\end{definition}

For instance, in Figure~\ref{fig:sample-MoESP}, the red, blue, and violet edges, together, form a result for the 6-seed sets CTP.  A simple tree decomposition of this solution is: \{A-4-D, A-1-2-B, B-7-E, B-8-F, B-3-C\}. It is easy to see that a tree $t$ has a unique simple tree decomposition $\theta(t)$. 


\begin{definition}[$p$-piecewise simple
  solution] \label{def:p-piecewise-simple}
  A result $t$ is $p$-piecewise simple ($p$ps, in short), for some integer $p$, if every edge set in the simple tree decomposition $\theta(t)$ 
  is $p$-simple (Def.~\ref{def:simple-edge-set}). 
\end{definition}

The sample result above in Figure~\ref{fig:sample-MoESP}
is $2$ps, since its simple tree decomposition only contains 2-simple edge sets.
The following important \moesp\ property guarantees it is found: 

\begin{property}[MoESP finds $2$-piecewise simple solutions]\label{prop:moesp-2ps}
 For any number of seed sets $m$,  MoESP is guaranteed to find any $2$-piecewise simple result. 
\end{property}

\begin{proof}
  Let $t$ be a $2$-piecewise simple solution and
  $\theta(t)=\{t_1,\ldots, t_r\}$ be its simple tree decomposition. It
  is easy to see that each $t_i$, $1\leq i \leq r$, is a path of the
  form $n_1^i, \ldots, n_m^i$ such that $n_1^i$ and $n_m^i$ are seeds,
  while no other intermediary node is a seed. 
  Lemma~\ref{lemma:n-rooted-path}, which still holds for MoESP, guarantees that rooted paths are built starting from both $n_1^i$ and $n_m^i$.
  As soon as these paths meet, a tree over the edges of $t_i$ is created, then thanks to MoESP, one tree rooted in $n_1^i$ and another rooted in $n_m^i$, over the edge set of $t_i$, are created.
  Because $\theta(t)$ is a simple tree decomposition of $t$, if $r=1$, the property is proved. If $r>1$, each seed-rooted tree based on the edge set of a $t_i$ has its root in common with at least another seed-rooted tree over another edge set(s) from $\theta(t)$. Therefore, aggressive \Merge ensures that they are eventually all merged, leading to one provenance for $t$. 
\end{proof}


For a CTP with any number $m$ of seed sets, a {\em path result} is one in which no node has more than two adjacent edges. In a path result, seed and non-seed nodes alternate, with the two ends of the paths being seeds. Thus, any path result is $2$ps. It follows then, as a direct consequence of Property~\ref{prop:moesp-2ps}:

 \begin{property}[MoESP finds all path results]\label{prop:moesp-path-solutions}
 For any CTP, MoESP finds all the path results. 
 \end{property}

 However, outside $2$ps results, MoESP may still fail. For instance, consider the graph in Figure~\ref{fig:moesp-incompleteness-eg}, and the seed sets $\{A\}, \{B\}, \{C\}$. The only result here is $3$-simple. A possible MoESP execution order is: 

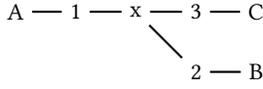
\begin{figure}
  \centering
\tikzstyle{node} = [text centered, fill=white]
\tikzstyle{arrow} = [thick,->,>=stealth]
\begin{tikzpicture}[node distance=8mm]
\node (x) [node] {x};
\node (n1) [node, left of=x] {1};
\node (A) [node, left of=n1] {A};
\node (n3) [node, right of=x] {3};
\node (C) [node, right of=n3] {C};
\node (n2) [node, below of=n3] {2};
\node (B) [node, right of=n2] {B};
\draw [line width=.8pt] (A) -> (n1);
\draw [line width=.8pt] (n1) -> (x); 
\draw [line width=.8pt] (x) -> (n2);
\draw [line width=.8pt] (x) -> (n3);
\draw [line width=.8pt] (n2) -> (B); 
\draw [line width=.8pt] (n3) -> (C);
\end{tikzpicture}
\vspace{-4mm}
\caption{MoESP incompleteness example.   \label{fig:moesp-incompleteness-eg}}
\vspace{-5mm}
\end{figure}

\begin{enumerate}
\item Starting from \underline{A}, \underline{B}, \underline{C},  \Grow\ produces A–\underline{1},  B–\underline{2},  C–\underline{3}; 
\item \label{step:moesp-bx} B–2–\underline{x}, followed by B–2–x–\underline{3}, which merges with C-\underline{3} into B–2–x–\underline{3}–C, leading also to \underline{B}–2–x–3–C and B–2–x–3–\underline{C}. 
\item \label{step:b2xa} B-2–x-\underline{1} which merges with A-\underline{1}, leading to  B-2-x-\underline{1}-A and similar trees rooted in B and A. 
\item \label{step:MoESP1} \Grow\ produces A-1-\underline{x}. ESP discards the \Merge\ of A-1-\underline{x} with B–2–\underline{x}, due to the rooted tree built at step~(\ref{step:b2xa}), over the same set of edges. 
\item \label{step:a1xc} A–1–x–\underline{3} is built, then \Merge\  with C-\underline{3} creates A–1–x–\underline{3}–C, and similar trees rooted in A and C. 
\item \label{step:MoESP2}  \Grow\ produces C-3-\underline{x}. ESP discards the merges of C-3-\underline{x} with  A-1-\underline{x} due to the 3-rooted tree built at step~(\ref{step:a1xc})  and  with  B–2–\underline{x} due to the 3-rooted tree built at step~(\ref{step:moesp-bx}). 
\item At this point, we have trees with two seeds, rooted in \underline{1}, \underline{3}, \underline{A}, \underline{B} and \underline{C}.  \Grow\  on any of them is impossible, because they already contain all the edges adjacent to their roots. There are no \Merge\ possibilities on their roots, either. Thus, the search fails to find a result. 
\end{enumerate}

%
%
%

At steps (\ref{step:MoESP1}) and (\ref{step:MoESP2}), ESP is ``short-sighted'': it prevents the construction of some trees, necessary for finding the result. Next, we present another optimization which prevents such errors.

%% file: smesp.tex
\mysubsection{\lesp\ algorithm}\label{sec:algo-lesp}
The Limited Edge-Set Pruning (\lesp), in short, works like ESP (Section~\ref{sec:algo-esp}), but it {\em limits} edge-set pruning, as follows.

\begin{itemize}
\item \label{step:spesp:signature} We assign to each node $n$, and maintain throughout \lesp\ execution, a {\em seed signature} $ss_n$, indicating the seed sets $S_i$, $1\leq i \leq m$, such that a $(n,s_i)$-rooted path (Def.~\ref{def:ns-rooted-path}) has been built from a seed $s_i\in S_i$, to $n$, since execution started. 
  For any seed $s\in S_i$, the signature $ss_s$ is initialized to $0\ldots 1\ldots0$ (a single $1$ in the $i$-th position). For a non-seed $n$, initially $ss_n$=0; the $i$-th bit is set to $1$ when node $n$ is reached by the first rooted path from a seed in $S_i$.
\item \label{step:smesp:limit} Prevent ESP from discarding a \Merge\
  tree rooted in $n$ such that: ($i$) $\sum(ss_n)\geq 3$, that is,
  there are at least $3$ bits set to $1$ in the signature $ss_n$; and ($ii$)~$n$ has at least $3$ adjacent edges in $\graph$. 
\end{itemize}

Intuitively, the condition on $ss_n$ {\em encourages merging on nodes already well-connected to seeds}. 
We denote by $d_n$  the number of $\graph$ edges adjacent to $n$; it can be computed and stored before evaluating any query.
The condition on $d_n$ focuses the ``protection against ESP'' to \Merge\ trees  {\em rooted in nodes where such protection is likely to be most useful}: specifically, those where $3$ or more rooted paths can meet (see Lemma~\ref{lemma:3-maximal-merge} below).
\Grow\ and \Merge apply on trees ``spared'' in this way with no restriction. 

Clearly, LESP creates all the trees built by ESP, and may create more. In particular, reconsider the  graph in Figure~\ref{fig:moesp-incompleteness-eg}, the associated seed sets, and the execution steps we traced in Section~\ref{sec:algo-moesp}. 
At step (\ref{step:moesp-bx}), $ss_x$ is initialized with $010$ (there is a path from B to x).
At step (\ref{step:MoESP1}), when A-1-\underline{x} is built, $ss_x$ becomes $110$;  since $\sum(ss_x)=2$, the tree A-1-\underline{x}-2-B is pruned. 
However, at step (\ref{step:MoESP2}), when C-3-\underline{x} is built, $ss_x$ becomes $111$, which, together with $d_x=3$, spares its \Merge\ result  A-1-\underline{x}-3-C (despite the presence of several trees with the same edges). In turn, this merges immediately with B-2-\underline{x} into a result. 

We formalize the guarantees of \lesp\ as follows.

\begin{definition}[$(u,n)$ rooted merge]\label{def:n-rooted-merge}
For an integer $u\geq 3$ and non-seed node $n$, the $(u,n)$ rooted merge is the rooted tree resulting from merging a set of $u$ $(n,s_i)$ rooted paths, for some seeds $s_1,\ldots,s_u$.
\end{definition}

It follows from the (\textsc{Merge2}) pre-condition (Section~\ref{sec:algo-gam}) that in an $(u,n)$ rooted merge, each $s_i$ belongs to a different seed set.
Further, it follows from the definition of an $(n,s_i)$-rooted path,
that in a $(u,n)$ rooted merge, all seeds are on leaves. In other
words, a $(u,n)$ rooted merge is a $u$-simple edge set.

\begin{lemma}
  \label{lemma:3-maximal-merge}
  Any $(3,n)$ rooted merge is guaranteed to be found  by \lesp.
\end{lemma}

\begin{proof}
  For any non-seed node $n$,  Lemma \ref{lemma:n-rooted-path} (which also holds for \lesp) ensures that any $(n,s_i)$-rooted path is found. 
  As soon as the third one is built, $\sum(ss_n)$ becomes $3$. This, and the hypothesis $d_n\geq 3$,  ensure that the \Merge\ of the three is not pruned. 
\end{proof}

\property\label{prop:n-maximal-merge}
For any integer $u\geq 3$ and non-seed node $n$, any $(u,n)$ rooted merge is guaranteed to be found by \lesp. 

\begin{proof}
For $u=3$ this is established by Lemma~\ref{lemma:3-maximal-merge}. 
Once the first $(3,n)$ rooted merge has been built and kept, this ensures both that $d_n\geq 3$ and $\sum(ss_n)\geq 3$. 
Then, whenever a new $(n,s_i)$ rooted path, satisfying the \Merge\ pre-conditions, is built, it is aggressively merged with the first $(3,n)$ rooted path, and the result is protected from pruning by \lesp's special provision. The same  holds during all subsequent merges with other $(n,s_j)$ rooted paths.
\end{proof}

For $4$ or more seed sets, \lesp\ may miss results that are not $(u,n)$ rooted merges. For instance, consider the following order of execution for $S=(\{A\}, \{B\}, \{C\}, \{D\})$ on the graph in Figure \ref{fig:gam-incompleteness-4-keyword}:


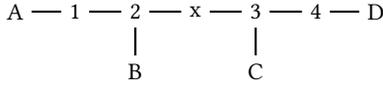
\begin{figure}
  \centering
  \tikzstyle{node} = [text centered, fill=white]
\tikzstyle{arrow} = [thick,->,>=stealth]
\begin{tikzpicture}[node distance=8mm]
\node (x) [node] {x};
\node (n2) [node, left of=x] {2};
\node (B) [node, below of=n2] {B};
\node (n1) [node, left of=n2] {1};
\node (A) [node, left of=n1] {A};
\node (n3) [node, right of=x] {3};
\node (C) [node, below of=n3] {C};
\node (n4) [node, right of=n3] {4};
\node (D) [node, right of=n4] {D};
\draw [line width=.8pt] (x) -> (n2);
\draw [line width=.8pt] (n2) -> (n1); 
\draw [line width=.8pt] (n1) -> (A);
\draw [line width=.8pt] (x) -> (n3);
\draw [line width=.8pt] (n3) -> (n4); 
\draw [line width=.8pt] (n4) -> (D);
\draw [line width=.8pt] (n2) -> (B);
\draw [line width=.8pt] (n3) -> (C); 
\end{tikzpicture}
\vspace{-4mm}
  \caption{\lesp\ incompleteness example with $4$ seed sets.}\label{fig:gam-incompleteness-4-keyword}
  \vspace{-4mm}
\end{figure}

\begin{enumerate}
\item From \underline{A}, \underline{B}, \underline{C}, \underline{D}, \Grow\ builds: A–\underline{1},  B–\underline{2}, C-\underline{3}, D-\underline{4}.
\item \label{enum:problem-trees1} \Grow\ builds B–2-\underline{1} which merges with A-\underline{1} into A-\underline{1}–2-B.
\item \label{enum:problem-trees2}  \Grow\ builds C-3-\underline{4} which merges with D-\underline{4} into C-3-\underline{4}-D.
\item \Grow\ builds: A-1-\underline{2};
    B–2-\underline{x} which cannot merge with B-\underline{2} as A-\underline{1}–2-B exists, and $\sum(ss_2)=2$;
    D-4-\underline{3} which cannot merge with C-\underline{3} as C-3-\underline{4}-D exists, and $\sum(ss_3)=2$.
\item C-3-\underline{x} merges with B–2-\underline{x} to build B–2-\underline{x}-3-C. 
\item C-3-x-\underline{2} merges with: A-1-\underline{2}, leading to C-3-x-\underline{2}-1-A;  and B–\underline{2}, leading to C-3-x-\underline{2}-B.
\item Similarly, B-2-x-\underline{3}, aggressively merges with C-\underline{3}, leading to B-2-x-\underline{3}-C, and D-4-\underline{3}, leading to B-2-x-\underline{3}-4-D.
\item Progressing similarly, we can only merge at most $3$ rooted paths, in nodes 2, x or 3. 
  We cannot merge with a path leading to the $4$th seed, because the trees with the edge sets A-1-2-B and C-3-4-D, built at (\ref{enum:problem-trees1}), (\ref{enum:problem-trees2}) above, are not rooted in 2 nor 3, respectively, and these are the only nodes satisfying the \lesp\  condition that ``spares'' some \Merge trees.
\end{enumerate}

\mysubsection{MoLESP algorithm}
\label{sec:algo-molesp}

Our last algorithm, called \molesp, is a GAM variant with ESP and {\em both} the modifications of MoESP (which injects more trees) and \lesp\ (which avoids ESP pruning for some \Merge trees). Clearly, \molesp\  finds all the trees found by MoESP and \lesp. Further: 

\begin{property}[MoLESP finds all $3$ps results] MoLESP is guaranteed to find all the $3$-piecewise simple results.\label{prop:smoesp-3ps}
\end{property}
\begin{proof}
  Let $t$ be a $3$ps result. If $t$ was $2$ps, MoESP finds it (Property~\ref{prop:moesp-2ps}), thus MoLESP also does.

 Now consider that $\theta(t)$ has some $3$-simple edge sets that are not $2$-simple (thus, $m\geq 3$). 
  We show that for any $3$-simple edge set in $\theta(t)$, one provenance is built. Let $t^3$ be such an edge set: its three leaves, denoted $n_1,n_2,n_3$, are seeds, and no internal node is a seed. Let $c$ denote the central node in $t^3$ (connected to $n_1,n_2,n_3$ by pairwise disjoint paths). 
  $t^3$ is a $(3,c)$ rooted merge (recall Def. \ref{def:n-rooted-merge}) and one provenance for it is built (Lemma \ref{lemma:3-maximal-merge}).
  
 The rest of the proof follows the idea in the proof of Property~\ref{prop:moesp-2ps}. The MoESP aspect of \molesp\ guarantees that for each edge set in $\theta(t)$, one tree rooted in each seed  is built and not pruned; eventually, aggressive \Merge of these trees builds a provenance for $t$. 
\end{proof}

As an important consequence: 

\begin{property}MoLESP is complete for $m\leq 3$ seed sets.
\end{property}
\begin{proof}\vspace{-1.5mm}
  Consider the possible result shapes: ($i$)~a single node $s_1=s_2=s_3$: no ESP applies, thus it is found;
  ($ii$)~a path going from $s_1=s_2$ to $s_3$; such a result is $2$-simple; 
  ($iii$)~a path going from $s_1$ to $s_2$ and then to $s_3$, for some pairwise distinct $s_1,s_2,s_3$; such a  result is $2$ps; 
  ($iv$)~a tree with three distinct leaves $s_1,s_2,s_3$, which is $3$-simple. 
  In cases ($ii$), ($iii$), ($iv$),  Property~\ref{prop:smoesp-3ps} ensures the result is found.
\end{proof}

\vspace{-1.5mm}
Our strongest completeness result is:

\begin{property}[Restricted \molesp\ completeness] For any CTP of $m\geq 1$ seeds, \molesp\  finds any result $t$, such that:
  each edge set $es\in \theta(t)$ is a $(u,n)$-rooted merge (Def.~\ref{def:n-rooted-merge}), for some integer $1\leq u\leq m$ and non-seed node $n$ in $es$. 
\label{prop:restricted-nps}
\end{property}

\begin{proof}\vspace{-1.5mm}
  Let $t$ be a result, and assume it is $v$-piecewise simple, for some integer $v$. If $v\in \{2, 3\}$, Property \ref{prop:smoesp-3ps} ensures \molesp\ finds it. 

  On the contrary, assume $v\geq 4$ and  let $t^4\in \theta(t)$ be a
  $(v,n)$-rooted merge for some non-seed node $n$, thus, also
  $v$-simple. Property \ref{prop:n-maximal-merge}, which also holds
  during \molesp, guarantees that one provenance for $t^4$ is
  built. The end of our proof leverages the MoESP aspect of the
  algorithm: for each such edge set in $\theta(t)$, one tree rooted in each seed is built and not pruned; eventually, aggressive \Merge of these trees builds a provenance for $t$. 
\end{proof}

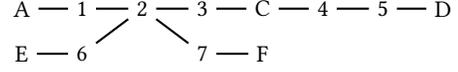
\begin{figure}[t!]
  \vspace{-6mm}
  \centering
  \tikzstyle{node} = [text centered, fill=white]
\tikzstyle{arrow} = [thick,->,>=stealth]
\begin{tikzpicture}[node distance=8mm]
\node (A) [node] {A};
\node (n1) [node, right of=A] {1};
\node (B) [node, right of=n1] {2};
\node (n2) [node, right of=B] {3};
\node (D) [node, right of=n2] {C};
\node (n3) [node, below of=n2, yshift=2mm] {7};
\node (C) [node, right of=n3] {F};
\node (n4) [node, right of=D] {4};
\node (n5) [node, right of=n4] {5};
\node (E) [node, right of=n5] {D};
\node (F) [node, below of=A, yshift=2mm] {E};
\node (n6) [node, right of=F] {6};
\draw [line width=.8pt] (A) -> (n1);
\draw [line width=.8pt] (n1) -> (B); 
\draw [line width=.8pt] (B) -> (n2);
\draw [line width=.8pt] (n2) -> (D);
\draw [line width=.8pt] (D) -> (n4); 
\draw [line width=.8pt] (n4) -> (n5);
\draw [line width=.8pt] (n5) -> (E); 
\draw [line width=.8pt] (B) -> (n3);
\draw [line width=.8pt] (n3) -> (C);
\draw [line width=.8pt] (n6) -> (F);
\draw [line width=.8pt] (n6) -> (B);
\end{tikzpicture}
\vspace{-4mm}
  \caption{\molesp\  completeness example.}\label{fig:gam-molesp}
\end{figure}

\begin{algorithm}[t!]
\setstretch{0.55}
\caption{\textsc{MoLESP}(graph $\graph$, seed sets $(S_1 \ldots,  S_m)$)}\label{alg:gamsearch-smoesp}
\KwOutput{Set of results, $\mathbf{Res}$}
Priority queue $\mathbf{PrioQ} \gets$ new priority queue\;
History $\mathbf{Hist}$ $\gets$ new set of edge sets\;
\ForEach {$S_i, 1\leq i \leq m$}{
\ForEach {$n_i^j \in S_i$}{
$t_i^j \gets$ \Init($n_i^j$); 
\textsc{processTree}($t_i^j$)\; 
}
}
\While{$\mathbf{PrioQ}$ is not empty}{
$(t,e) \gets poll(\mathbf{PrioQ})$; 
$t' \gets $ \Grow ($t,e$)\;  
Update $ss_{root(t')}$; 
\textsc{processTree}($t'$)\;
}
\end{algorithm}
\setlength{\textfloatsep}{0pt}
\setlength{\floatsep}{0pt}
\begin{algorithm}[t!]
\setstretch{0.55}
\caption{Procedure \textsc{processTree}(provenance $t$)} \label{alg:process-tree}
\If{\textsc{isNew}$(t)$}{
Add $t$ to $\mathbf{Hist}$ \;
\If{\textsc{isResult}$(t)$}{
Add $t$ to $\mathbf{Res}$\;
}
\Else{
\textsc{recordForMerging}($t$)\; 
\If{$t$ is not a MoESP tree}{ 
 \For{edge $e \in$ adjacentEdges($t.root$)}{
 \If{$hasNotBeenInQueue(t,e)$}{
 \label{line:push}Add $(t,e)$ to $\mathbf{PrioQ}$\; 
}
}
}
}
}
\end{algorithm}

\begin{algorithm}[t!]\caption{Procedure \textsc{recordForMerging}(tree $t$)}
\setstretch{0.55}
Add $t$ to $\mathbf{TreesRootedIn}[t.root]$\;
\For{$n \in (nodes(t) \, \cap \, \cup_i(S_i))$}{
\label{line:moesp1} Copy $t$ into a new tree $t'$, rooted at $n$, with provenance \textsc{Mo}$(t,n)$\;
Add $t'$ to $\mathbf{TreesRootedIn}[n]$\;
\label{line:moesp2}\textsc{MergeAll}$(t')$\;
}
\end{algorithm}

\begin{algorithm}[t!]\caption{Procedure \textsc{isNew}(tree $t$)}
\setstretch{0.55}
%
\If{$t \notin \mathbf{Hist}$}{
return $true$\;
} 

\If{$\Sigma(ss_{t.root}) \geq 3$ and $d_{t.root}\geq 3$}{
\If{$t \notin \mathbf{TreesRootedIn}[t.root]$}{
return $true$\;
}
}
return $false$\; 
\end{algorithm}

\vspace{-15mm}
For example, in Figure~\ref{fig:gam-molesp}, with the six seeds $A$ to $F$, the result is guaranteed to be found by \molesp. Depending on the exploration order, \moesp\ and \lesp\ may not find it.

\vspace{1mm}
\noindent\textbf{MoLESP algorithm}  Algorithms
\ref{alg:gamsearch-smoesp} to \ref{alg:mergeall-smoesp}, together,
implement MoLESP. They share a set of global variables whose names
start with an uppercase letter: $\mathbf{Res}, \mathbf{PrioQ}$,
$\mathbf{Hist}$ (the search history), and $\mathbf{TreesRootedIn}$ (to
store the trees by their roots); the latter is needed to find \Merge\ candidates fast. Variables with lowercase names are local to each algorithm.
 \textsc{processTree} feeds the priority queue with (tree, edge) pairs at line \ref{line:push}. 
 \textsc{recordForMerging} injects the extra MoESP trees (Section~\ref{sec:algo-moesp}) at lines \ref{line:moesp1} to \ref{line:moesp2}. 
 \textsc{isNew} implements limited edge-set pruning based on the history, and the two conditions that can ``spare'' a tree from pruning (Section~\ref{sec:algo-lesp}). 
\textsc{mergeAll} implements aggressive merging; by calling \textsc{processTree} on each new \Merge\ result, through \textsc{recordForMerging}, the result is available in the future iterations of \textsc{mergeAll}, thus ensuring all the desired \Merge.

\mysubsection{CTP evaluation in the presence of filters}
\label{sec:filters}
We now briefly explain how various CTP filters (Section~\ref{sec:language-construct}) can be inserted within the above algorithms. 
\textsf{UNI}-directional search is
enforced by adding pre-conditions to \Grow\ and \Merge, to ensure we only create the desired provenances. 
\textsf{\small LABEL} $\{l_1,l_2,\ldots,l_k\}$ is enforced 
by restricting the \Grow\ edges to only those carrying one of these labels; in GAM and its variants, we only add in the queue (line \ref{line:push} in \textsc{processTree}),  (tree, edge) pairs where the edge has an allowed label. \textsf{\small MAX} $n$ prevents \Grow\ and \Merge\ from creating a tree of more than $n$ edges. \textsf{\small timeout} $T$ is checked after each newly found rooted tree and within each algorithm's main loop.

For \textsf{\small SCORE $\sigma$ [TOP $k$]}, the simplest implementation calls $\sigma$ on each new result; a vast majority of the proposed score functions can score each result independently. If the score of a result can only be computed once {\em all} the results are found, e.g.~\cite{rootrank@comad2019, kws@infsys2020}, the results need to be accumulated. For \textbf{any given score} $\sigma$, a smarter implementation may {\em favor (with guarantees, or just heuristically) the early production of higher-score results, by  appropriately chosing the priority queue order}; this allows search to finish faster. \textbf{Any} order can be chosen in conjunction with \molesp, since its completeness guarantees are independent of the exploration order.

\setlength{\textfloatsep}{0pt}
\setlength{\floatsep}{0pt}
\begin{algorithm}[t!]
\setstretch{0.55}
\caption{Procedure \textsc{MergeAll}(tree $t$)}\label{alg:mergeall-smoesp}
$\mathbf{toBeMerged} \gets \{t\}$\;
\While{$\mathbf{toBeMerged} \neq \emptyset$}{
$\mathbf{currentTrees} \gets \mathbf{toBeMerged}$; $\mathbf{toBeMerged} \gets \emptyset$\;
\For{$t'$ $\in$ $\mathbf{currentTrees}$}{
$\mathbf{mergePartners} \gets \mathbf{TreesRootedIn}[t'.root]$\;
\For{$t_p \in \mathbf{mergePartners}$}{
\If{sat($t'$)$ \,\cap\, $sat($t_p$)$\,=\emptyset$ and $t' \cap t_p = \{t'.root\}$}{
$t'' \gets\,$\Merge($t', t_p$)\; 
\If{\textsc{isNew}($t''$)}{
Add $t''$ to $\mathbf{toBeMerged}$\;
\textsc{processTree}($t''$)\;
}
}
} 
}
}
\end{algorithm}

\mysubsection{Handling very large seed sets}
\label{sec:survival}
Our CTP evaluation algorithms build \Init\ trees for each seed. This has two risks: 
($i$)~when one or more seed sets are $\nodes$ (all graph nodes), exploring them all may be unfeasible; 
($ii$)~one or more seed sets may be subsets of $\nodes$, yet still much larger, e.g., one or more orders of magnitude, than the other seed sets.
To handle ($i$), assuming other seed sets are smaller, we only start
exploring (\Init, \Grow\ etc.) from the other seed sets, and simplify
accordingly the algorithms, since any encountered node is acceptable as a match for the $\nodes$ seed set(s).
To handle ($ii$), borrowing ideas from prior work~\cite{banks-2}, we use {\em multiple priority queues}, one for each subset of the seed sets, and \Grow\  at any point {\em from the queue having the fewest (tree, edge) pairs}. Thus, exploration initially focuses on the neighborhood of the smaller seed sets, and hopefully encounters \Init\ trees from the large seed sets, leading to results.

%% file: evaluation1.tex
\mysection{Experimental evaluation}\label{sec:experiments}


We compare CTP evaluation algorithms, then consider systems capable, to some extent, to evaluate the language we introduced.

\mysubsection{Software and hardware setup}\label{subsec:impl}

We implemented a parser and a query compiler for our language (Section~\ref{sec:language-construct}) as an extension of SPARQL, and all the CTP evaluation algorithms from Section~\ref{sec:ctp-evaluation}, in Java 11. Our graphs are stored in a simple table \textsf{\small graph(\underline{id},source, edgeLabel, target)} within PostgreSQL 12.4; unless otherwise specified, we delegate to Postgres the BGP evaluation, and joining their results with CTP ones (Section~\ref{sec:query-evaluation}). 
When comparing CTP evaluation algorithms with in-memory competitors, we load the graph in memory prior to evaluating CTPs.

We executed our experiments on a server equipped with 2x10-core Intel Xeon E5-2640 CPUs $@$ 2.4GHz, with 128-GB DRAM.  Every execution point is averaged over $3$ executions.

\mysubsection{Baselines}
\label{sec:baselines}
\noindent\textbf{CTP evaluation (keyword search) algorithms}
Our focus is on algorithms that search for connecting trees ($i$)~traversing edges in both directions, ($ii$)~orthogonally wrt the score function,  ($iii$)~exhaustively, at least up to $m$=3 seed sets, ($iv$)~capable of returning as many solutions as requested, if given enough time and memory, and ($v$)~applicable to arbitrary graphs, i.e., not requiring a regular graph structure.
In the literature, only the GAM algorithm~\cite{gam-inf-sys-2022} (Section~\ref{sec:algo-gam}) fits the bill. The \bfs, \bfsm, \bfsam\  algorithms (Section~\ref{sec:algo-bfs} and \ref{sec:bfs-variants}) also satisfy these conditions, and are thus natural comparison baselines; like virtually all algorithms from the literature, they start from the seeds and move gradually away looking for results.

  QGSTP \cite{qgstp@www21} and LANCET \cite{lancet@vldb2021} are the most recent GSTP approximation algorithms, for specific cost functions based on node and edge (LANCET) weights. 
  LANCET relies on DPBF~\cite{dpbf} to find an initial result, which it then improves.  
  Since QGSTP has shown strong advantage over DPBF~\cite{qgstp@www21}, we select QGSTP as a baseline. 
  QGSTP runs in polynomial time in the size of the graph, and by design, returns only \emph{one} result; we used the authors' code.

\noindent\textbf{Graph query engines}
Our first two baselines only support  {\em checking, but not returning} unbounded-length, unidirectional paths whose edge labels match a regular expression that users {\em must} provide, that is: one cannot ask for ``any path''.
Specifically, we use \textbf{Virtuoso} OpenSource v7.2.6 to evaluate SPARQL 1.1 queries that come as close as possible to the semantics of our language. 
Internally, Virtuoso translates an incoming SPARQL query into an SQL dialect\footnote{Accessible using the built-in function \texttt{sparql\_to\_sql\_text()}.} before executing it. 
Our second baseline, named \textbf{Virtuoso-SQL}, consists of editing these SQL-like queries to remove label constraints and thus query the graph for connectivity between nodes.
However, Virtuoso's SQL dialect prevented us from {\em returning} the nodes and edge labels along the found paths (whereas standard recursive SQL allows it).

Our next three baselines support checking {\em and returning} paths. 
\textbf{JEDI}~\cite{jedi@vldb2018}  returns all the data paths matching a SPARQL property path; we use the authors' code. 
\textbf{Neo4j} supports Cypher queries asking for all directed or undirected paths between two sets of nodes. 
Finally, we used recursive queries in \textbf{Postgres} v12.4 to return the label on paths between node pairs.

\mysubsection{Datasets and queries}
\label{sec:datasets}
We experiment with both synthetic and real-world RDF graphs. 

\begin{figure}
\centering
\tikzstyle{node} = [text centered, fill=white]
\tikzstyle{arrow} = [thick,->,>=stealth]

%
\begin{tikzpicture}[node distance=0.7cm]
\node (A) [node] { A};
\node (n1) [node, right of=A] { 1};
\node (n2) [node, right of=n1] { 2};
\node (B) [node, right of=n2] { B};
\node (n3) [node, right of=B] { 3};
\node (n4) [node, right of=n3] { 4};
\node (C) [node, right of=n4] { C};
\node (n5) [node, below of=A] { 5};
\node (D) [node, below of=n5] { D};
\node (n7) [node, below of=B] { 7};
\node (F) [node, below of=n7] { F};
\node (n9) [node, below of=C] { 9};
\node (H) [node, below of=n9] { H};
\draw [line width=.8pt] (A) -> (n1); 
\draw  [line width=.8pt] (n1) -> (n2); 
\draw [line width=.8pt]  (n2) -> (B); 
\draw [line width=.8pt] (B) -> (n3); 
\draw  [line width=.8pt] (n3) -> (n4); 
\draw  [line width=.8pt] (n4) -> (C); 

\draw [line width=.8pt] (A) -> (n5); 
\draw  [line width=.8pt] (n5) -> (D); 

\draw [line width=.8pt] (B) -> (n7); 
\draw  [line width=.8pt] (n7) -> (F); 

\draw [line width=.8pt] (C) -> (n9); 
\draw  [line width=.8pt] (n9) -> (H); 
\end{tikzpicture}
\begin{tikzpicture}[node distance=0.7cm]
\node (c) [node] {5};
\node (n1) [node, left of=c] {1};
\node (A) [node, left of=n1] {A};
\node (n2) [node, right of=c] {2};
\node (B) [node, right of=n2] {B};
\node (n3) [node, below of=n1] {3};
\node (C) [node, below of=A] {C};
\node (n4) [node, below of=n2] {4};
\node (D) [node, below of=B] {D};
\draw [line width=.8pt] (c) -> (n1); 
\draw [line width=.8pt] (n1) -> (A); 
\draw [line width=.8pt] (c) -> (n2); 
\draw [line width=.8pt] (n2) -> (B); 
\draw [line width=.8pt] (c) -> (n3); 
\draw [line width=.8pt] (n3) -> (C); 
\draw [line width=.8pt] (c) -> (n4); 
\draw [line width=.8pt] (n4) -> (D); 
\end{tikzpicture}
\begin{tikzpicture}[node distance=0.7cm]
\node (A) [node] {A};
\node (n1) [node, right of=A] {1};
\node (B) [node, right of=n1] {B};
\node (n2) [node, right of=B] {2};
\node (C) [node, right of=n2] {C};
\draw [line width=.8pt] (A) -> (n1); 
\draw [line width=.8pt] (n1) -> (B); 
\draw [line width=.8pt] (B) -> (n2);
\draw [line width=.8pt] (n2) -> (C); 
\end{tikzpicture}
\vspace{-4mm}
\caption{Synthetic graphs: Comb($3,1,2,3$) at the top left, Star($4,2$) at the top right, and Line($3,1$) at the bottom. \label{fig:syngraph-topology}}
\end{figure}
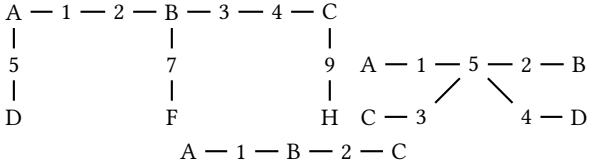

\titledparagraph{To compare CTP evaluation algorithms}, we generate three sets of  parameterized graphs and associated CTPs 
(Figure \ref{fig:syngraph-topology}).  The seeds are  labeled $A, B, \ldots, H$, non-seed nodes are labeled $1, 2$ etc.; each seed set is of size 1. 
\textbf{Line}($m,nL$) contains $m$ seeds, each connected to the next/previous seed by $nL$ intermediary nodes, using $sL$$=$$nL$$+$$1$ edges.  
\textbf{Comb}($nA,nS,sL,dBA$) consists of a line, from which a lateral segment (called {\em bristle}) exits each seed. There are $nA$ bristles, each made of $nS$ segments (a segment ends in another seed); each bristle segment has $sL$ triples, and there are $dBA$ nodes in the main line between two successive bristles.
The number of seeds is $m$$=$$nA$$\cdot$$(nS$$+$$1$$)$. 
\textbf{Star}($m,sL$) has a central node connected to each of the $m$ seeds by a line of $sL$ edges. 

On each Line, Comb, and Star graph, we run a CTP defined by the $m$ seeds, having $1$ result. For instance, on the Star in Figure~\ref{fig:syngraph-topology}, the seed sets are $\{A\}, \{B\}, \{C\}, \{D\}$.  On Line and Comb, the result is $2$ps (Def.~\ref{def:p-piecewise-simple}), while on Star, it is a $(u,n)$ rooted merge (Def.~\ref{def:n-rooted-merge}). Thus, by Property~\ref{prop:restricted-nps}, MoLESP is guaranteed to find them. 
The topology of Line graphs minimizes the number of subtrees for a given number of edges and seeds; specifically, there are $O((m$$\cdot$$nL)^2)$ subtrees, while the number of rooted trees is in  $O((m$$\cdot$$nL)^3)$. On the contrary, the Star topology raises the number of subtrees to $O(2^m$$\cdot$$sL^2)$, while its number of rooted trees is in $O(2^m$$\cdot$$sL^3)$. In Comb and Line graphs, MoESP trees (Section~\ref{sec:algo-moesp}) are part of results. 

\titledparagraph{To study the evaluation of our extended query language}, we generate parameterized \textbf{Connected Dense Forest (CDF)} graphs (see Figure~\ref{fig:dense-forest-summary}). 
Each graph contains a {\em top forest}, and a {\em bottom forest}; each of these is a set of $N_T$ disjoint, complete binary trees of depth $3$. 
{\em Links} connect leaves from the top and bottom forests. We generate CDFs for $m$$\in$$\{2,$$3\}$: when $m$$=$$2$, chains of triples connect a top leaf to a bottom one; when $m$$=$$3$, a Y-shaped connection goes from a top-forest leaf, to two bottom-forest ones. A CDF graph contains $N_L$ links, each made of $S_L$ triples. Only top leaves that are targets of ``c'' edges can participate to links, and we concentrate the links on 50\% of them (the others have no links). When $m$$=$$2$, only 50\% of the bottom forest leaves that are targets of ``g'' edges can participate; when $m$=3, 50\% of all the bottom forest leaf can participate.
The links are uniformly distributed across the eligible leaves. 
%
%
A CDF has $12$$\cdot$$N_T$$+$$N_L$$\cdot$$S_L$ edges; it has 
$14$$\cdot$$N_T$$+$$N_L$$\cdot$$(S_L$$-$$1$$)$ nodes if  $m$$=$$2$,  and
$14$$\cdot$$N_T$$+$$N_L$$\cdot$$S_L$ if $m$$=$$3$.  

On CDF graphs with $m$$=$$2$, we run the query
($v$,$tl$,$\underline{l}$) $\query$($x$,"c",$tl$), ($v$,"g",$bl$), ($bl$,$tl$,$\underline{l}$) whose two BGPs bind $tl$, respectively, $bl$ to leaves from the top and bottom forest, while its CTP asks for all the paths between each pair of such leaves. 
On graphs with $m$$=$$3$, we run ($v$,$tl$,$\underline{l}$)$\query$($x$, "c", $tl$), ($v$,"g",$bl_1$), ($v$,"h",$bl_2$),($tl$, $bl_1$, $bl_2$, $\underline{l}$), requiring connecting trees between $tl$, $bl_1$ and $bl_2$.
Each CDF query has $N_L$ answers, one for each link.

\input{figures/figure-CDF}

\vspace{1.5mm}
\titledparagraph{Real-world graphs} To compare with JEDI~\cite{jedi@vldb2018} and QGSTP~\cite{qgstp@www21},
we reused their datasets (a $6$M triples subset of YAGO3, and a $18$M triples subset of DBPedia), as well as their
queries.


%% file: figures/figure-CDF.tex
\begin{figure}[t]
\centering
\tikzstyle{node} = [text centered, fill=white]
\tikzstyle{arrow} = [thick,font = {\scriptsize}]
\scalebox{0.9}{
\begin{tikzpicture}[node distance=0.8cm and 0.8cm]
\node (n1) at (0,0) [node] {1};
\node (n2) [node, below left of=n1] {2};
\node (n3) [node, below right of=n1] {3};
\node (n5) [node, below of=n2] {5};
\node (n4) [node, left of=n5] {4};
\node (n6) [node, below of=n3] {6};
\node (n7) [node, right of=n6] {7};
\draw [->,line width=.8pt] (n1) -> (n2) node [midway,left,xshift=1mm,yshift=2mm] {a};
\draw [->,line width=.8pt]  (n1) -> (n3) node [midway,right,xshift=-1mm,yshift=2mm] {b};
\draw  [->,line width=.8pt] (n2) -> (n4) node [midway,left,xshift=1mm,yshift=1.5mm] {c};
\draw [->,line width=.8pt]   (n2) -> (n5) node [midway,right,yshift=1mm] {d};
\draw [->,line width=.8pt]  (n3) -> (n6) node [midway,left,yshift=0.9mm] {c};
\draw [->,line width=.8pt]  (n3) -> (n7) node [midway,right,xshift=-1mm,yshift=1.5mm] {d};
\node (n8) [node, below of=n5,xshift=1.5mm] {8};
\node (n12) [node, below of=n4, yshift=-8mm] {12};
\node (n13) [node, right of=n12] {13};
\node (n10) [node, below of=n13] {10};
\node (n9) [node, below right of=n10] {9};
\node (n11) [node, above right of=n9] {11};
\node (n14) [node, above of=n11] {14};
\node (n15) [node, right of=n14] {15};
\node (n16) [node, below of=n6] {16};
\node (dots1) [node, right of=n16,xshift=8mm] {$\ldots$};

\draw [->,line width=.8pt] (n6) -> (n8) node [midway,left] {link};
\draw [->,line width=.8pt]  (n8) -> (n12) node [midway,left,xshift=1.4mm,yshift=1mm] {link};
\draw [->,line width=.8pt]  (n10) -> (n12) node [midway,left,xshift=1.4mm,yshift=-1.5mm] {g};
\draw [->,line width=.8pt]   (n10) -> (n13) node [midway,right] {h};
\draw [->,line width=.8pt]  (n11) -> (n14) node [midway,left,yshift=-0.5mm] {g};
\draw [->,line width=.8pt]  (n11) -> (n15) node [midway,right,xshift=-1mm,yshift=-1mm] {h};
\draw [->,line width=.8pt]  (n9) -> (n10) node [midway,left,xshift=1mm,yshift=-2mm] {e};
\draw [->,line width=.8pt]  (n9) -> (n11) node [midway,right,xshift=-1mm,yshift=-2mm] {f};
\draw [->,line width=.8pt] (n6) -> (n16) node [midway,right,xshift=-.7mm,yshift=1mm] {link};
\draw [->,line width=.8pt]  (n16) -> (n14) node [midway,right] {link};

\node (x1) at (2,0) [node] {18};
\node (x2) [node, below left of=x1,xshift=1mm] {$\ldots$};
\node (x3) [node, below right of=x1,xshift=-1mm] {$\ldots$};
\draw [->,line width=.8pt] (x1) -> (x2) node [midway,left,xshift=1mm,yshift=2mm] {a};
\draw  [->,line width=.8pt] (x1) -> (x3) node [midway,right,xshift=-1mm,yshift=2mm] {b};
\node (x9)  [node, right of=n9, xshift=12mm] {19};
\node (x10) [node, above left of=x9,xshift=1mm] {$\ldots$};
\node (x11) [node, above right of=x9,xshift=-1mm] {$\ldots$};
\draw [->,line width=.8pt] (x9) -> (x10) node [midway,left,xshift=1mm,yshift=-2mm] {e};
\draw [->,line width=.8pt]  (x9) -> (x11) node [midway,right,xshift=-1mm,yshift=-2mm] {f};

\draw (2.8,0) -> (2.8,-4.5); 

\node (m1) at (4.7,0) [node] {1};
\node (m2) [node, below left of=m1] {2};
\node (m3) [node, below right of=m1] {3};
\node (m5) [node, below of=m2] {5};
\node (m4) [node, left of=m5] {4};
\node (m6) [node, below of=m3] {6};
\node (m7) [node, right of=m6] {7};
\draw [->,line width=.8pt] (m1) -> (m2) node [midway,left,xshift=1mm,yshift=2mm] {a};
\draw  [->,line width=.8pt] (m1) -> (m3) node [midway,right,xshift=-1mm,yshift=2mm] {b};
\draw [->,line width=.8pt]  (m2) -> (m4) node [midway,left,xshift=1mm,yshift=1.5mm] {c};
\draw [->,line width=.8pt]   (m2) -> (m5) node [midway,right,yshift=1mm] {d};
\draw [->,line width=.8pt]  (m3) -> (m6) node [midway,left,yshift=0.9mm] {c};
\draw [->,line width=.8pt]  (m3) -> (m7) node [midway,right,xshift=-1mm,yshift=1.5mm] {d};
\node (m8) [node, below of=m5,xshift=-3mm] {8};
\node (m12) [node, below of=m4,yshift=-8mm] {12};
\node (m13) [node, right of=m12] {13};
\node (m10) [node, below of=m13] {10};
\node (m9) [node, below right of=m10] {9};
\node (m11) [node, above right of=m9] {11};
\node (m14) [node, above of=m11] {14};
\node (m15) [node, right of=m14] {15};
\node (m16) [node, above of=m14,xshift=1mm] {16};
\node (dots2) [node, right of=m16,xshift=12mm] {$\ldots$};

\draw [->,line width=.8pt] (m4) -> (m8) node [midway,left] {link};
\draw  [->,line width=.8pt] (m8) -> (m12) node [midway,left,xshift=-.5mm,yshift=.5mm] {link};
\draw  [->,line width=.8pt] (m8) -> (m13) node [midway,right,xshift=-.3mm,yshift=.5mm] {link};
\draw [->,line width=.8pt] (m6) -> (m16) node [midway,left] {link};
\draw [->,line width=.8pt]  (m16) -> (m14)  node [midway,left,yshift=.5mm] {link};
\draw [->,line width=.8pt]  (m16) -> (m15) node [midway,right,xshift=-.5mm,yshift=.8mm] {link};
\draw [->,line width=.8pt]  (m10) -> (m12) node [midway,left,xshift=1.4mm,yshift=-1.5mm] {g};
\draw  [->,line width=.8pt]  (m10) -> (m13) node [midway,right] {h};
\draw  [->,line width=.8pt] (m11) -> (m14) node [midway,left,yshift=-0.5mm] {g};
\draw [->,line width=.8pt]  (m11) -> (m15) node [midway,right,xshift=-1mm,yshift=-1mm] {h};
\draw  [->,line width=.8pt] (m9) -> (m10) node [midway,left,xshift=1mm,yshift=-2mm] {e};
\draw  [->,line width=.8pt] (m9) -> (m11) node [midway,right,xshift=-1mm,yshift=-2mm] {f};

\node (y1) at (6.6,0) [node] {17};
\node (y2) [node, below left of=y1,xshift=1mm] {$\ldots$};
\node (y3) [node, below right of=y1,xshift=-1mm] {$\ldots$};
\draw [->,line width=.8pt] (y1) -> (y2) node [midway,left,xshift=1mm,yshift=2mm] {a};
\draw  [->,line width=.8pt] (y1) -> (y3) node [midway,right,xshift=-1mm,yshift=2mm] {b};
\node (y9)  [node, right of=m9, xshift=12mm] {18};
\node (y10) [node, above left of=y9,xshift=1mm] {$\ldots$};
\node (y11) [node, above right of=y9,xshift=-1mm] {$\ldots$};
\draw [->,line width=.8pt] (y9) -> (y10) node [midway,left,xshift=1mm,yshift=-2mm] {e};
\draw [->,line width=.8pt]  (y9) -> (y11) node [midway,right,xshift=-1mm,yshift=-2mm] {f};
\end{tikzpicture}}
\vspace{-5.5mm}
\caption{CDF graphs generated with $m$=2, $S_L$=2 (left), and  with $m$=3, $S_L=3$ (right). \label{fig:dense-forest-summary}}
\end{figure}
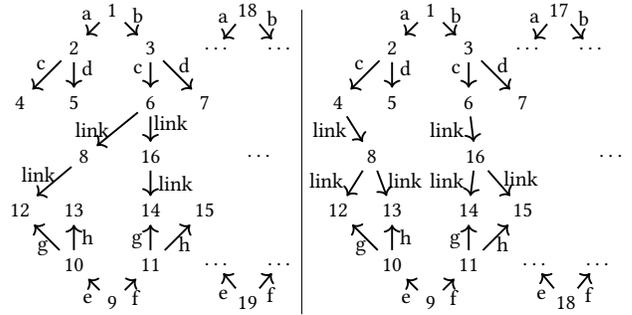

%% file: evaluation2.tex
\mysubsection{CTP evaluation algorithms}
\label{sec:exp-gam}


\begin{figure*}[th!]
\subfloat[CTP runtime on Line]{\includegraphics[width=0.3\textwidth]{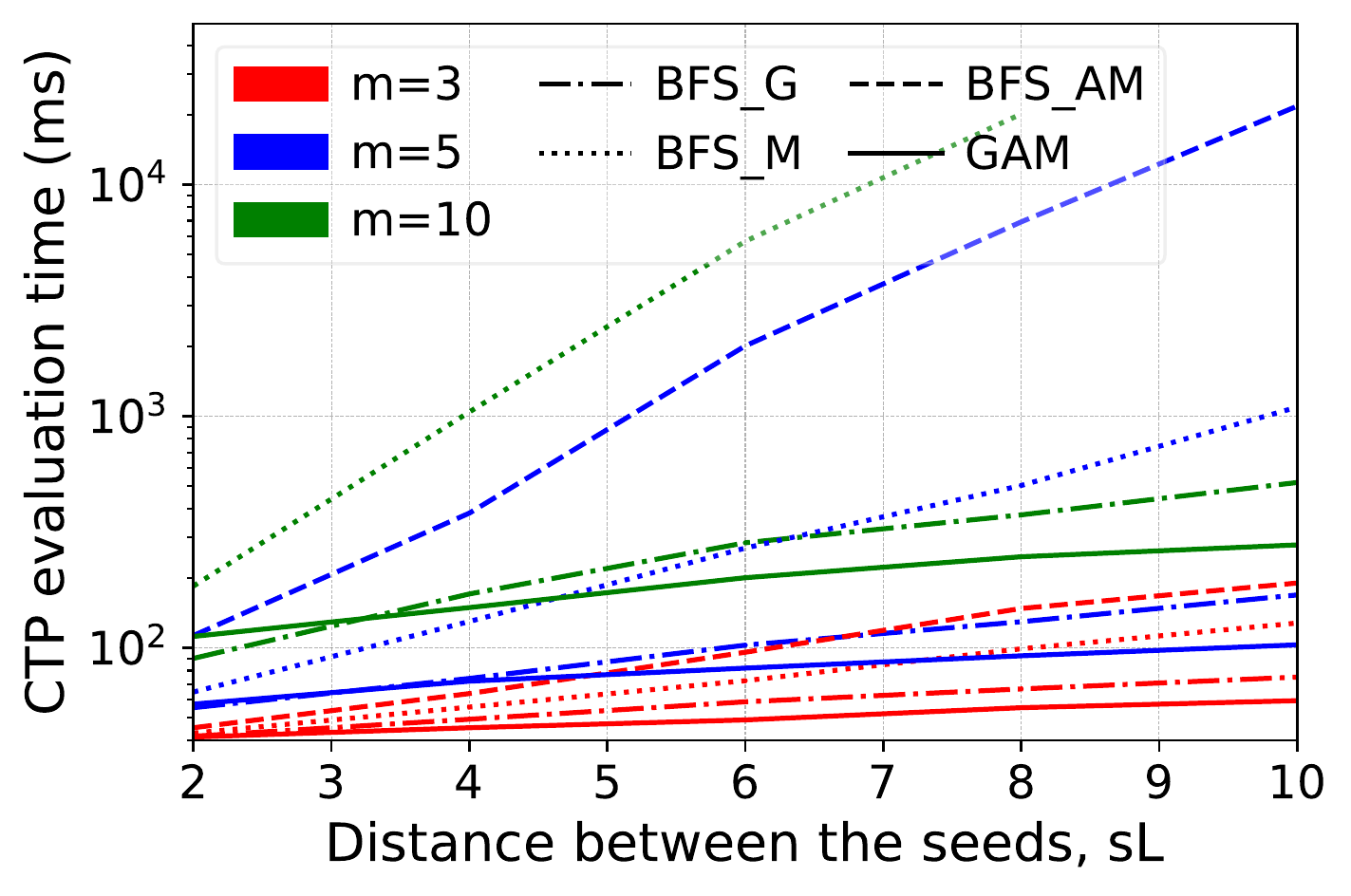}\label{fig:bfs-variants-line-rt}}
\subfloat[CTP runtime on Comb]{\includegraphics[width=0.3\textwidth]{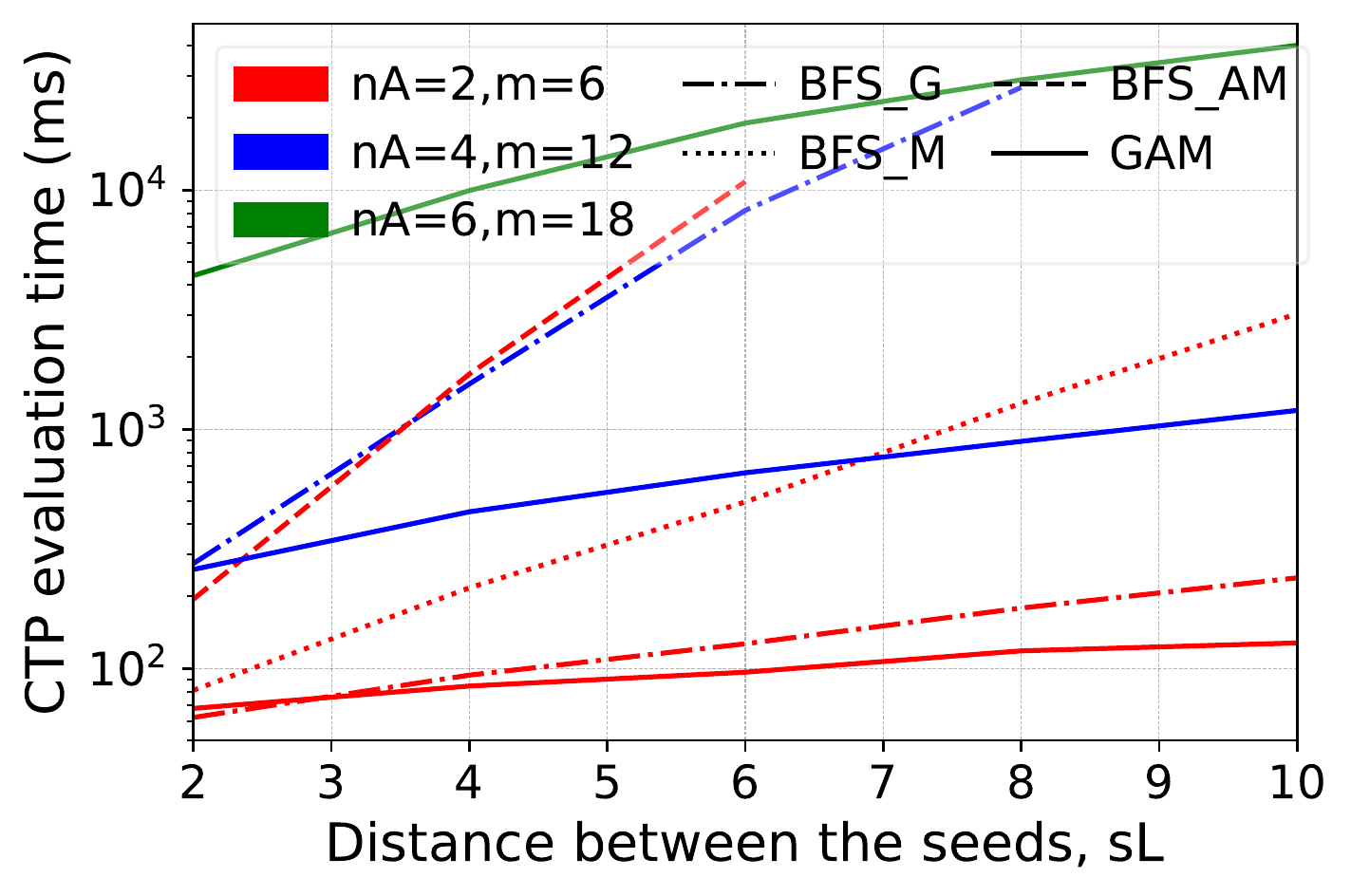}\label{fig:bfs-variants-comb-rt}}
\subfloat[CTP runtime on Star]{\includegraphics[width=0.3\textwidth]{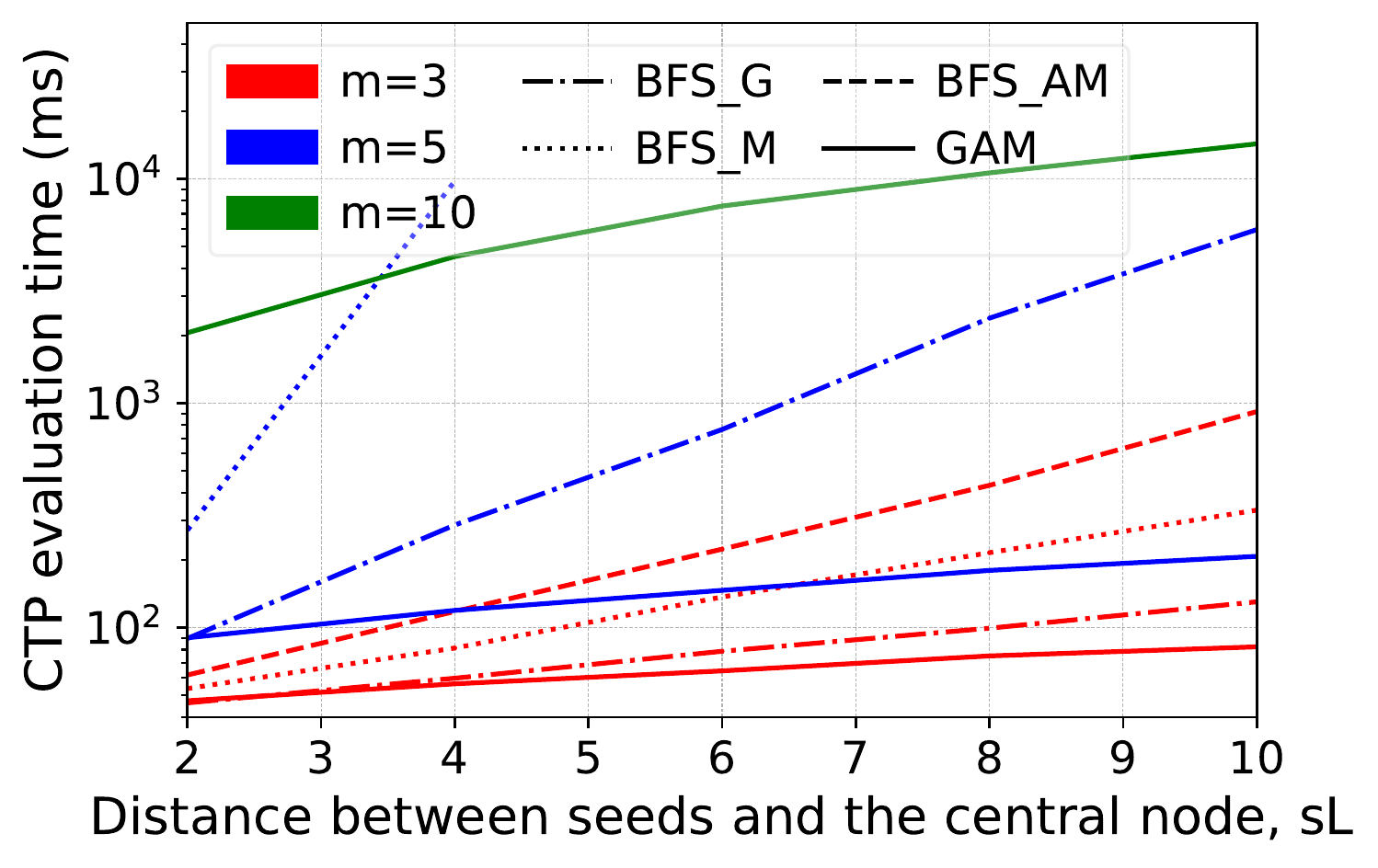}\label{fig:bfs-variants-star-rt}}\\
\vspace{-4mm}
\caption{Comparison of complete CTP evaluation baselines.}\label{fig:bfs-variants-synthetic}
\vspace{-4mm}
\end{figure*}

\begin{figure*}[t!]
  \vspace{-2mm}
\subfloat[CTP runtime on Line]{\includegraphics[width=0.3\textwidth]{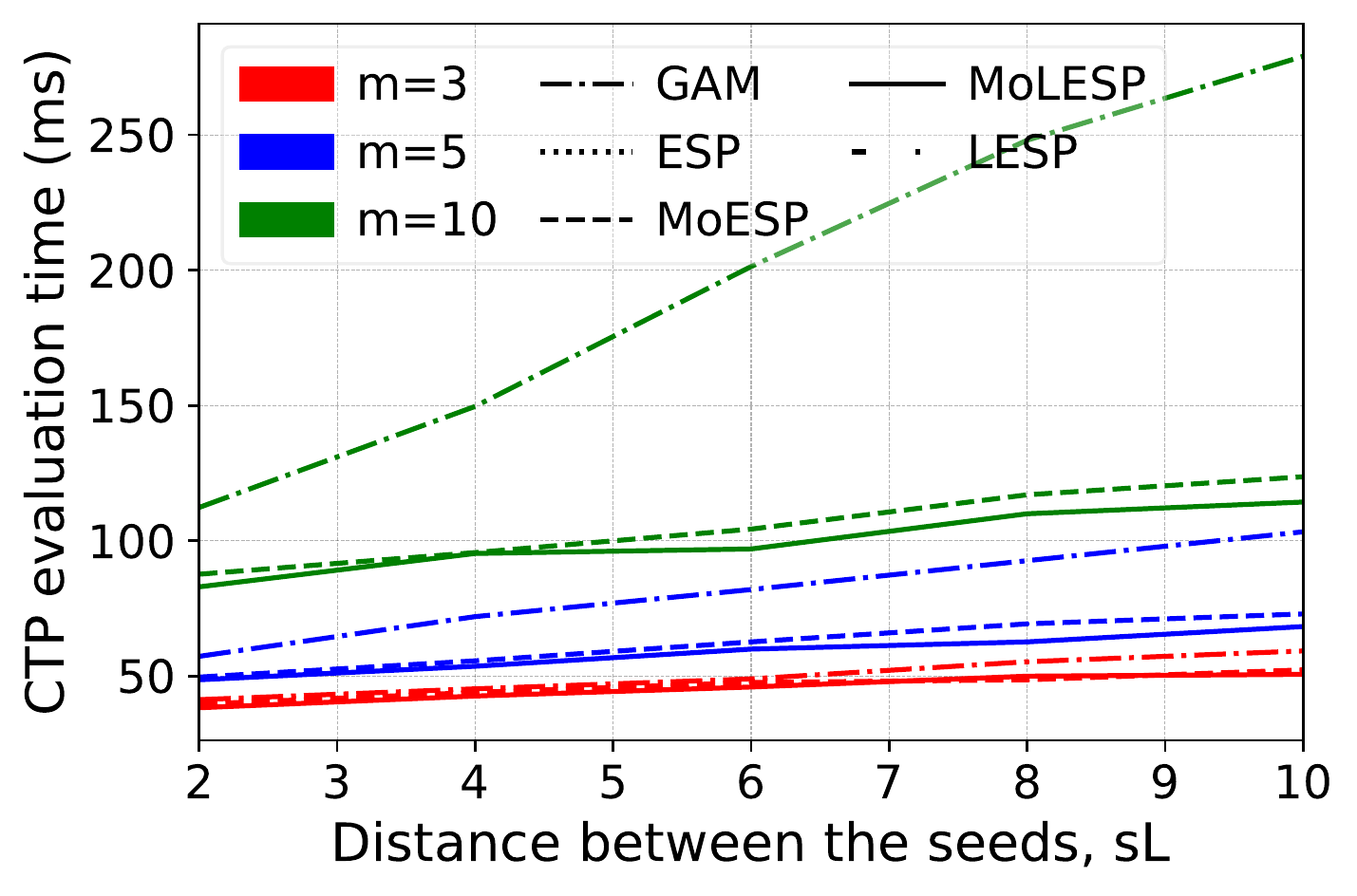}\label{fig:gam-variants-line-rt}}
\subfloat[CTP runtime on Comb]{\includegraphics[width=0.3\textwidth]{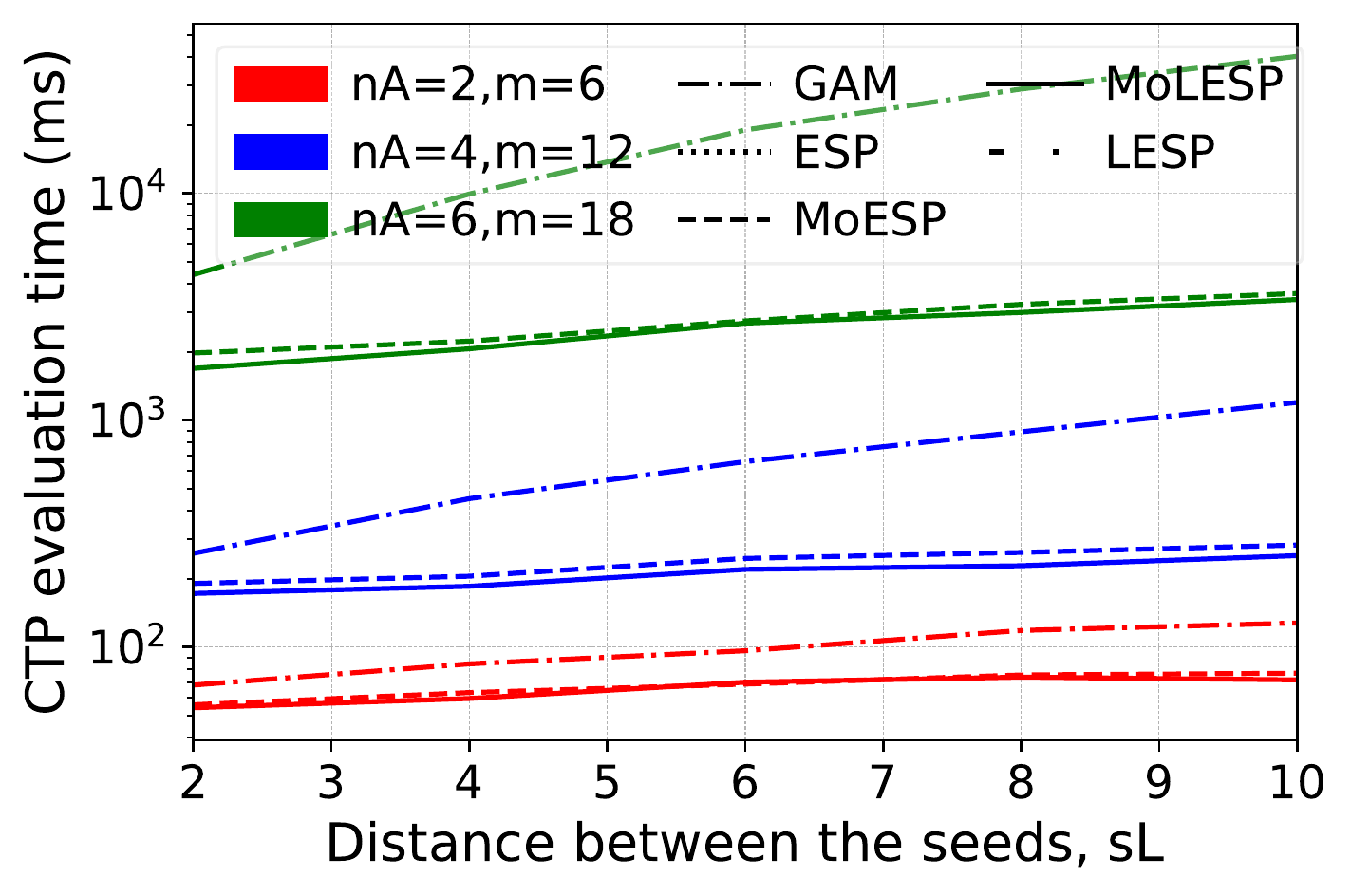}\label{fig:gam-variants-comb-rt}}
\subfloat[CTP runtime on Star]{\includegraphics[width=0.3\textwidth]{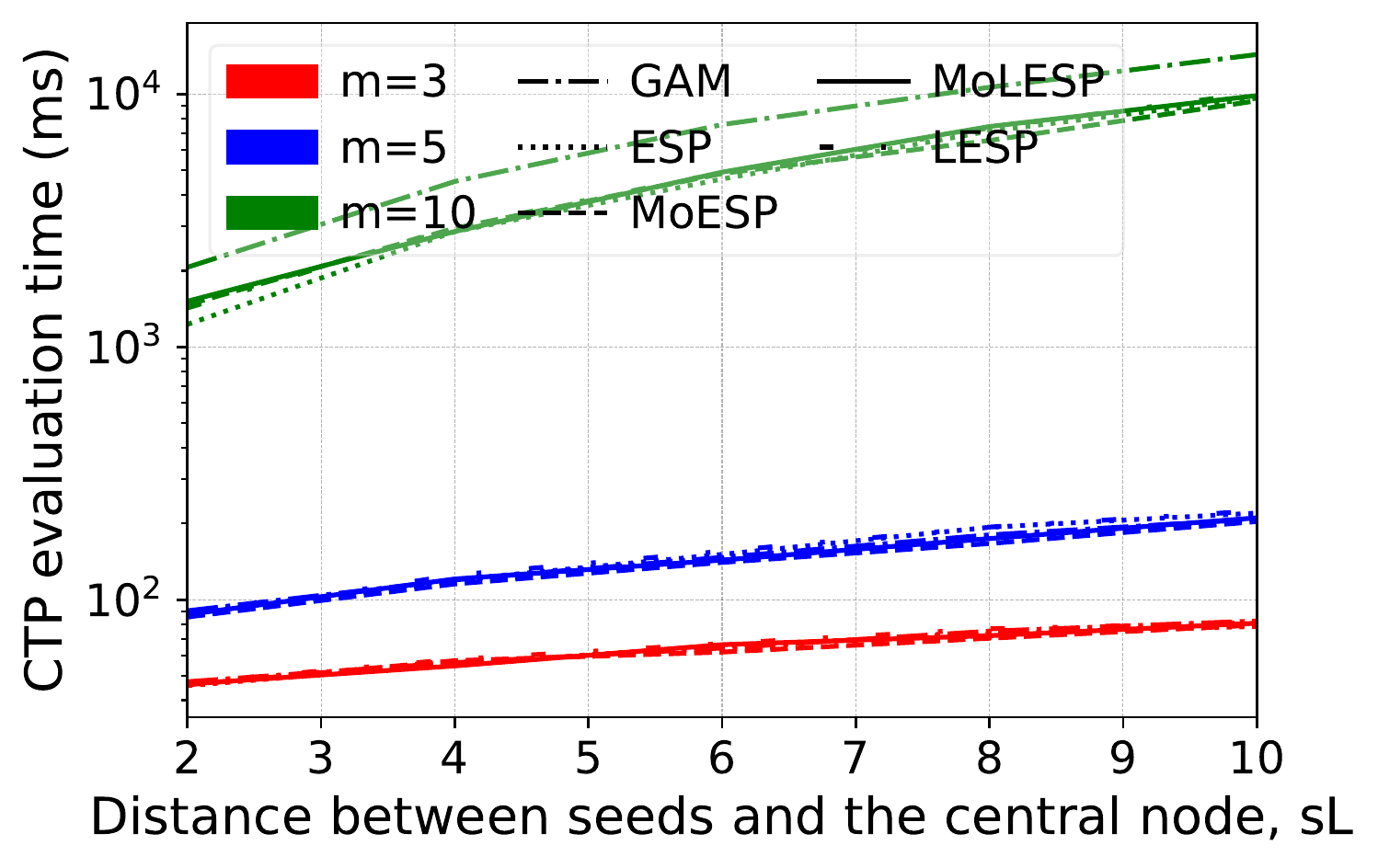}\label{fig:gam-variants-star-rt}}
\vspace{-2mm}
\subfloat[Number of provenances on Line]{\includegraphics[width=0.3\textwidth]{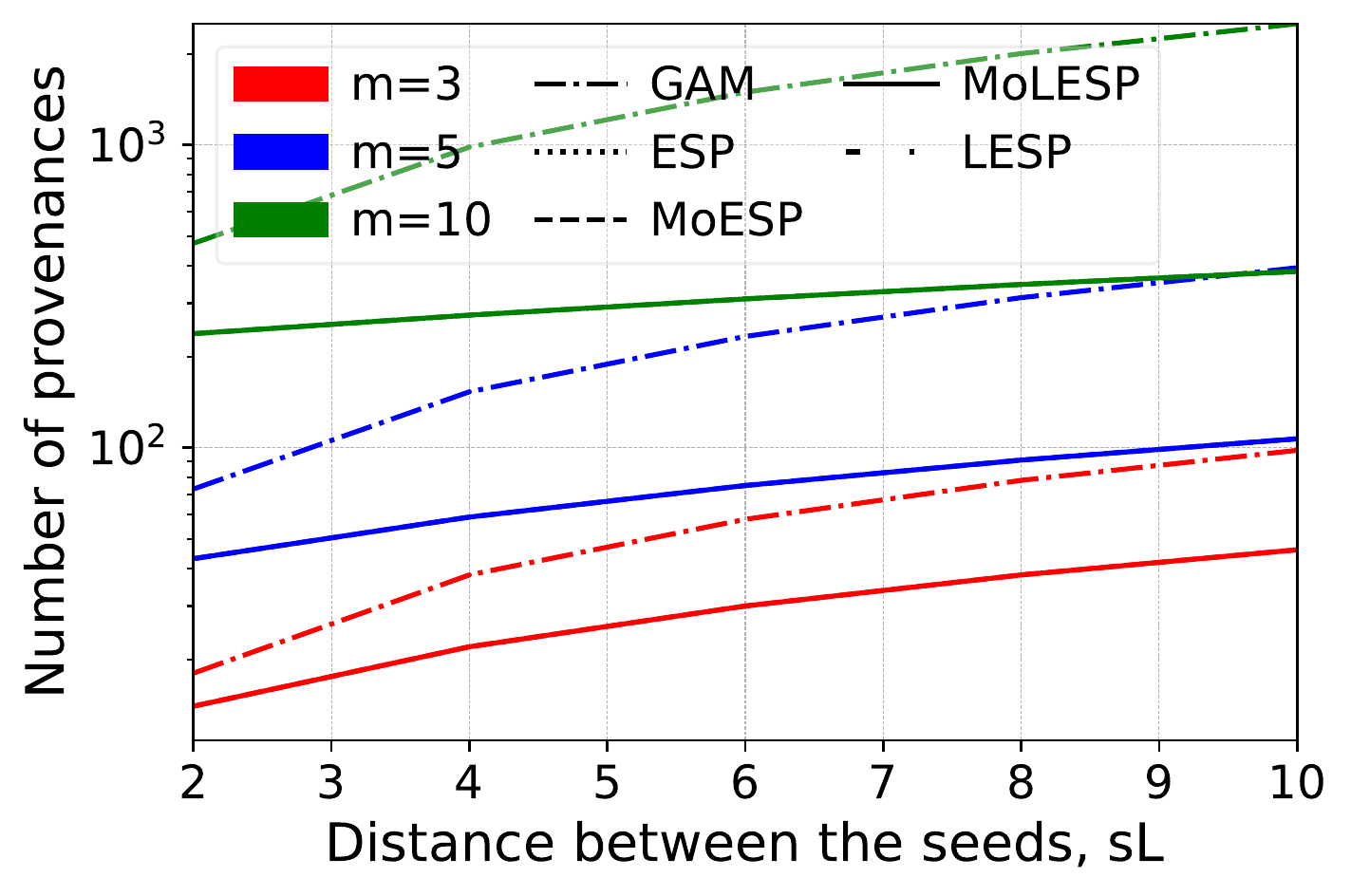}\label{fig:gam-variants-line-at}}
\subfloat[Number of provenances on Comb]{\includegraphics[width=0.3\textwidth]{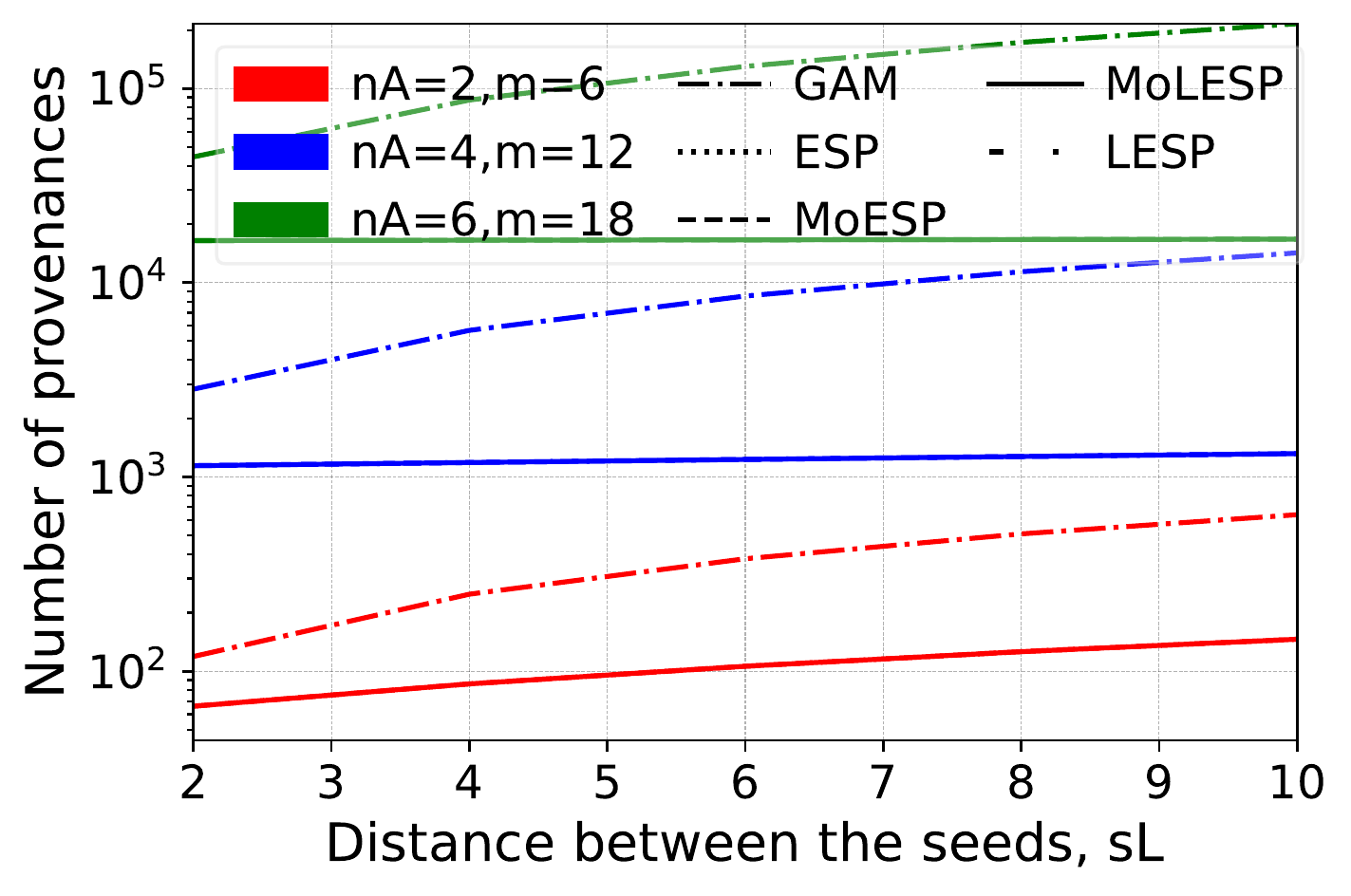}\label{fig:gam-variants-comb-at}}
\subfloat[Number of provenances on Star]{\includegraphics[width=0.3\textwidth]{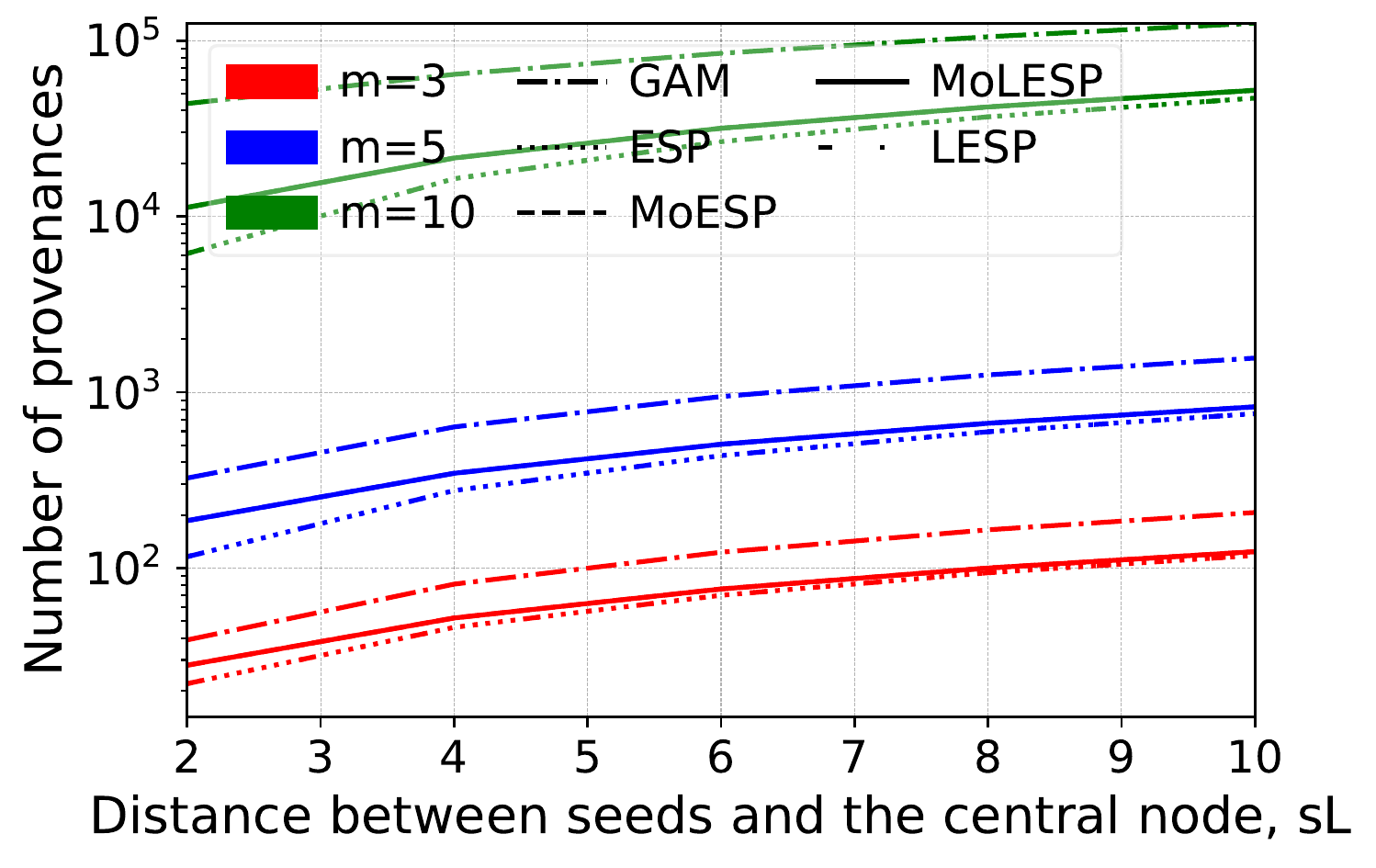}\label{fig:gam-variants-star-at}}

 \vspace{-4mm}
 \caption{Graphs for GAM variants on synthetic benchmarks. } \label{fig:gam-variants-synthetic}
 \vspace{-4mm}
 \end{figure*}

\vspace{2mm}
\mysubsubsection{Complete (baseline) algorithms} \label{sec:exp-gam-bfs}
We start by comparing the algorithms without any pruning: \bfs\ (Section \ref{sec:algo-bfs}), GAM (Section \ref{sec:algo-gam}), and the \bfs\ variants \bfsm\ and \bfsam\ (Section \ref{sec:algo-gam}), on synthetic Line, Comb and Star graphs of increasing size. We used a \textsf{\small timeout} $T$ of 10 minutes. \textbf{In all experiments with GAM and all its variants, our exploration order (queue priority) favors the smallest trees, and breaks ties arbitrarily}. Figure~\ref{fig:bfs-variants-synthetic} depicts the algorithm running time;  the color indicates the number of seed sets ($3$, $5$ or $10$), while the line pattern indicates the algorithm. {\em Missing points (or curves) denote algorithms that did not complete by the timeout}.
Note the logarithmic $y$ axes. 

Across these plots, \textbf{\bfsm\ performs worse than \bfsam}. On Line graphs, the difference is a factor $2\times$ for $m=3$ and up to $100\times$ for $m=10$. On the Comb and Star graphs,
\bfsm\ times out on the larger graphs and queries. \textbf{\bfsam\ takes even more than \bfsm}, by a factor of $15\times$, thus more executions timed out. 
\textbf{GAM is much faster} and completes execution in all cases. The reason, as explained in Section~\ref{sec:algo-bfs}, is that breadth-first algorithms waste effort by minimizing results, and may find a tree in even more different ways than GAM, since they grow from any node. Thus,  {\em we exclude breadth-first algorithms from the subsequent comparisons}.

\printIfExtVersion{\titledparagraph{Unbalanced Star}
The unbalanced star also has a $(nK,\_)$-rooted merge as the solution. Results over unbalanced star are shown in Figure \ref{fig:bfs-variants-synthetic}. \molesp\ takes more time on this dataset when compared to a balanced star but the increase in runtime follows a similar trend.  None of the \bfs\ variants finish execution before \TO\ for the $10$-seed unbalanced star. The number of trees constructed (Figure \ref{fig:bfs-variants-ustar-at}) follow a similar trend as the balanced star. \MM{Should we maybe remove this graph for \bfs\ variants? It is more insightful for the GAM variants.}}

\mysubsubsection{GAM algorithm variants}\label{sec:exp-gam-variants}
On the same graphs, we compare GAM (Section \ref{sec:algo-gam}), \esp\ (Section \ref{sec:algo-esp}), \moesp\ (Section \ref{sec:algo-moesp}), \lesp\ (Section \ref{sec:algo-lesp}) and \molesp\ (Section \ref{sec:algo-molesp}) with the same timeout. 
Figure~\ref{fig:gam-variants-synthetic} shows the algorithm running time as well as the number of provenances they built. 
In all graphs but Figure~\ref{fig:gam-variants-line-rt}, the $y$ axis is logarithmic. 
On Line and Comb graphs, \textbf{\esp\ and \lesp\ failed to find results} due to edge set pruning, as explained in Section~\ref{sec:algo-esp}, thus the corresponding curves are missing.
\moesp\ and \molesp\ build the same number of provenances on Line and Comb graphs.

The plots show, first, that edge set pruning significantly reduces the running time: \textbf{\molesp\ is faster than GAM} by a factor ranging from $1.3\times$ (Line graphs) to $15\times$ (Comb graphs, $nA$$=$$6$, $m$$=$$18$). Second, on the Star graphs, where the limited edge-set pruning (Section~\ref{sec:algo-lesp}) applies, the performance difference between \moesp\ and \molesp\ is small. This shows that \textbf{the extra cost incurred by \lesp\ and \molesp, which limit or compensate for edge-set pruning} (by injecting more trees), \textbf{is worth paying for the completeness guarantees} of \molesp. Overall, the algorithm running times closely track the numbers of built provenances, further highlighting the interest of controlling the latter through pruning.

\mysubsubsection{Comparison with QGSTP on real-world data}\label{sec:exp-qgstp}

We now compare the winner of the above comparisons, namely \molesp, with QGSTP~\cite{qgstp@www21} on the $18$M edges DBPedia dataset and $312$ CTPs used in their evaluation. Among these, $83$ CTPs  (respectively, $98$, $85$, $38$, $8$) have $2$ (respectively, $3$, $4$, $5$, $6$) seed sets. 
To align with QGSTP, we added a \textsf{\small UNI} filter (unidirectional exploration only), and \textsf{\small LIMIT 1} to stop after the first result. 
Each QGSTP returned result is such that Property~\ref{prop:restricted-nps} ensures \molesp\ finds it. 
Figure \ref{fig:gam-qgstp-dbpedia} shows the average runtimes grouped by $m$. GAM is faster than QGSTP for $m$$\leq$$5$, but timed-out for the $8$ CTPs with $m$$=$$6$. \molesp\ is about $6$-$7\times$ faster than QGSTP for all $m$ values, and scales well as $m$ increases. Thus, \textbf{\molesp\ is competitive also on large real-world graphs and queries}.

\mysubsection{Extended query evaluation}
\label{sec:exp-eql}
\vspace{2mm} 
\mysubsubsection{Synthetic queries on CDF benchmark}\label{sec:exp-cdf} 
We now compare our EQL query evaluation system with the graph query baselines, on our CDF graphs (Section~\ref{sec:datasets}) generated with $m$$\in$$\{2$,$3\}$, $S_L$$\in$$\{3$,$6\}$, 
$18$K to $2.4$M edges, leading to $2$K up to $200$K results ($N_L$), respectively. We used $T$$=$$15$ minutes.
As explained in Section~\ref{sec:language-construct}, the paths returned by the baselines, which we ``stitch'' for  $m$$=$$3$, semantically differ from CTP results;
the baselines' reported time {\em do not include the time to minimize nor deduplicate their results}. 
 
For $m$$=$$2$, Figure \ref{fig:comparison-dense-forest-m2} shows that all systems scale linearly in the input size (note the logarithmic time axis). {\em For each system, the lower curve is on graphs with $S_L$$=$$3$, while the upper curve is on graphs with $S_L$$=$$6$ (these graphs are larger, thus curves go farther at right). All missing points correspond to time-out.} JEDI succeeded only on the smallest graph, Neo4j timed-out on all. 
Virtuoso-SPARQL is the fastest, closely followed by Virtuoso-SQL; they are both  {\em unidirectional, require the edge labels}, and {\em do not return paths}.  Unidirectional \molesp, which we included to compare with unidirectional baselines, is slower by approximately $3\times$ only. 
JEDI is slower than \molesp\ by $10^2\times$ on the smallest graph, and timed-out on the others. Postgres is faster than JEDI, yet at least $10\times$ slower than \molesp.
\textbf{\molesp\ is the only feasible bidirectional algorithm}; it completes in less than $2$ minutes on the largest graph with $2.4$M edges.

Figure \ref{fig:comparison-dense-forest-m3} shows similar results for $m$$=$$3$.  
Postgres timed-out in all cases. 
Virtuoso-SPARQL is $7\times$ faster than Virtuoso-SQL; both return {\em non-minimal, duplicate results}. \textsf{\small UNI}-\molesp\ outperforms every system, while \emph{also returning connecting trees}. Note that the {\em bidirectional} \molesp\ found about $7$$\times$ more results than the $N_L$ expected ones, by also connecting bottom leaves {\em without a common parent} through their grandparent node; 
these results are filtered by the join between the BGPs and the CTP (Section \ref{sec:query-evaluation}). 
Despite the much larger search space due to bidirectionality, \molesp\ scales well with the size of the graph.


 \begin{figure}[t!]
   \centering
   \includegraphics[width=0.3\textwidth]{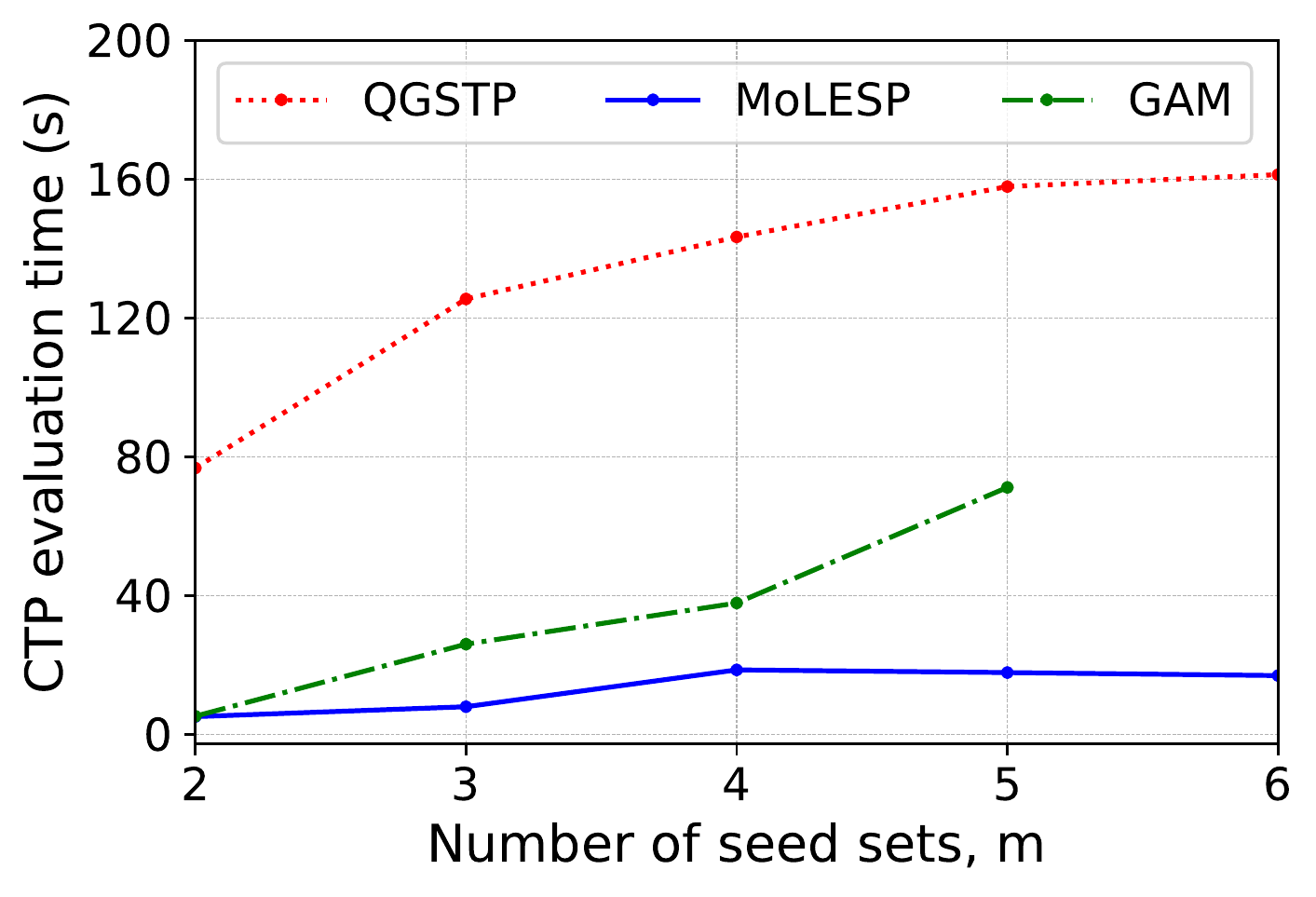}
   \vspace{-5mm}
\caption{GAM and \molesp\ vs. QGSTP~\cite{qgstp@www21} on DBPedia.}\label{fig:gam-qgstp-dbpedia}
\end{figure}

\begin{figure}[t!]
  \includegraphics[width=0.5\textwidth]{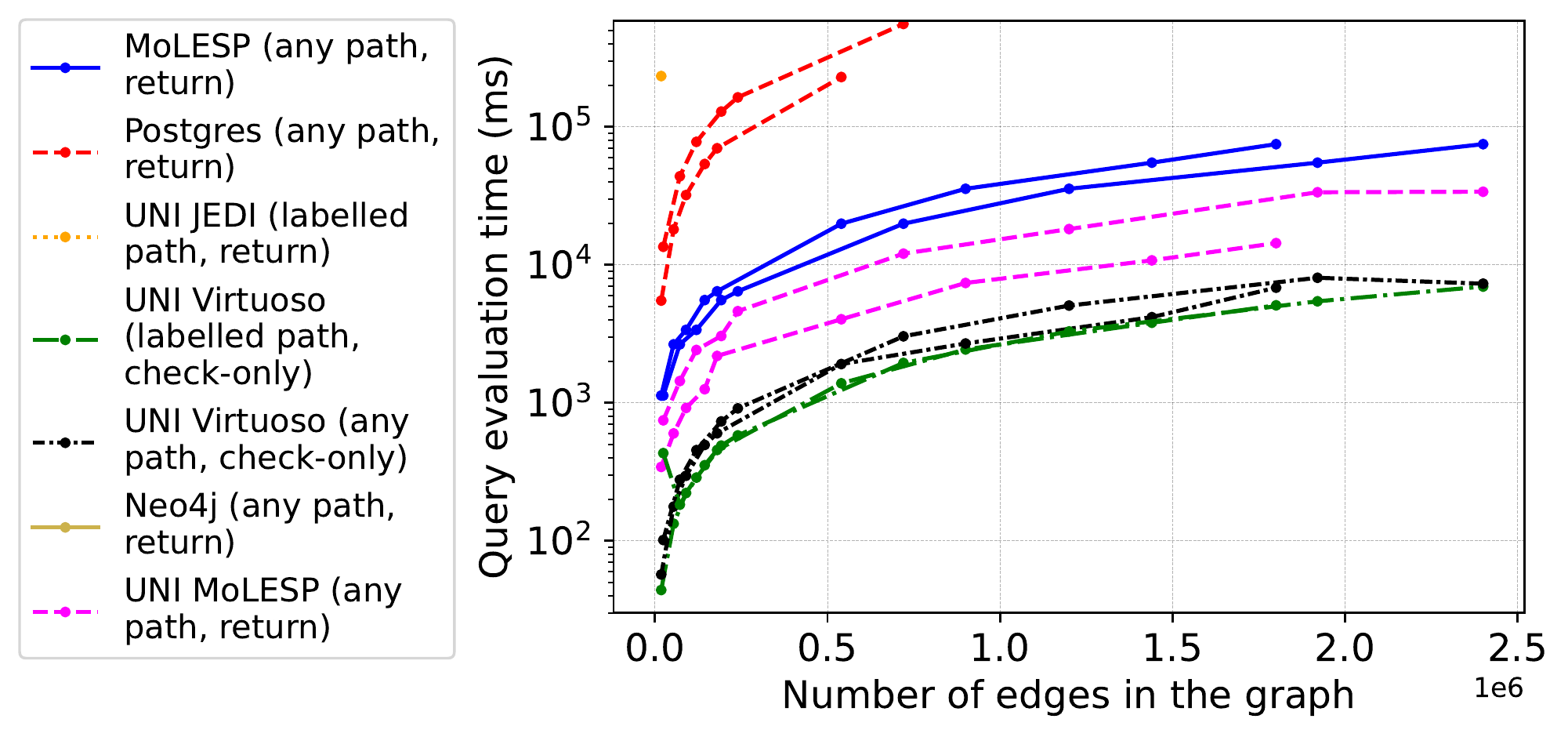}
  \vspace{-9mm}
\caption{CDF benchmark performance for $m$$=$$2$, $S_L$$\in$$\{3$,$6\}$.\label{fig:comparison-dense-forest-m2}}
\end{figure}

\begin{figure}[t!]
  \includegraphics[width=0.5\textwidth]{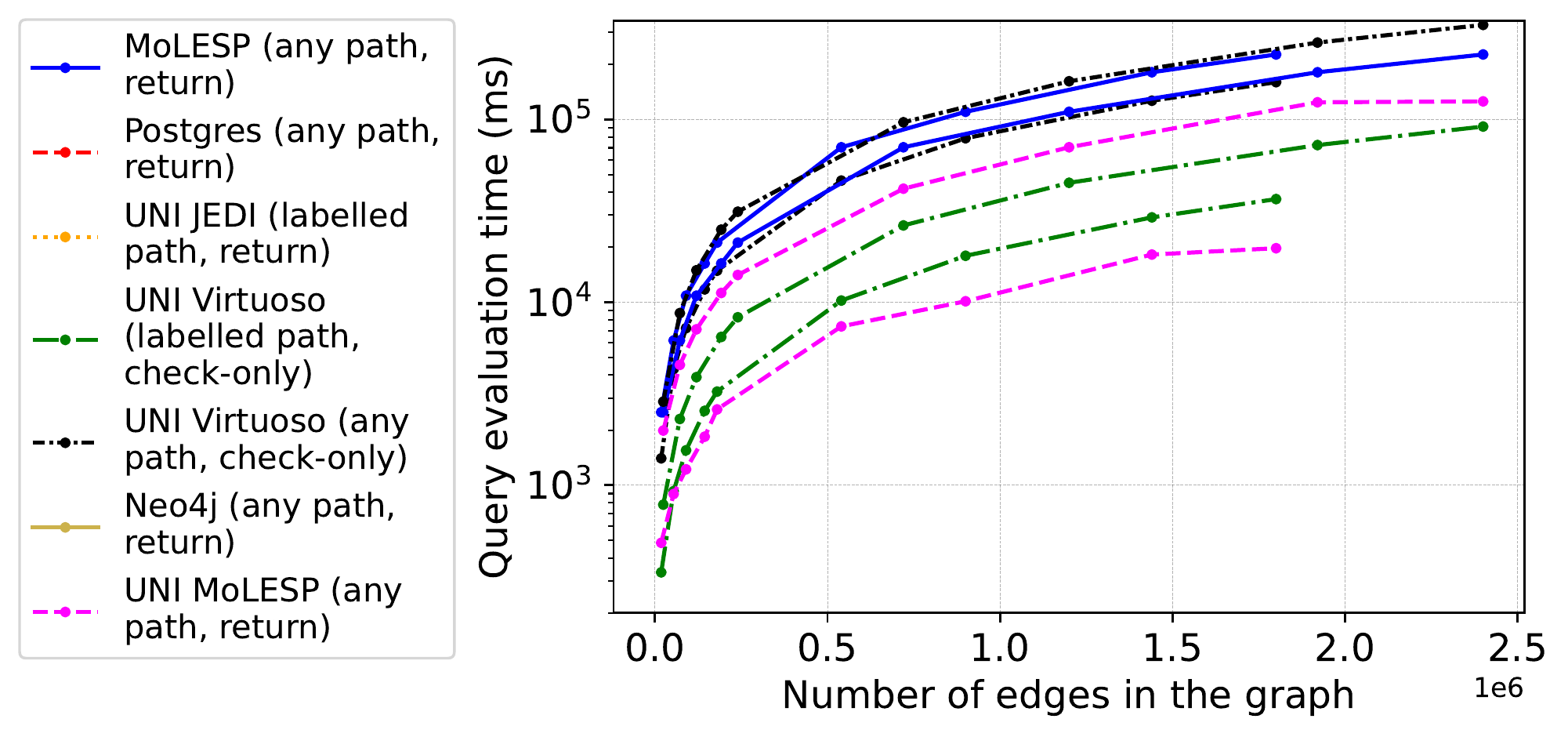}
  \vspace{-9mm}
\caption{CDF benchmark performance for $m$$=$$3$, $S_L$$\in$$\{3$,$6\}$.\label{fig:comparison-dense-forest-m3}}
\end{figure}

\mysubsubsection{Comparison with JEDI on real-world data}\label{sec:exp-yago}
JEDI~\cite{jedi@vldb2018} used a set of (unidirectional, label-constrained) SPARQL 1.1 queries over YAGO3.
  Table \ref{tab:comparison-yago} shows the queries' characteristics. 
  We compare \molesp\ similarly constrained (\textsf{\small UNI} and \textsf{\small LABEL}), on these queries, with JEDI, Virtuoso and Neo4j (Postgres timed-out on all). 
  Query $J_2$ has one very large seed set, while query $J_3$ has a $\nodes$ seed set. {\em On queries $J_2$ and $J_3$, \molesp\ timed out}. Thus, we applied the optimizations described in Section~\ref{sec:survival}, which enabled it to perform as shown. 
  Virtuoso-SPARQL completed query $J_1$, then ran out of memory.
  Compared with JEDI, our query evaluation engine is $2\times$ faster on $J_1$, close on $J_2$, and around  $3\times$ slower on $J_3$. 
   \molesp\ took around 30\% of the total time, the rest being spent by Postgres in the BGP evaluation and final joins.
This shows that \textbf{the optimizations described in Section~\ref{sec:survival} make \molesp\ robust also to large seed sets.}
 

\begin{table}
\centering
\scalebox{0.85}{\begin{tabular}{|p{4cm}| r|r|r|r|} 
 \hline
 {Query} & {JEDI} & {\molesp\ } & {Virtuoso}& {Neo4j} \\ 
 \hline\hline
 $J_1$: $3$ BGPs, $2$ CTPs & $3.9$ & $1.9$ & $0.2$& \TO\ \\ 
 \hline
 $J_2$: $2$ BGPs, $1$ CTP, large seed set & $0.9$ & $1$ & OOM & \TO\ \\
 \hline
 $J_3$: $1$ CTP,  $\nodes$ seed set & $0.75$ & $2.3$ & OOM &  $1.27$ \\
 \hline
\end{tabular}}
\caption{Query evaluation times (seconds) on YAGO3 dataset.\label{tab:comparison-yago}}
\vspace{-4mm}
\end{table}


%% file: related-work.tex
\mysection{Related work and perspectives}\label{sec:related-work}

We focused on extending a \textbf{graph query language}, such as SPARQL \cite{sparql11}, Cypher \cite{neo4j-cypher} or GraphQL \cite{graphql}, with {\em connecting tree patterns} (CTPs) that they currently do not support (our requirement (\textbf{R1}) from Section~\ref{sec:introduction}).  
Specifically, SPARQL 1.1 property paths
($i$)~allow to {\em check} that some paths connect two nodes, not to {\em return} the path(s);
($ii$)~do not allow searching for {\em arbitrary} paths (users have to specify a regular expression);
($iii$)~are restricted to {\em unidirectional} paths only. Some PG query languages such as Neo4j's Cypher lift these restrictions, however, its implementation does not scale (Section~\ref{sec:exp-cdf}) \cite{neo4j-blog}.  RPQProv \cite{dey@edbt2013} uses recursive SQL to return path labels;  JEDI \cite{jedi@iswc2018,jedi@vldb2018} builds over SPARQL 1.1 by returning all unidirectional paths. 
Many works focus on finding label-constrained paths between nodes \cite{sparql2l, gubichev@webdb2011, gubichev@grades2013, fletcher@edbt2016,  yakovets@sigmod2016, valstar@sigmod2017, wadhwa@sigmod2019, kuijpers@edbt2021, peng@vldb2022, rpq@icde2022, na@icde2022}, typically by using precomputed indexes or sketches. In our CTP evaluation algorithm, an index could be integrated by ``reading from it'' paths (or subtrees) on which to \Grow\ and \Merge. 
{\em Our CTPs extend finding paths, to finding trees that connect an arbitrary number of seed sets ($m$$\geq$$3$), traversing edges in any direction} by default; we guarantee completeness for $m$$\leq$$3$ and finding a large set of results for arbitrary $m$. As we explained (Section~\ref{sec:language-construct}), path stitching leads to different results, which may require deduplication and minimization. 

The CTP evaluation problem is directly related to \textbf{keyword search in (semi-)structured data},
addressed in many algorithms, some of which are surveyed in \cite{kssurvey,coffman@tkde2014}.
These prior studies differ from ours as follows: ($i$)~\cite{DBXplorer, discover, efficient@vldb2003, spark, spark2, oliveira@icde2018, xrank, hristidis@icde2003, tran@icde2009, star} are schema-dependent;
($ii$)~\cite{objectrank,tran@icde2009, scalable-ks} assume available a compact summary of the graph;
($iii$)~\cite{banks-2, dpbf, blinks, li@sigmod2016} depend heavily on their score functions for pruning the search, particularly to approximate the best result \cite{dpbf, li@sigmod2016} or return only top-$k$ results \cite{ease, spark,blinks, banks-1, parallel-ks};
($iv$)~\cite{banks-1, banks-1-demo, DBXplorer, discover, blinks} are only unidirectional. For these reasons, they fail to meet our requirements (\textbf{R2}) to (\textbf{R5}) as outlined in Section \ref{sec:introduction}.





The Java-based GAM algorithm used in this work~\cite{gam-inf-sys-2022} was sped up by up to $100\times$ in a multi-threaded, C++ version~\cite{anadiotis:hal-03337650}. \molesp\ brings new, orthogonal, optimizations, and novel guarantees. 

Our future work includes developing adaptive EQL optimization and execution strategies and applying it to graph exploration for investigative journalism. 


%% file: ms.bbl

\begin{thebibliography}{46}


\ifx \showCODEN    \undefined \def \showCODEN     #1{\unskip}     \fi
\ifx \showDOI      \undefined \def \showDOI       #1{#1}\fi
\ifx \showISBNx    \undefined \def \showISBNx     #1{\unskip}     \fi
\ifx \showISBNxiii \undefined \def \showISBNxiii  #1{\unskip}     \fi
\ifx \showISSN     \undefined \def \showISSN      #1{\unskip}     \fi
\ifx \showLCCN     \undefined \def \showLCCN      #1{\unskip}     \fi
\ifx \shownote     \undefined \def \shownote      #1{#1}          \fi
\ifx \showarticletitle \undefined \def \showarticletitle #1{#1}   \fi
\ifx \showURL      \undefined \def \showURL       {\relax}        \fi
\providecommand\bibfield[2]{#2}
\providecommand\bibinfo[2]{#2}
\providecommand\natexlab[1]{#1}
\providecommand\showeprint[2][]{arXiv:#2}

\bibitem[\protect\citeauthoryear{Aditya, Bhalotia, Chakrabarti, Hulgeri, Nakhe,
  Parag, and Sudarshan}{Aditya et~al\mbox{.}}{2002}]%
        {banks-1-demo}
\bibfield{author}{\bibinfo{person}{B. Aditya}, \bibinfo{person}{Gaurav
  Bhalotia}, \bibinfo{person}{Soumen Chakrabarti}, \bibinfo{person}{Arvind
  Hulgeri}, \bibinfo{person}{Charuta Nakhe}, \bibinfo{person}{Parag}, {and}
  \bibinfo{person}{S. Sudarshan}.} \bibinfo{year}{2002}\natexlab{}.
\newblock \showarticletitle{{BANKS:} Browsing and Keyword Searching in
  Relational Databases}. In \bibinfo{booktitle}{\emph{Proceedings of 28th
  International Conference on Very Large Data Bases, {VLDB} 2002, Hong Kong,
  August 20-23, 2002}}. \bibinfo{pages}{1083--1086}.
\newblock
\urldef\tempurl%
\url{https://doi.org/10.1016/B978-155860869-6/50114-1}
\showDOI{\tempurl}


\bibitem[\protect\citeauthoryear{Aebeloe, Montoya, Setty, and Hose}{Aebeloe
  et~al\mbox{.}}{2018a}]%
        {jedi@vldb2018}
\bibfield{author}{\bibinfo{person}{Christian Aebeloe},
  \bibinfo{person}{Gabriela Montoya}, \bibinfo{person}{Vinay Setty}, {and}
  \bibinfo{person}{Katja Hose}.} \bibinfo{year}{2018}\natexlab{a}.
\newblock \showarticletitle{Discovering Diversified Paths in Knowledge Bases}.
\newblock \bibinfo{journal}{\emph{Proc. {VLDB} Endow.}} \bibinfo{volume}{11},
  \bibinfo{number}{12} (\bibinfo{year}{2018}), \bibinfo{pages}{2002--2005}.
\newblock
\urldef\tempurl%
\url{https://doi.org/10.14778/3229863.3236245}
\showDOI{\tempurl}
\newblock
\shownote{Code available at: http://qweb.cs.aau.dk/jedi/.}


\bibitem[\protect\citeauthoryear{Aebeloe, Setty, Montoya, and Hose}{Aebeloe
  et~al\mbox{.}}{2018b}]%
        {jedi@iswc2018}
\bibfield{author}{\bibinfo{person}{Christian Aebeloe}, \bibinfo{person}{Vinay
  Setty}, \bibinfo{person}{Gabriela Montoya}, {and} \bibinfo{person}{Katja
  Hose}.} \bibinfo{year}{2018}\natexlab{b}.
\newblock \showarticletitle{Top-K Diversification for Path Queries in Knowledge
  Graphs}. In \bibinfo{booktitle}{\emph{Proceedings of the {ISWC} 2018 Posters
  {\&} Demonstrations, Industry and Blue Sky Ideas Tracks co-located with 17th
  International Semantic Web Conference {(ISWC} 2018), Monterey, USA, October
  8th - to - 12th, 2018}} \emph{(\bibinfo{series}{{CEUR} Workshop
  Proceedings})}, \bibfield{editor}{\bibinfo{person}{Marieke van Erp},
  \bibinfo{person}{Medha Atre}, \bibinfo{person}{Vanessa L{\'{o}}pez},
  \bibinfo{person}{Kavitha Srinivas}, {and} \bibinfo{person}{Carolina Fortuna}}
  (Eds.), Vol.~\bibinfo{volume}{2180}. \bibinfo{publisher}{CEUR-WS.org}.
\newblock
\urldef\tempurl%
\url{http://ceur-ws.org/Vol-2180/paper-01.pdf}
\showURL{%
\tempurl}


\bibitem[\protect\citeauthoryear{Agrawal, Chaudhuri, and Das}{Agrawal
  et~al\mbox{.}}{2002}]%
        {DBXplorer}
\bibfield{author}{\bibinfo{person}{Sanjay Agrawal}, \bibinfo{person}{Surajit
  Chaudhuri}, {and} \bibinfo{person}{Gautam Das}.}
  \bibinfo{year}{2002}\natexlab{}.
\newblock \showarticletitle{DBXplorer: {A} System for Keyword-Based Search over
  Relational Databases}. In \bibinfo{booktitle}{\emph{Proceedings of the 18th
  International Conference on Data Engineering, San Jose, CA, USA, February 26
  - March 1, 2002}}, \bibfield{editor}{\bibinfo{person}{Rakesh Agrawal} {and}
  \bibinfo{person}{Klaus~R. Dittrich}} (Eds.). \bibinfo{publisher}{{IEEE}
  Computer Society}, \bibinfo{pages}{5--16}.
\newblock
\urldef\tempurl%
\url{https://doi.org/10.1109/ICDE.2002.994693}
\showDOI{\tempurl}


\bibitem[\protect\citeauthoryear{Anadiotis, Balalau, Bouganim, Chimienti,
  Galhardas, Haddad, Horel, Manolescu, and Youssef}{Anadiotis
  et~al\mbox{.}}{2021}]%
        {anadiotis:hal-03337650}
\bibfield{author}{\bibinfo{person}{Angelos-Christos Anadiotis},
  \bibinfo{person}{Oana Balalau}, \bibinfo{person}{Th{\'e}o Bouganim},
  \bibinfo{person}{Francesco Chimienti}, \bibinfo{person}{Helena Galhardas},
  \bibinfo{person}{Mhd~Yamen Haddad}, \bibinfo{person}{St{\'e}phane Horel},
  \bibinfo{person}{Ioana Manolescu}, {and} \bibinfo{person}{Youssr Youssef}.}
  \bibinfo{year}{2021}\natexlab{}.
\newblock \showarticletitle{{Empowering Investigative Journalism with
  Graph-based Heterogeneous Data Management}}.
\newblock \bibinfo{journal}{\emph{{Bulletin of the Technical Committee on Data
  Engineering}}} (\bibinfo{date}{Sept.} \bibinfo{year}{2021}).
\newblock
\urldef\tempurl%
\url{https://hal.archives-ouvertes.fr/hal-03337650}
\showURL{%
\tempurl}


\bibitem[\protect\citeauthoryear{Anadiotis, Balalau, Concei{\c{c}}{\~{a}}o,
  Galhardas, Haddad, Manolescu, Merabti, and You}{Anadiotis
  et~al\mbox{.}}{2022}]%
        {gam-inf-sys-2022}
\bibfield{author}{\bibinfo{person}{Angelos{-}Christos~G. Anadiotis},
  \bibinfo{person}{Oana Balalau}, \bibinfo{person}{Catarina
  Concei{\c{c}}{\~{a}}o}, \bibinfo{person}{Helena Galhardas},
  \bibinfo{person}{Mhd~Yamen Haddad}, \bibinfo{person}{Ioana Manolescu},
  \bibinfo{person}{Tayeb Merabti}, {and} \bibinfo{person}{Jingmao You}.}
  \bibinfo{year}{2022}\natexlab{}.
\newblock \showarticletitle{Graph integration of structured, semistructured and
  unstructured data for data journalism}.
\newblock \bibinfo{journal}{\emph{Inf. Syst.}}  \bibinfo{volume}{104}
  (\bibinfo{year}{2022}), \bibinfo{pages}{101846}.
\newblock
\urldef\tempurl%
\url{https://doi.org/10.1016/j.is.2021.101846}
\showDOI{\tempurl}


\bibitem[\protect\citeauthoryear{Anyanwu, Maduko, and Sheth}{Anyanwu
  et~al\mbox{.}}{2007}]%
        {sparql2l}
\bibfield{author}{\bibinfo{person}{Kemafor Anyanwu}, \bibinfo{person}{Angela
  Maduko}, {and} \bibinfo{person}{Amit~P. Sheth}.}
  \bibinfo{year}{2007}\natexlab{}.
\newblock \showarticletitle{{SPARQ2L:} towards support for subgraph extraction
  queries in rdf databases}. In \bibinfo{booktitle}{\emph{Proceedings of the
  16th International Conference on World Wide Web, {WWW} 2007, Banff, Alberta,
  Canada, May 8-12, 2007}}. \bibinfo{pages}{797--806}.
\newblock
\urldef\tempurl%
\url{https://doi.org/10.1145/1242572.1242680}
\showDOI{\tempurl}


\bibitem[\protect\citeauthoryear{Arroyuelo, Hogan, Navarro, and
  Rojas{-}Ledesma}{Arroyuelo et~al\mbox{.}}{2022}]%
        {rpq@icde2022}
\bibfield{author}{\bibinfo{person}{Diego Arroyuelo}, \bibinfo{person}{Aidan
  Hogan}, \bibinfo{person}{Gonzalo Navarro}, {and} \bibinfo{person}{Javiel
  Rojas{-}Ledesma}.} \bibinfo{year}{2022}\natexlab{}.
\newblock \showarticletitle{Time- and Space-Efficient Regular Path Queries on
  Graphs}.
\newblock  (\bibinfo{year}{2022}).
\newblock


\bibitem[\protect\citeauthoryear{Balmin, Hristidis, and
  Papakonstantinou}{Balmin et~al\mbox{.}}{2004}]%
        {objectrank}
\bibfield{author}{\bibinfo{person}{Andrey Balmin}, \bibinfo{person}{Vagelis
  Hristidis}, {and} \bibinfo{person}{Yannis Papakonstantinou}.}
  \bibinfo{year}{2004}\natexlab{}.
\newblock \showarticletitle{ObjectRank: Authority-Based Keyword Search in
  Databases}. In \bibinfo{booktitle}{\emph{(e)Proceedings of the Thirtieth
  International Conference on Very Large Data Bases, {VLDB} 2004, Toronto,
  Canada, August 31 - September 3 2004}}. \bibinfo{pages}{564--575}.
\newblock
\urldef\tempurl%
\url{https://doi.org/10.1016/B978-012088469-8.50051-6}
\showDOI{\tempurl}


\bibitem[\protect\citeauthoryear{Bhalotia, Hulgeri, Nakhe, Chakrabarti, and
  Sudarshan}{Bhalotia et~al\mbox{.}}{2002}]%
        {banks-1}
\bibfield{author}{\bibinfo{person}{Gaurav Bhalotia}, \bibinfo{person}{Arvind
  Hulgeri}, \bibinfo{person}{Charuta Nakhe}, \bibinfo{person}{Soumen
  Chakrabarti}, {and} \bibinfo{person}{S. Sudarshan}.}
  \bibinfo{year}{2002}\natexlab{}.
\newblock \showarticletitle{Keyword Searching and Browsing in Databases using
  {BANKS}}. In \bibinfo{booktitle}{\emph{Proceedings of the 18th International
  Conference on Data Engineering, San Jose, CA, USA, February 26 - March 1,
  2002}}. \bibinfo{pages}{431--440}.
\newblock
\urldef\tempurl%
\url{https://doi.org/10.1109/ICDE.2002.994756}
\showDOI{\tempurl}


\bibitem[\protect\citeauthoryear{Bowman}{Bowman}{2022}]%
        {neo4j-blog}
\bibfield{author}{\bibinfo{person}{Andrew Bowman}.}
  \bibinfo{year}{2022}\natexlab{}.
\newblock \showarticletitle{{Tuning Cypher queries by understanding
  cardinality}}.
\newblock  (\bibinfo{year}{2022}).
\newblock
\urldef\tempurl%
\url{https://neo4j.com/developer/kb/understanding-cypher-cardinality/#_distinct_nodes_from_variable_length_paths}
\showURL{%
\tempurl}


\bibitem[\protect\citeauthoryear{Coffman and Weaver}{Coffman and
  Weaver}{2014}]%
        {coffman@tkde2014}
\bibfield{author}{\bibinfo{person}{Joel Coffman} {and}
  \bibinfo{person}{Alfred~C. Weaver}.} \bibinfo{year}{2014}\natexlab{}.
\newblock \showarticletitle{An Empirical Performance Evaluation of Relational
  Keyword Search Techniques}.
\newblock \bibinfo{journal}{\emph{{IEEE} Trans. Knowl. Data Eng.}}
  \bibinfo{volume}{26}, \bibinfo{number}{1} (\bibinfo{year}{2014}),
  \bibinfo{pages}{30--42}.
\newblock
\urldef\tempurl%
\url{https://doi.org/10.1109/TKDE.2012.228}
\showDOI{\tempurl}


\bibitem[\protect\citeauthoryear{Consortium}{Consortium}{2013}]%
        {sparql11}
\bibfield{author}{\bibinfo{person}{WWW Consortium}.}
  \bibinfo{year}{2013}\natexlab{}.
\newblock \showarticletitle{{SPARQL 1.1}}.
\newblock  (\bibinfo{year}{2013}).
\newblock
\urldef\tempurl%
\url{https://www.w3.org/TR/sparql11-overview/}
\showURL{%
\tempurl}


\bibitem[\protect\citeauthoryear{de~Oliveira, da~Silva, de~Moura, and
  Rodrigues}{de~Oliveira et~al\mbox{.}}{2018}]%
        {oliveira@icde2018}
\bibfield{author}{\bibinfo{person}{Pericles de Oliveira},
  \bibinfo{person}{Altigran~S. da Silva}, \bibinfo{person}{Edleno~Silva de
  Moura}, {and} \bibinfo{person}{Rosiane Rodrigues}.}
  \bibinfo{year}{2018}\natexlab{}.
\newblock \showarticletitle{Match-Based Candidate Network Generation for
  Keyword Queries over Relational Databases}. In \bibinfo{booktitle}{\emph{34th
  {IEEE} International Conference on Data Engineering, {ICDE} 2018, Paris,
  France, April 16-19, 2018}}. \bibinfo{pages}{1344--1347}.
\newblock
\urldef\tempurl%
\url{https://doi.org/10.1109/ICDE.2018.00146}
\showDOI{\tempurl}


\bibitem[\protect\citeauthoryear{Dey, Cuevas{-}Vicentt{\'{\i}}n, K{\"{o}}hler,
  Gribkoff, Wang, and Lud{\"{a}}scher}{Dey et~al\mbox{.}}{2013}]%
        {dey@edbt2013}
\bibfield{author}{\bibinfo{person}{Saumen~C. Dey},
  \bibinfo{person}{V{\'{\i}}ctor Cuevas{-}Vicentt{\'{\i}}n},
  \bibinfo{person}{Sven K{\"{o}}hler}, \bibinfo{person}{Eric Gribkoff},
  \bibinfo{person}{Michael Wang}, {and} \bibinfo{person}{Bertram
  Lud{\"{a}}scher}.} \bibinfo{year}{2013}\natexlab{}.
\newblock \showarticletitle{On implementing provenance-aware regular path
  queries with relational query engines}. In \bibinfo{booktitle}{\emph{Joint
  2013 {EDBT/ICDT} Conferences, {EDBT/ICDT} '13, Genoa, Italy, March 22, 2013,
  Workshop Proceedings}}, \bibfield{editor}{\bibinfo{person}{Giovanna
  Guerrini}} (Ed.). \bibinfo{publisher}{{ACM}}, \bibinfo{pages}{214--223}.
\newblock
\urldef\tempurl%
\url{https://doi.org/10.1145/2457317.2457353}
\showDOI{\tempurl}


\bibitem[\protect\citeauthoryear{Ding, Yu, Wang, Qin, Zhang, and Lin}{Ding
  et~al\mbox{.}}{2007}]%
        {dpbf}
\bibfield{author}{\bibinfo{person}{Bolin Ding}, \bibinfo{person}{Jeffrey~Xu
  Yu}, \bibinfo{person}{Shan Wang}, \bibinfo{person}{Lu Qin},
  \bibinfo{person}{Xiao Zhang}, {and} \bibinfo{person}{Xuemin Lin}.}
  \bibinfo{year}{2007}\natexlab{}.
\newblock \showarticletitle{Finding Top-k Min-Cost Connected Trees in
  Databases}.
\newblock  (\bibinfo{year}{2007}), \bibinfo{pages}{836--845}.
\newblock
\urldef\tempurl%
\url{https://doi.org/10.1109/ICDE.2007.367929}
\showDOI{\tempurl}


\bibitem[\protect\citeauthoryear{Fletcher, Peters, and Poulovassilis}{Fletcher
  et~al\mbox{.}}{2016}]%
        {fletcher@edbt2016}
\bibfield{author}{\bibinfo{person}{George H.~L. Fletcher},
  \bibinfo{person}{Jeroen Peters}, {and} \bibinfo{person}{Alexandra
  Poulovassilis}.} \bibinfo{year}{2016}\natexlab{}.
\newblock \showarticletitle{Efficient regular path query evaluation using path
  indexes}. In \bibinfo{booktitle}{\emph{Proceedings of the 19th International
  Conference on Extending Database Technology, {EDBT} 2016, Bordeaux, France,
  March 15-16, 2016, Bordeaux, France, March 15-16, 2016}},
  \bibfield{editor}{\bibinfo{person}{Evaggelia Pitoura},
  \bibinfo{person}{Sofian Maabout}, \bibinfo{person}{Georgia Koutrika},
  \bibinfo{person}{Am{\'{e}}lie Marian}, \bibinfo{person}{Letizia Tanca},
  \bibinfo{person}{Ioana Manolescu}, {and} \bibinfo{person}{Kostas Stefanidis}}
  (Eds.). \bibinfo{publisher}{OpenProceedings.org}, \bibinfo{pages}{636--639}.
\newblock
\urldef\tempurl%
\url{https://doi.org/10.5441/002/edbt.2016.67}
\showDOI{\tempurl}


\bibitem[\protect\citeauthoryear{Foundation}{Foundation}{2022}]%
        {graphql}
\bibfield{author}{\bibinfo{person}{The~GraphQL Foundation}.}
  \bibinfo{year}{2022}\natexlab{}.
\newblock \showarticletitle{{GraphQL}}.
\newblock  (\bibinfo{year}{2022}).
\newblock
\urldef\tempurl%
\url{https://graphql.org/}
\showURL{%
\tempurl}


\bibitem[\protect\citeauthoryear{Gubichev, Bedathur, and Seufert}{Gubichev
  et~al\mbox{.}}{2013}]%
        {gubichev@grades2013}
\bibfield{author}{\bibinfo{person}{Andrey Gubichev},
  \bibinfo{person}{Srikanta~J. Bedathur}, {and} \bibinfo{person}{Stephan
  Seufert}.} \bibinfo{year}{2013}\natexlab{}.
\newblock \showarticletitle{Sparqling kleene: fast property paths in {RDF-3X}}.
  In \bibinfo{booktitle}{\emph{First International Workshop on Graph Data
  Management Experiences and Systems, {GRADES} 2013, co-located with
  {SIGMOD/PODS} 2013, New York, NY, USA, June 24, 2013}}. \bibinfo{pages}{14}.
\newblock
\urldef\tempurl%
\url{https://doi.org/10.1145/2484425.2484443}
\showDOI{\tempurl}


\bibitem[\protect\citeauthoryear{Gubichev and Neumann}{Gubichev and
  Neumann}{2011}]%
        {gubichev@webdb2011}
\bibfield{author}{\bibinfo{person}{Andrey Gubichev} {and}
  \bibinfo{person}{Thomas Neumann}.} \bibinfo{year}{2011}\natexlab{}.
\newblock \showarticletitle{Path Query Processing on Very Large {RDF} Graphs}.
  In \bibinfo{booktitle}{\emph{Proceedings of the 14th International Workshop
  on the Web and Databases 2011, WebDB 2011, Athens, Greece, June 12, 2011}}.
\newblock
\urldef\tempurl%
\url{http://webdb2011.rutgers.edu/papers/Paper21/pathwebdb.pdf}
\showURL{%
\tempurl}


\bibitem[\protect\citeauthoryear{Guo, Shao, Botev, and Shanmugasundaram}{Guo
  et~al\mbox{.}}{2003}]%
        {xrank}
\bibfield{author}{\bibinfo{person}{Lin Guo}, \bibinfo{person}{Feng Shao},
  \bibinfo{person}{Chavdar Botev}, {and} \bibinfo{person}{Jayavel
  Shanmugasundaram}.} \bibinfo{year}{2003}\natexlab{}.
\newblock \showarticletitle{{XRANK:} Ranked Keyword Search over {XML}
  Documents}. In \bibinfo{booktitle}{\emph{Proceedings of the 2003 {ACM}
  {SIGMOD} International Conference on Management of Data, San Diego,
  California, USA, June 9-12, 2003}}. \bibinfo{pages}{16--27}.
\newblock
\urldef\tempurl%
\url{https://doi.org/10.1145/872757.872762}
\showDOI{\tempurl}


\bibitem[\protect\citeauthoryear{He, Wang, Yang, and Yu}{He
  et~al\mbox{.}}{2007}]%
        {blinks}
\bibfield{author}{\bibinfo{person}{Hao He}, \bibinfo{person}{Haixun Wang},
  \bibinfo{person}{Jun Yang}, {and} \bibinfo{person}{Philip~S. Yu}.}
  \bibinfo{year}{2007}\natexlab{}.
\newblock \showarticletitle{{BLINKS:} ranked keyword searches on graphs}. In
  \bibinfo{booktitle}{\emph{Proceedings of the {ACM} {SIGMOD} International
  Conference on Management of Data, Beijing, China, June 12-14, 2007}}.
  \bibinfo{pages}{305--316}.
\newblock
\urldef\tempurl%
\url{https://doi.org/10.1145/1247480.1247516}
\showDOI{\tempurl}


\bibitem[\protect\citeauthoryear{Hristidis, Gravano, and
  Papakonstantinou}{Hristidis et~al\mbox{.}}{2003a}]%
        {efficient@vldb2003}
\bibfield{author}{\bibinfo{person}{Vagelis Hristidis}, \bibinfo{person}{Luis
  Gravano}, {and} \bibinfo{person}{Yannis Papakonstantinou}.}
  \bibinfo{year}{2003}\natexlab{a}.
\newblock \showarticletitle{Efficient IR-Style Keyword Search over Relational
  Databases}. In \bibinfo{booktitle}{\emph{Proceedings of 29th International
  Conference on Very Large Data Bases, {VLDB} 2003, Berlin, Germany, September
  9-12, 2003}}. \bibinfo{pages}{850--861}.
\newblock
\urldef\tempurl%
\url{https://doi.org/10.1016/B978-012722442-8/50080-X}
\showDOI{\tempurl}


\bibitem[\protect\citeauthoryear{Hristidis and Papakonstantinou}{Hristidis and
  Papakonstantinou}{2002}]%
        {discover}
\bibfield{author}{\bibinfo{person}{Vagelis Hristidis} {and}
  \bibinfo{person}{Yannis Papakonstantinou}.} \bibinfo{year}{2002}\natexlab{}.
\newblock \showarticletitle{{DISCOVER:} Keyword Search in Relational
  Databases}. In \bibinfo{booktitle}{\emph{{VLDB}}}.
\newblock
\urldef\tempurl%
\url{http://www.vldb.org/conf/2002/S19P02.pdf}
\showURL{%
\tempurl}


\bibitem[\protect\citeauthoryear{Hristidis, Papakonstantinou, and
  Balmin}{Hristidis et~al\mbox{.}}{2003b}]%
        {hristidis@icde2003}
\bibfield{author}{\bibinfo{person}{Vagelis Hristidis}, \bibinfo{person}{Yannis
  Papakonstantinou}, {and} \bibinfo{person}{Andrey Balmin}.}
  \bibinfo{year}{2003}\natexlab{b}.
\newblock \showarticletitle{Keyword Proximity Search on {XML} Graphs}. In
  \bibinfo{booktitle}{\emph{Proceedings of the 19th International Conference on
  Data Engineering, March 5-8, 2003, Bangalore, India}}.
  \bibinfo{pages}{367--378}.
\newblock
\urldef\tempurl%
\url{https://doi.org/10.1109/ICDE.2003.1260806}
\showDOI{\tempurl}


\bibitem[\protect\citeauthoryear{Kacholia, Pandit, Chakrabarti, Sudarshan,
  Desai, and Karambelkar}{Kacholia et~al\mbox{.}}{2005}]%
        {banks-2}
\bibfield{author}{\bibinfo{person}{Varun Kacholia}, \bibinfo{person}{Shashank
  Pandit}, \bibinfo{person}{Soumen Chakrabarti}, \bibinfo{person}{S.
  Sudarshan}, \bibinfo{person}{Rushi Desai}, {and} \bibinfo{person}{Hrishikesh
  Karambelkar}.} \bibinfo{year}{2005}\natexlab{}.
\newblock \showarticletitle{Bidirectional Expansion For Keyword Search on Graph
  Databases}. In \bibinfo{booktitle}{\emph{Proceedings of the 31st
  International Conference on Very Large Data Bases, Trondheim, Norway, August
  30 - September 2, 2005}}. \bibinfo{pages}{505--516}.
\newblock
\urldef\tempurl%
\url{http://www.vldb.org/archives/website/2005/program/paper/wed/p505-kacholia.pdf}
\showURL{%
\tempurl}


\bibitem[\protect\citeauthoryear{Kasneci, Ramanath, Sozio, Suchanek, and
  Weikum}{Kasneci et~al\mbox{.}}{2009}]%
        {star}
\bibfield{author}{\bibinfo{person}{Gjergji Kasneci}, \bibinfo{person}{Maya
  Ramanath}, \bibinfo{person}{Mauro Sozio}, \bibinfo{person}{Fabian~M.
  Suchanek}, {and} \bibinfo{person}{Gerhard Weikum}.}
  \bibinfo{year}{2009}\natexlab{}.
\newblock \showarticletitle{{STAR:} Steiner-Tree Approximation in Relationship
  Graphs}. In \bibinfo{booktitle}{\emph{Proceedings of the 25th International
  Conference on Data Engineering, {ICDE} 2009, March 29 2009 - April 2 2009,
  Shanghai, China}}. \bibinfo{pages}{868--879}.
\newblock
\urldef\tempurl%
\url{https://doi.org/10.1109/ICDE.2009.64}
\showDOI{\tempurl}


\bibitem[\protect\citeauthoryear{Kuijpers, Fletcher, Lindaaker, and
  Yakovets}{Kuijpers et~al\mbox{.}}{2021}]%
        {kuijpers@edbt2021}
\bibfield{author}{\bibinfo{person}{Jochem Kuijpers}, \bibinfo{person}{George
  Fletcher}, \bibinfo{person}{Tobias Lindaaker}, {and} \bibinfo{person}{Nikolay
  Yakovets}.} \bibinfo{year}{2021}\natexlab{}.
\newblock \showarticletitle{Path Indexing in the Cypher Query Pipeline}. In
  \bibinfo{booktitle}{\emph{Proceedings of the 24th International Conference on
  Extending Database Technology, {EDBT} 2021, Nicosia, Cyprus, March 23 - 26,
  2021}}. \bibinfo{pages}{582--587}.
\newblock
\urldef\tempurl%
\url{https://doi.org/10.5441/002/edbt.2021.68}
\showDOI{\tempurl}


\bibitem[\protect\citeauthoryear{Le, Li, Kementsietsidis, and Duan}{Le
  et~al\mbox{.}}{2014}]%
        {scalable-ks}
\bibfield{author}{\bibinfo{person}{Wangchao Le}, \bibinfo{person}{Feifei Li},
  \bibinfo{person}{Anastasios Kementsietsidis}, {and} \bibinfo{person}{Songyun
  Duan}.} \bibinfo{year}{2014}\natexlab{}.
\newblock \showarticletitle{Scalable Keyword Search on Large {RDF} Data}.
\newblock \bibinfo{journal}{\emph{{IEEE} Trans. Knowl. Data Eng.}}
  \bibinfo{volume}{26}, \bibinfo{number}{11} (\bibinfo{year}{2014}),
  \bibinfo{pages}{2774--2788}.
\newblock
\urldef\tempurl%
\url{https://doi.org/10.1109/TKDE.2014.2302294}
\showDOI{\tempurl}


\bibitem[\protect\citeauthoryear{Li, Ooi, Feng, Wang, and Zhou}{Li
  et~al\mbox{.}}{2008}]%
        {ease}
\bibfield{author}{\bibinfo{person}{Guoliang Li}, \bibinfo{person}{Beng~Chin
  Ooi}, \bibinfo{person}{Jianhua Feng}, \bibinfo{person}{Jianyong Wang}, {and}
  \bibinfo{person}{Lizhu Zhou}.} \bibinfo{year}{2008}\natexlab{}.
\newblock \showarticletitle{{EASE:} an effective 3-in-1 keyword search method
  for unstructured, semi-structured and structured data}. In
  \bibinfo{booktitle}{\emph{Proceedings of the {ACM} {SIGMOD} International
  Conference on Management of Data, {SIGMOD} 2008, Vancouver, BC, Canada, June
  10-12, 2008}}. \bibinfo{pages}{903--914}.
\newblock
\urldef\tempurl%
\url{https://doi.org/10.1145/1376616.1376706}
\showDOI{\tempurl}


\bibitem[\protect\citeauthoryear{Li, Qin, Yu, and Mao}{Li
  et~al\mbox{.}}{2016}]%
        {li@sigmod2016}
\bibfield{author}{\bibinfo{person}{Rong{-}Hua Li}, \bibinfo{person}{Lu Qin},
  \bibinfo{person}{Jeffrey~Xu Yu}, {and} \bibinfo{person}{Rui Mao}.}
  \bibinfo{year}{2016}\natexlab{}.
\newblock \showarticletitle{Efficient and Progressive Group Steiner Tree
  Search}. In \bibinfo{booktitle}{\emph{Proceedings of the 2016 International
  Conference on Management of Data, {SIGMOD} Conference 2016, San Francisco,
  CA, USA, June 26 - July 01, 2016}}. \bibinfo{pages}{91--106}.
\newblock
\urldef\tempurl%
\url{https://doi.org/10.1145/2882903.2915217}
\showDOI{\tempurl}


\bibitem[\protect\citeauthoryear{Luo, Lin, Wang, and Zhou}{Luo
  et~al\mbox{.}}{2007}]%
        {spark}
\bibfield{author}{\bibinfo{person}{Yi Luo}, \bibinfo{person}{Xuemin Lin},
  \bibinfo{person}{Wei Wang}, {and} \bibinfo{person}{Xiaofang Zhou}.}
  \bibinfo{year}{2007}\natexlab{}.
\newblock \showarticletitle{Spark: top-k keyword query in relational
  databases}. In \bibinfo{booktitle}{\emph{Proceedings of the {ACM} {SIGMOD}
  International Conference on Management of Data, Beijing, China, June 12-14,
  2007}}. \bibinfo{pages}{115--126}.
\newblock
\urldef\tempurl%
\url{https://doi.org/10.1145/1247480.1247495}
\showDOI{\tempurl}


\bibitem[\protect\citeauthoryear{Luo, Wang, Lin, Zhou, Wang, and Li}{Luo
  et~al\mbox{.}}{2011}]%
        {spark2}
\bibfield{author}{\bibinfo{person}{Yi Luo}, \bibinfo{person}{Wei Wang},
  \bibinfo{person}{Xuemin Lin}, \bibinfo{person}{Xiaofang Zhou},
  \bibinfo{person}{Jianmin Wang}, {and} \bibinfo{person}{Keqiu Li}.}
  \bibinfo{year}{2011}\natexlab{}.
\newblock \showarticletitle{{SPARK2:} Top-k Keyword Query in Relational
  Databases}.
\newblock \bibinfo{journal}{\emph{{IEEE} Trans. Knowl. Data Eng.}}
  \bibinfo{volume}{23}, \bibinfo{number}{12} (\bibinfo{year}{2011}),
  \bibinfo{pages}{1763--1780}.
\newblock
\urldef\tempurl%
\url{https://doi.org/10.1109/TKDE.2011.60}
\showDOI{\tempurl}


\bibitem[\protect\citeauthoryear{Na, Yi, Whang, Moon, and Hyun}{Na
  et~al\mbox{.}}{2022}]%
        {na@icde2022}
\bibfield{author}{\bibinfo{person}{Inju Na}, \bibinfo{person}{Ilyeop Yi},
  \bibinfo{person}{Kyu{-}Young Whang}, \bibinfo{person}{Yang{-}Sae Moon}, {and}
  \bibinfo{person}{Soon~J. Hyun}.} \bibinfo{year}{2022}\natexlab{}.
\newblock \showarticletitle{Regular Path Query Evaluation Sharing a Reduced
  Transitive Closure Based on Graph Reduction}.
\newblock  (\bibinfo{year}{2022}).
\newblock


\bibitem[\protect\citeauthoryear{Neo4j}{Neo4j}{2022}]%
        {neo4j-cypher}
\bibfield{author}{\bibinfo{person}{Inc. Neo4j}.}
  \bibinfo{year}{2022}\natexlab{}.
\newblock \showarticletitle{{Cypher Query Language}}.
\newblock  (\bibinfo{year}{2022}).
\newblock
\urldef\tempurl%
\url{https://neo4j.com/developer/cypher/}
\showURL{%
\tempurl}


\bibitem[\protect\citeauthoryear{Peng, Lin, Zhang, Zhang, and Qin}{Peng
  et~al\mbox{.}}{2022}]%
        {peng@vldb2022}
\bibfield{author}{\bibinfo{person}{You Peng}, \bibinfo{person}{Xuemin Lin},
  \bibinfo{person}{Ying Zhang}, \bibinfo{person}{Wenjie Zhang}, {and}
  \bibinfo{person}{Lu Qin}.} \bibinfo{year}{2022}\natexlab{}.
\newblock \showarticletitle{Answering reachability and K-reach queries on large
  graphs with label constraints}.
\newblock \bibinfo{journal}{\emph{{VLDB} J.}} \bibinfo{volume}{31},
  \bibinfo{number}{1} (\bibinfo{year}{2022}), \bibinfo{pages}{101--127}.
\newblock
\urldef\tempurl%
\url{https://doi.org/10.1007/s00778-021-00695-0}
\showDOI{\tempurl}


\bibitem[\protect\citeauthoryear{S. and Haritsa}{S. and Haritsa}{2019}]%
        {rootrank@comad2019}
\bibfield{author}{\bibinfo{person}{Vinay~M. S.} {and}
  \bibinfo{person}{Jayant~R. Haritsa}.} \bibinfo{year}{2019}\natexlab{}.
\newblock \showarticletitle{Root Rank: {A} Relational Operator for {KWS} Result
  Ranking}. In \bibinfo{booktitle}{\emph{Proceedings of the {ACM} India Joint
  International Conference on Data Science and Management of Data, {COMAD/CODS}
  2019, Kolkata, India, January 3-5, 2019}}. \bibinfo{pages}{103--111}.
\newblock
\urldef\tempurl%
\url{https://doi.org/10.1145/3297001.3297014}
\showDOI{\tempurl}


\bibitem[\protect\citeauthoryear{S. and Haritsa}{S. and Haritsa}{2020}]%
        {kws@infsys2020}
\bibfield{author}{\bibinfo{person}{Vinay~M. S.} {and}
  \bibinfo{person}{Jayant~R. Haritsa}.} \bibinfo{year}{2020}\natexlab{}.
\newblock \showarticletitle{Operator implementation of Result Set Dependent
  {KWS} scoring functions}.
\newblock \bibinfo{journal}{\emph{Inf. Syst.}}  \bibinfo{volume}{89}
  (\bibinfo{year}{2020}), \bibinfo{pages}{101465}.
\newblock
\urldef\tempurl%
\url{https://doi.org/10.1016/j.is.2019.101465}
\showDOI{\tempurl}


\bibitem[\protect\citeauthoryear{Shi, Cheng, Tran, Kharlamov, and Shen}{Shi
  et~al\mbox{.}}{2021}]%
        {qgstp@www21}
\bibfield{author}{\bibinfo{person}{Yuxuan Shi}, \bibinfo{person}{Gong Cheng},
  \bibinfo{person}{Trung{-}Kien Tran}, \bibinfo{person}{Evgeny Kharlamov},
  {and} \bibinfo{person}{Yulin Shen}.} \bibinfo{year}{2021}\natexlab{}.
\newblock \showarticletitle{Efficient Computation of Semantically Cohesive
  Subgraphs for Keyword-Based Knowledge Graph Exploration}. In
  \bibinfo{booktitle}{\emph{{WWW} '21: The Web Conference 2021, Virtual Event /
  Ljubljana, Slovenia, April 19-23, 2021}},
  \bibfield{editor}{\bibinfo{person}{Jure Leskovec}, \bibinfo{person}{Marko
  Grobelnik}, \bibinfo{person}{Marc Najork}, \bibinfo{person}{Jie Tang}, {and}
  \bibinfo{person}{Leila Zia}} (Eds.). \bibinfo{publisher}{{ACM} / {IW3C2}},
  \bibinfo{pages}{1410--1421}.
\newblock
\urldef\tempurl%
\url{https://doi.org/10.1145/3442381.3449900}
\showDOI{\tempurl}
\newblock
\shownote{Code available at: https://github.com/nju-websoft/QGSTP.}


\bibitem[\protect\citeauthoryear{Sun, Xiao, Cui, Halgamuge, Lappas, and
  Luo}{Sun et~al\mbox{.}}{2021}]%
        {lancet@vldb2021}
\bibfield{author}{\bibinfo{person}{Yahui Sun}, \bibinfo{person}{Xiaokui Xiao},
  \bibinfo{person}{Bin Cui}, \bibinfo{person}{Saman~K. Halgamuge},
  \bibinfo{person}{Theodoros Lappas}, {and} \bibinfo{person}{Jun Luo}.}
  \bibinfo{year}{2021}\natexlab{}.
\newblock \showarticletitle{Finding Group Steiner Trees in Graphs with both
  Vertex and Edge Weights}.
\newblock \bibinfo{journal}{\emph{Proc. {VLDB} Endow.}} \bibinfo{volume}{14},
  \bibinfo{number}{7} (\bibinfo{year}{2021}), \bibinfo{pages}{1137--1149}.
\newblock
\urldef\tempurl%
\url{https://doi.org/10.14778/3450980.3450982}
\showDOI{\tempurl}


\bibitem[\protect\citeauthoryear{Tran, Wang, Rudolph, and Cimiano}{Tran
  et~al\mbox{.}}{2009}]%
        {tran@icde2009}
\bibfield{author}{\bibinfo{person}{Thanh Tran}, \bibinfo{person}{Haofen Wang},
  \bibinfo{person}{Sebastian Rudolph}, {and} \bibinfo{person}{Philipp
  Cimiano}.} \bibinfo{year}{2009}\natexlab{}.
\newblock \showarticletitle{Top-k Exploration of Query Candidates for Efficient
  Keyword Search on Graph-Shaped {(RDF)} Data}. In
  \bibinfo{booktitle}{\emph{Proceedings of the 25th International Conference on
  Data Engineering, {ICDE} 2009, March 29 2009 - April 2 2009, Shanghai,
  China}}. \bibinfo{pages}{405--416}.
\newblock
\urldef\tempurl%
\url{https://doi.org/10.1109/ICDE.2009.119}
\showDOI{\tempurl}


\bibitem[\protect\citeauthoryear{Valstar, Fletcher, and Yoshida}{Valstar
  et~al\mbox{.}}{2017}]%
        {valstar@sigmod2017}
\bibfield{author}{\bibinfo{person}{Lucien D.~J. Valstar},
  \bibinfo{person}{George H.~L. Fletcher}, {and} \bibinfo{person}{Yuichi
  Yoshida}.} \bibinfo{year}{2017}\natexlab{}.
\newblock \showarticletitle{Landmark Indexing for Evaluation of
  Label-Constrained Reachability Queries}. In
  \bibinfo{booktitle}{\emph{Proceedings of the 2017 {ACM} International
  Conference on Management of Data, {SIGMOD} Conference 2017, Chicago, IL, USA,
  May 14-19, 2017}}. \bibinfo{pages}{345--358}.
\newblock
\urldef\tempurl%
\url{https://doi.org/10.1145/3035918.3035955}
\showDOI{\tempurl}


\bibitem[\protect\citeauthoryear{Wadhwa, Prasad, Ranu, Bagchi, and
  Bedathur}{Wadhwa et~al\mbox{.}}{2019}]%
        {wadhwa@sigmod2019}
\bibfield{author}{\bibinfo{person}{Sarisht Wadhwa}, \bibinfo{person}{Anagh
  Prasad}, \bibinfo{person}{Sayan Ranu}, \bibinfo{person}{Amitabha Bagchi},
  {and} \bibinfo{person}{Srikanta Bedathur}.} \bibinfo{year}{2019}\natexlab{}.
\newblock \showarticletitle{Efficiently Answering Regular Simple Path Queries
  on Large Labeled Networks}. In \bibinfo{booktitle}{\emph{Proceedings of the
  2019 International Conference on Management of Data, {SIGMOD} Conference
  2019, Amsterdam, The Netherlands, June 30 - July 5, 2019}}.
  \bibinfo{pages}{1463--1480}.
\newblock
\urldef\tempurl%
\url{https://doi.org/10.1145/3299869.3319882}
\showDOI{\tempurl}


\bibitem[\protect\citeauthoryear{Wang and Aggarwal}{Wang and Aggarwal}{2010}]%
        {kssurvey}
\bibfield{author}{\bibinfo{person}{Haixun Wang} {and} \bibinfo{person}{Charu~C.
  Aggarwal}.} \bibinfo{year}{2010}\natexlab{}.
\newblock \showarticletitle{A Survey of Algorithms for Keyword Search on Graph
  Data}.
\newblock In \bibinfo{booktitle}{\emph{Managing and Mining Graph Data}},
  \bibfield{editor}{\bibinfo{person}{Charu~C. Aggarwal} {and}
  \bibinfo{person}{Haixun Wang}} (Eds.). \bibinfo{series}{Advances in Database
  Systems}, Vol.~\bibinfo{volume}{40}. \bibinfo{publisher}{Springer},
  \bibinfo{pages}{249--273}.
\newblock
\urldef\tempurl%
\url{https://doi.org/10.1007/978-1-4419-6045-0\_8}
\showDOI{\tempurl}


\bibitem[\protect\citeauthoryear{Yakovets, Godfrey, and Gryz}{Yakovets
  et~al\mbox{.}}{2016}]%
        {yakovets@sigmod2016}
\bibfield{author}{\bibinfo{person}{Nikolay Yakovets}, \bibinfo{person}{Parke
  Godfrey}, {and} \bibinfo{person}{Jarek Gryz}.}
  \bibinfo{year}{2016}\natexlab{}.
\newblock \showarticletitle{Query Planning for Evaluating {SPARQL} Property
  Paths}. In \bibinfo{booktitle}{\emph{Proceedings of the 2016 International
  Conference on Management of Data, {SIGMOD} Conference 2016, San Francisco,
  CA, USA, June 26 - July 01, 2016}}. \bibinfo{pages}{1875--1889}.
\newblock
\urldef\tempurl%
\url{https://doi.org/10.1145/2882903.2882944}
\showDOI{\tempurl}


\bibitem[\protect\citeauthoryear{Yang, Agrawal, Jagadish, Tung, and Wu}{Yang
  et~al\mbox{.}}{2019}]%
        {parallel-ks}
\bibfield{author}{\bibinfo{person}{Yueji Yang}, \bibinfo{person}{Divyakant
  Agrawal}, \bibinfo{person}{H.~V. Jagadish}, \bibinfo{person}{Anthony K.~H.
  Tung}, {and} \bibinfo{person}{Shuang Wu}.} \bibinfo{year}{2019}\natexlab{}.
\newblock \showarticletitle{An Efficient Parallel Keyword Search Engine on
  Knowledge Graphs}. In \bibinfo{booktitle}{\emph{35th {IEEE} International
  Conference on Data Engineering, {ICDE} 2019, Macao, China, April 8-11,
  2019}}. \bibinfo{pages}{338--349}.
\newblock
\urldef\tempurl%
\url{https://doi.org/10.1109/ICDE.2019.00038}
\showDOI{\tempurl}


\end{thebibliography}
